\documentclass[english,dvipsnames]{article}
\usepackage{lmodern}

\usepackage[T1]{fontenc}
\usepackage[latin9]{inputenc}
\usepackage{geometry}
\usepackage{color}
\usepackage{babel}
\usepackage{float}
\usepackage{url}
\usepackage{dsfont}
\usepackage{amsmath}
\usepackage{amsthm}
\usepackage{amssymb}
\usepackage{stmaryrd}
\usepackage{graphicx}
\usepackage{setspace}
\usepackage[authoryear]{natbib}
\usepackage[unicode=true,
 bookmarks=true,bookmarksnumbered=false,bookmarksopen=false,
 breaklinks=false,pdfborder={0 0 1},backref=page,colorlinks=true]
 {hyperref}
\hypersetup{pdftitle={Unbiased Poisson},
 linkcolor=blue,urlcolor=blue,citecolor=blue}

\makeatletter

\floatstyle{ruled}
\newfloat{algorithm}{tbp}{loa}
\providecommand{\algorithmname}{Algorithm}
\floatname{algorithm}{\protect\algorithmname}

\numberwithin{equation}{section}
\numberwithin{figure}{section}
\numberwithin{table}{section}
\theoremstyle{plain}
\newtheorem{thm}{\protect\theoremname}
\theoremstyle{plain}
\newtheorem{assumption}[thm]{\protect\assumptionname}
\theoremstyle{plain}
\newtheorem{prop}[thm]{\protect\propositionname}
\theoremstyle{definition}
\newtheorem{defn}[thm]{\protect\definitionname}
\theoremstyle{remark}
\newtheorem{rem}[thm]{\protect\remarkname}
\theoremstyle{plain}
\newtheorem{lem}[thm]{\protect\lemmaname}
\theoremstyle{plain}
\newtheorem{cor}[thm]{\protect\corollaryname}
\theoremstyle{definition}
\newtheorem{example}[thm]{\protect\examplename}

\usepackage{color}
\usepackage{babel}
\usepackage{units}
\usepackage{caption}
\usepackage{subcaption}
\usepackage{algorithm}

\numberwithin{algorithm}{section}
\theoremstyle{plain}

\usepackage{dsfont}
\usepackage{anyfontsize}
\usepackage{fontawesome}

\newcommand{\Uniform}{\text{Uniform}}
\newcommand{\dif}{\bar \Delta}
\newcommand{\lr}[1]{\left(#1\right)}
\newcommand{\lrb}[1]{\left[#1\right]}
\newcommand{\lrcb}[1]{\left\{#1\right\}}

\makeatother

\addto\captionsbritish{\renewcommand{\algorithmname}{Algorithm}}
\addto\captionsbritish{\renewcommand{\assumptionname}{Assumption}}
\addto\captionsbritish{\renewcommand{\corollaryname}{Corollary}}
\addto\captionsbritish{\renewcommand{\definitionname}{Definition}}
\addto\captionsbritish{\renewcommand{\lemmaname}{Lemma}}
\addto\captionsbritish{\renewcommand{\propositionname}{Proposition}}
\addto\captionsbritish{\renewcommand{\remarkname}{Remark}}
\addto\captionsbritish{\renewcommand{\theoremname}{Theorem}}
\addto\captionsenglish{\renewcommand{\assumptionname}{Assumption}}
\addto\captionsenglish{\renewcommand{\corollaryname}{Corollary}}
\addto\captionsenglish{\renewcommand{\definitionname}{Definition}}
\addto\captionsenglish{\renewcommand{\lemmaname}{Lemma}}
\addto\captionsenglish{\renewcommand{\propositionname}{Proposition}}
\addto\captionsenglish{\renewcommand{\remarkname}{Remark}}
\addto\captionsenglish{\renewcommand{\theoremname}{Theorem}}
\addto\captionsenglish{\renewcommand{\examplename}{Example}}

\providecommand{\assumptionname}{Assumption}
\providecommand{\corollaryname}{Corollary}
\providecommand{\definitionname}{Definition}
\providecommand{\lemmaname}{Lemma}
\providecommand{\propositionname}{Proposition}
\providecommand{\remarkname}{Remark}
\providecommand{\theoremname}{Theorem}
\providecommand{\examplename}{Example}

\begin{document}
\title{Solving the Poisson equation using coupled Markov chains}
\author{Randal Douc\thanks{randal.douc@telecom-sudparis.eu}~ (SAMOVAR, T\'el\'ecom SudParis, Institut Polytechnique de Paris),\\
Pierre E. Jacob\thanks{pierre.jacob@essec.edu}~ (ESSEC Business
School),
 Anthony Lee\thanks{anthony.lee@bristol.ac.uk}~ (University of Bristol)\\
\& Dootika Vats\thanks{dootika@iitk.ac.in}~ (IIT Kanpur)}

\maketitle

\begin{abstract}
This article shows how coupled Markov chains that meet exactly after a random number of iterations
can be used to generate unbiased estimators of the solutions of the Poisson equation. 
Through this connection, we re-derive
known unbiased estimators of expectations with respect to the stationary distribution of a Markov chain
and provide conditions for the finiteness of their moments.
We further construct
unbiased estimators of the asymptotic variance of Markov chain
ergodic averages, and provide conditions for the finiteness of the estimators' moments
of any order. If their second moment is finite, the average of independent copies of such estimators converges to the asymptotic variance at the Monte Carlo rate, 
comparing favorably to known rates for batch means and spectral variance estimators.
The results are illustrated with numerical experiments.
\end{abstract}
\global\long\def\fishytestestimator{G}%
\global\long\def\fishytest{g}%
\global\long\def\test{h}%

\tableofcontents{}

\section{Introduction\label{sec:intro}}

\subsection{Central Limit Theorem and the Poisson equation\label{sec:intro:clt}}

Markov chain Monte Carlo (MCMC) methods form a convenient family of
simulation techniques with many applications in statistics. With $(\mathbb{X},\mathcal{X})$
a measurable space and $\pi$ a probability measure of statistical
interest, MCMC involves simulation of a discrete-time, time-homogeneous
Markov chain $X=(X_{t})_{t\geq0}$, with a $\pi$-invariant Markov
transition kernel $P$ and initial distribution $\pi_{0}$. Letting
$L^{p}(\pi)=\{f:\pi(\left|f\right|^{p})<\infty\}$, where $\pi(f)=\int f(x)\pi({\rm d}x)$,
the interest is to approximate $\pi(h)$ for some function $h\in L^{1}(\pi)$, termed the \emph{test function}.

In particular, after simulating the chain until time $t$, one may
approximate an integral of interest $\pi(h)$ via the \emph{ergodic average} $t^{-1}\sum_{s=0}^{t-1}h(X_{s})$.
Under weak assumptions, such averages converge almost surely to $\pi(h)$
as $t\to\infty$ \citep[see for example Theorem 17.0.1 in][]{meyn:tweedie:1993}, and under stronger but still realistic assumptions
on $\pi_{0}$, $P$ and $h$, ergodic averages
satisfy central limit theorems (CLTs),
\begin{equation}
\sqrt{t}\left(\frac{1}{t}\sum_{s=0}^{t-1}h(X_{s})-\pi(h)\right)\overset{d}{\to}\text{Normal}(0,v(P,h)),\qquad\text{as }t\to\infty,\label{eq:clt}
\end{equation}
where $v(P,h)$ is the asymptotic variance associated with the Markov
kernel $P$ and the function $h$ \citep[see for example Theorem 1 in][]{jones2004markov}. 
One standard route to proving a CLT is via a solution of the Poisson
equation for $h$ associated with $P$, i.e. any function $g$ such
that
\begin{equation}
(I-P)g=h-\pi(h)=:h_{0},\label{eq:poisson-equation}
\end{equation}
where $P$ is viewed as a Markov operator, i.e. $Pg:x\mapsto \int P(x,{\rm d}y)g(y)$ and $I$ is the identity.
Theorem~\ref{thm:clt-kappa} below \citep[or, for example, Theorem 21.2.5 in][]{douc2018MarkovChains} 
provides conditions for \eqref{eq:clt} to hold and for the asymptotic variance to be of the form
\begin{equation}
v(P,\test)=\mathbb{E}_{\pi}\left[\left\{ \fishytest(X_{1})-P\fishytest(X_{0})\right\} ^{2}\right]=2\pi(h_{0}\cdot g)-\pi(h_{0}^{2}),\label{eq:avar-intro}
\end{equation}
where $\mathbb{E}_{\pi}$ indicates that $X_{0}\sim\pi$. 
We focus on solutions $g=g_{\star}+c$ with $g_{\star}$ defined as follows.

\begin{defn}
\label{def:g-star} With $\test\in L^{1}(\pi)$ and $\test_{0}=\test-\pi(\test)$,
define the function
\begin{equation}
\fishytest_{\star}:=\sum_{t=0}^{\infty}P^{t}h_{0}.
\end{equation}
\end{defn}

If $g_{\star}$ is well-defined,
then it is straightforward to check that $(I-P)g_{\star}=h_{0},$
so that $g_{\star}$ is indeed a solution to the Poisson equation \eqref{eq:poisson-equation}. For brevity, we will call solutions to \eqref{eq:poisson-equation} \emph{fishy}
functions. Fishy functions are not unique, since $g_{\star}+c$
is also fishy for any constant $c\in\mathbb{R}$. Moreover, if $\pi$ is the unique invariant distribution of $P$ and if $g_{\star}\in L^{1}(\pi)$
(see Theorem \ref{thm:clt-kappa})
then $\pi(g_{\star})=0$, and by \citet[Lemma 21.2.2,][]{douc2018MarkovChains}
all fishy functions $g$ are equal to $g_{\star}$ up to an additive
constant.
Inserting the function $g_\star$ in place of $g$ in \eqref{eq:avar-intro} leads to the
more familiar expression
\begin{equation}\label{eq:avarfamiliar}
v(P,\test)={\rm var}_{\pi}(h(X_{0}))+2\sum_{t=1}^{\infty}{\rm cov}_{\pi}(h(X_{0}),h(X_{t})),
\end{equation}
where the subscript $\pi$ indicates that $X_{0}\sim\pi$. That expression 
is retrieved with direct calculations from developing ${\rm var}_{\pi}(t^{-1/2}\sum_{s=0}^{t-1}h(X_{s}))$
and taking the limit as $t\to\infty$.

\subsection{Couplings of Markov chains\label{sec:intro:coupling}}

Our contributions rely on couplings of Markov chains.
For any two distributions $\mu$ and $\nu$ on $(\mathbb{X},\mathcal{X})$, a coupling refers to
a pair of random variables $(U,V)$ on $(\mathbb{X}\times \mathbb{X},\mathcal{X}\otimes \mathcal{X})$ such that $U\sim \mu$ and $V\sim \nu$.
This extends to Markov chains, viewed as distributions on the path space $(\prod_{t=0}^\infty \mathbb{X}, \bigotimes_{t=0}^\infty\mathcal{X})$.
We consider pairs of chains that start from different distributions but evolve with the same transition $P$, 
and we will focus on Markovian couplings: for two Markov chains $(X_t)_{t\geq 0}$ and $(Y_t)_{t\geq 0}$
we will consider a joint process on $(\prod_{t=0}^\infty \mathbb{X}\times \mathbb{X}, \bigotimes_{t=0}^\infty\mathcal{X}\otimes\mathcal{X})$ that is Markov
with transition kernel $\bar{P}$, and such that its first and second coordinates are distributed
as $(X_t)_{t\geq 0}$ and $(Y_t)_{t\geq 0}$ respectively.

\subsection{Contributions\label{sec:intro:contrib}}

In Section \ref{sec:coupledchainsandfishyfunctions} we elicit links between couplings of Markov chains and the Poisson equation \eqref{eq:poisson-equation}. The connection leads to estimators of pointwise evaluations of the solutions. We establish the lack of bias of these estimators and the finiteness of their moments under conditions on the function $h$ and on the coupled transition $\bar{P}$.
In Section \ref{subsec:recoverGlynnRhee} we re-derive the unbiased estimators of $\pi(h)$ pioneered by \citet{glynn2014exact}, and the variants of \citet{joa2020}. We refer to these estimators of $\pi(h)$ as unbiased MCMC. Our derivation leads to simple conditions for the finiteness of their moments of any order (as opposed to second moments only in the aforementioned articles), which will be used in subsequent sections. We also relate the efficiency of unbiased MCMC with the efficiency of ergodic average MCMC, which is the inverse of the asymptotic variance $v(P,h)$.
In Section \ref{sec:Asymptotic-variance} we propose estimators of the asymptotic variance $v(P,h)$, combining the estimators of fishy functions from Section \ref{sec:coupledchainsandfishyfunctions} with unbiased MCMC as in Section \ref{subsec:recoverGlynnRhee}. The proposed estimators of $v(P,h)$ are built from independent runs of coupled chains of random lengths, instead of long runs. They converge to $v(P,h)$ at the Monte Carlo rate, which is faster than known rates of convergence for batch means and spectral variance estimators.  Estimators of $v(P,h)$ can be used to compare the 
efficiencies of unbiased and ergodic average MCMC, and to compare MCMC algorithms to one another without ever relying on long run asymptotics.
We experiment with the proposed estimators in 
Section \ref{sec:numerical:experiments} and Section \ref{sec:discussion} concludes.

\section{Coupled chains and fishy functions\label{sec:coupledchainsandfishyfunctions}}

\subsection{\label{subsec:Coupled-Markov-chains}Coupled Markov chains and meeting time}

For a given $\pi$-invariant Markov kernel $P$, 
the kernel $\bar{P}$ is a coupling of $P$ with
itself in that it satisfies
\begin{equation}
\bar{P}(x,y;A\times\mathbb{X})=P(x,A),\quad\bar{P}(x,y;\mathbb{X}\times A)=P(y,A),\qquad A\in\mathcal{X}.
\end{equation}
We consider a time-homogeneous,
discrete-time Markov chain $(X,Y)=(X_{t},Y_{t})_{t\geq0}$ with Markov
kernel $\bar{P}$, such that $X=(X_{t})_{t\geq0}$ and $Y=(Y_{t})_{t\geq0}$ are
both time-homogeneous, discrete-time Markov chains with kernel $P$.
We use subscripts to denote the distribution of $(X_{0},Y_{0})$.
For example $\mathbb{P}_{x,y}$ is the law of $(X,Y)$ when $(X_{0},Y_{0})=(x,y)$,
and $\mathbb{P}_{\bar{\nu}}$ indicates $(X_{0},Y_{0})\sim \bar{\nu}$, for some distribution $\bar{\nu}$ on the joint space. When only
one chain is referenced, e.g. $X$, we may similarly write $\mathbb{P}_{x}$
and $\mathbb{P}_{\mu}$ to indicate $X_{0}=x$ and $X_{0}\sim\mu$,
respectively. 

We require the coupling $\bar{P}$ to be \emph{successful}, in the terminology of \citet{pitman1976coupling}, which means that
the \emph{meeting time}, defined  as
\begin{equation}\label{eq:deftau}
\tau:=\inf\{t\geq0:X_{t}=Y_{t}\},
\end{equation}
is almost surely finite. Furthermore, we impose that $X_{t}=Y_{t}$ with probability $1$
for all $t\geq\tau$. \citet{johnson1998coupling,joa2020} and
others have shown how to construct such couplings for
some realistic MCMC algorithms;
Appendix~\ref{appx:couplingrh}
provides pointers.

We now introduce the main assumption in this manuscript: we require the meeting time $\tau$ to have $\kappa$ finite
moments, with $\kappa>1$, for two chains that would independently start from the stationary distribution $\pi$, i.e. $(X_{0},Y_{0})\sim\pi\otimes\pi$, the independent coupling of $\pi$ with itself.

\begin{assumption}
\label{assu:tau-moment-kappa}The Markov transition kernel $P$ is $\pi$-irreducible
and for some $\kappa>1$, $\mathbb{E}_{\pi\otimes\pi}[\tau^{\kappa}]<\infty.$
\end{assumption}

Interpretation and verification of this assumption are discussed in Section \ref{subsec:interpretassumption}.
The assumption that $P$ is $\pi$-irreducible is equivalent to assuming that it has an irreducibility measure,
since an invariant probability measure is a maximal irreducibility measure \citep[Theorem~9.2.15]{douc2018MarkovChains} and implies that
$\pi$ is necessarily the unique invariant probability measure for $P$ \citep[Corollary~9.2.16]{douc2018MarkovChains}.
Assumption~\ref{assu:tau-moment-kappa} is sufficient to guarantee that the chain is aperiodic (Appendix~\ref{subsec:ass-meeting-times}), 
to justify 
the existence of the function $g_\star$ in Definition \ref{def:g-star}, 
to justify that $g_\star$ is solution to \eqref{eq:poisson-equation} in $L_0^1(\pi)=\{f\in L^1(\pi):\pi(f)=0\}$, and 
to establish the CLT \eqref{eq:clt} for a class of test functions.

\begin{thm}
\label{thm:clt-kappa}Under Assumption~\ref{assu:tau-moment-kappa},
let $h\in L^{m}(\pi)$ for some $m>2\kappa/(\kappa-1)$. Then $g_{\star}\in L_{0}^{1}(\pi)$,
$h_{0}\cdot g_{\star}\in L^{1}(\pi)$ and the CLT \eqref{eq:clt}
holds for $\pi$-almost all $X_{0}$ with $v(P,h)=2\pi(h_{0}\cdot g_{\star})-\pi(h_{0}^{2})<\infty$.
\end{thm}

The proof in Appendix~\ref{subsec:clt-poisson}
relies heavily on the strategy of \citet[Section~21.4.1]{douc2018MarkovChains}
but features $g_{\star}$ from Definition~\ref{def:g-star}
more prominently, and uses the CLT condition from \citet{maxwell2000central}
rather than \citet[Theorem~21.4.1]{douc2018MarkovChains}. Note that
the CLT does not require $g_{\star}\in L_{0}^{2}(\pi)$.
As a complement to Theorem \ref{thm:clt-kappa}, Theorem~\ref{thm:poisson-Lp}
states that, if $m>\kappa/(\kappa - 1)$,  then $g_\star \in L_0^p(\pi)$ for $p\geq 1$ such that $p^{-1} > m^{-1} + \kappa^{-1}$.

\subsection{\label{subsec:estimatingsolutions}Unbiased approximation of fishy
functions}

Solutions of the Poisson equation have been studied extensively
\citep[see, e.g.,][]{glynn1996liapounov,glynn2022solving}.
However, the Poisson equation is  not analytically solvable for most Markov chains and
functions of interest, and consistent approximations have been lacking.
We consider a family of fishy functions that
are amenable to estimation using coupled chains.


\begin{defn}
\label{def:gy}For $y\in\mathbb{X}$ and $g_\star$ in Definition~\ref{def:g-star}, define the function
\begin{equation}
\fishytest_{y}:x\mapsto\fishytest_{\star}(x)-\fishytest_{\star}(y) = \sum_{t=0}^\infty P^th(x) - P^th(y).\label{eq:explicitfishy-1}
\end{equation}
\end{defn}
Since $g_{\star}(y)$ is
a constant with respect to $x$, $x\mapsto g_y(x)$ is fishy as long as $g_{\star}$ is well-defined. When $y$ is fixed, in the sequel we may write $g$ instead
of $g_{y}$.

\begin{defn}
\label{def:Gyx}For $x,y \in \mathbb{X}$, the 
proposed estimator of $g_{y}(x)$ is
\begin{equation}
G_{y}(x):=\sum_{t=0}^{\tau-1}h(X_{t})-h(Y_{t}),\label{eq:estimatorpoisson}
\end{equation}
where $(X,Y)$ is a Markov chain starting from $(X_0,Y_0)=(x,y)$, evolving according to $\bar{P}$ defined in Section~\ref{subsec:Coupled-Markov-chains},
and $\tau$ is defined in \eqref{eq:deftau}. The random variable $G_{y}(x)$ is thus a function of the coupled chains $(X,Y)$ from time zero to $\tau-1$. 
We will sometimes denote $G_{y}$ by $G$.
\end{defn}

The simple intuition behind Definition~\ref{def:Gyx} is that we
can equivalently write
\[
G_{y}(x)=\sum_{t=0}^{\infty}h(X_{t})-h(Y_{t}),
\]
since $h(X_t) = h(Y_t)$ for all $t\geq \tau$, and each term $h(X_t) - h(Y_t)$ has expectation
$P^th(x) - P^th(y)$ under $\mathbb{P}_{x,y}$. Hence, $\mathbb{E}[G_y(x)]$ is equal to $g_y(x)$
if we can justify the interchange of expectation and infinite sum.
The following result, established in Appendix~\ref{subsec:approx-fishy-function}
provides conditions for $G_y(x)$ to be unbiased and to have finite moments.
\begin{thm}
\label{thm:fishy-main-text}Under Assumption~\ref{assu:tau-moment-kappa},
let $h\in L^{m}(\pi)$ for some $m>\kappa/(\kappa-1)$. For $\pi\otimes\pi$-almost
all $(x,y)$, $\mathbb{E}\left[G_{y}(x)\right]=g_{\star}(x)-g_{\star}(y) = g_y(x)$
and for $p\geq1$ such that $\frac{1}{p}>\frac{1}{m}+\frac{1}{\kappa}$,
$\mathbb{E}\left[\left|G_{y}(x)\right|^{p}\right]<\infty$.
\end{thm}

For our subsequent asymptotic variance estimators in Section \ref{sec:Asymptotic-variance}, it is important that this result is pointwise.
The random variable $G_{y}(x)$ may be simulated using Algorithm~\ref{alg:coupledchains}
with $L=0$ (setting $L\geq 1$ will be useful from Section~\ref{subsec:recoverGlynnRhee} onwards).
The cost of sampling $\fishytestestimator_{y}(x)$ is the cost of
running a pair of chains until they meet, which is typically comparable
to twice the cost of running one chain for the same number of steps,
i.e. $2\tau$.

\begin{rem}
  We can generalize Definitions~\ref{def:gy}-\ref{def:Gyx} by introducing a distribution $\nu$ on $\mathbb{X}$
instead of a fixed $y\in\mathbb{X}$. We may sample
$Y_{0}\sim\nu$, and generate
$G_{Y_{0}}(x)$ given $Y_{0}$ as above. Under adequate assumptions,
this would have expectation $g_\nu(x)$ with $\fishytest_{\nu}:x\mapsto g_{\star}(x)-\nu(g_{\star})$.
To unbiasedly
approximate $g_{\star}(x)$ itself, we could estimate $g_{\nu}(x)$
and $\pi(g_{\nu})$ in an unbiased manner, for an arbitrary distribution
$\nu$, and take their difference.
\end{rem}

\begin{algorithm}[H]
Input: initial states $x$, $y$, Markov kernel $P$, coupled kernel
$\bar{P}$, lag $L\geq0$.
\begin{enumerate}
\item Set $X_{0}=x$, $Y_{0}=y$.
\item If $L\geq1$, for $t=1,\ldots,L$, sample $X_{t}$ from $P(X_{t-1},\cdot)$.
\item For $t\geq L$, sample $(X_{t+1},Y_{t-L+1})$ from $\bar{P}((X_{t},Y_{t-L}),\cdot)$
until $X_{t+1}=Y_{t-L+1}$.
\end{enumerate}
Output: coupled chains and meeting time $\tau=\inf\{t> L: X_t = Y_{t-L}\}$. \caption{Simulation of coupled lagged chains.\label{alg:coupledchains}}
\end{algorithm}


\subsection{Interpretation and verification of Assumption \texorpdfstring{\ref{assu:tau-moment-kappa}}{2}\label{subsec:interpretassumption}}

Since we use Assumption \ref{assu:tau-moment-kappa} extensively,
this section provides details on how to interpret and to verify it via drift conditions.
The assumption implies a polynomially decaying survival function
$\mathbb{P}_{x,y}(\tau>t)$ of order $\kappa$ and, conversely, one
may verify the assumption by showing that $\mathbb{P}_{x,y}(\tau>t)$
decays polynomially with order $s>\kappa$ with a dependence on $(x,y)$
that is not too strong. The proof of the next result is in Appendix~\ref{subsec:ass-meeting-times}.
\begin{prop}
\label{prop:tau-moment-survival}If Assumption~\ref{assu:tau-moment-kappa}
holds then
\[
\forall t \geq 0 \quad \mathbb{P}_{x,y}(\tau>t)\leq\mathbb{E}_{x,y}[\tau^{\kappa}](t+1)^{-\kappa},
\]
and $\mathbb{E}_{x,y}[\tau^{\kappa}]<\infty$ for $\pi\otimes\pi$-almost
all $(x,y)$. Conversely, if for some $s>\kappa$, there exists $\tilde{C}:\mathbb{X}\times\mathbb{X}\to\mathbb{R}$
with $\pi\otimes\pi(\tilde{C})<\infty$ such that for $\pi\otimes\pi$-almost
all $(x,y)$, we have 
\[
\forall t \geq 0 \quad  \mathbb{P}_{x,y}(\tau>t)\leq\tilde{C}(x,y)(t+1)^{-s},
\]
then $\mathbb{E}_{\pi\otimes\pi}\left[\tau^{\kappa}\right]<\infty$.
\end{prop}

We can adapt results from \citet{joa2020} to verify Assumption \ref{assu:tau-moment-kappa} in the case of geometrically ergodic Markov chains.
Variants of the following assumption are referred to as a \emph{geometric drift condition}. 
Assumption \ref{assumption:drift} is very similar to Condition $D_{\text{g}}(V,\lambda,b,C)$, Definition 14.1.5 in \citet{douc2018MarkovChains}.
Such conditions have been shown to hold for many 
combinations of MCMC algorithms and target distributions, see
for example
\citet{jarner2000geometric} on Metropolis--Rosenbluth--Teller--Hastings (MRTH) 
with random walk proposals,
\citet{roberts1996exponential} on Langevin algorithms,
and \citet{durmus2020hmc} on Hamiltonian Monte Carlo. 

\begin{assumption}\label{assumption:drift}
  The Markov kernel $P$ is $\pi$-invariant, $\pi$-irreducible, and there exists a measurable function 
  $V:\;\mathbb{X}\to [1,\infty)$, $\lambda\in (0,1)$, $b \in (0,\infty)$ and a small set $\mathcal{C}$ such that
  \[\forall x \in \mathbb{X} \quad PV(x) \leq \lambda V(x) + b\mathds{1}(x \in \mathcal{C}).\]
\end{assumption}

While Assumption \ref{assumption:drift} is on the transition $P$, 
the next assumption is on the coupling $\bar{P}$. 

\begin{assumption}\label{assumption:meetingprobafromC}
  There exist $\mathcal{C}$ and $\epsilon\in (0,1)$ such that, with $\mathcal{D} = \{(x,y)\in\mathbb{X}\times\mathbb{X}:x=y\}$, 
\begin{equation}\label{mino}
\inf_{(x,y)\in\mathcal{C}\times\mathcal{C}} \bar P((x,y), \mathcal{D}) \geq \epsilon.
\end{equation}
\end{assumption}

Assumption~\ref{assumption:meetingprobafromC} states that meeting occurs in one step with probability at least $\epsilon$ 
when both chains are simultaneously in $\mathcal{C}$ and evolve according to $\bar{P}$. It is for example satisfied
for couplings of Metropolis--Rosenbluth--Teller--Hastings (MRTH) with Normal proposals using the reflection-maximal coupling described in Appendix~\ref{appx:couplingrh}
for any bounded set
$\mathcal{C}$ on which the target density is upper bounded.

\begin{prop}\label{prop:vgeometric}
  Suppose that $P$ satisfies Assumption \ref{assumption:drift} with a small set $\mathcal{C}$
  of the form $\mathcal{C} = \{x: V(x)\leq \bar{v}\}$, for some $\bar{v}\geq 1$,
  for which Assumption \ref{assumption:meetingprobafromC} holds with some $\epsilon>0$,
  and assume that $\lambda + 2b/(1+\bar{v})<1$.
Then Assumption \ref{assu:tau-moment-kappa} holds for all $\kappa>1$.
\end{prop}

The above result is very similar to Proposition 4 of \citet{joa2020}, and its proof is omitted.
The only difference comes from the initialization from $\pi \otimes \pi$ rather than
$\pi_0 P \otimes \pi_0$, where $\pi_0$ is the initial distribution of the chain. The explicit assumption $\pi_0(V)<\infty$ in Proposition 4 of \citet{joa2020} becomes $\pi(V)<\infty$, but this always holds under Assumption \ref{assumption:drift},
as stated in Lemma 14.1.10 of \citet{douc2018MarkovChains}. 
As noted in \citet{joa2020}, if sub-level sets $\mathcal{C}(v) = \{x: V(x)\leq v\}$ are small for all sufficiently large $v$ and Assumption~\ref{assumption:drift} holds with $\mathcal{C}=\mathcal{C}(v)$,
then Assumption~\ref{assumption:drift} also holds for $\mathcal{C}=\mathcal{C}(\bar{v})$ with $\bar{v}\geq v$ and so the condition $\lambda + 2b/(1+\bar{v})<1$ can be satisfied by taking $\bar{v}$ large enough, while the quantity $\epsilon$ in Assumption~\ref{assumption:meetingprobafromC} typically decreases but remains positive.

Proposition 4 of \citet{joa2020} was used to show that the meeting time had geometric tails
for couplings of Hamiltonian Monte Carlo in \citet[][Theorem 1]{heng2019unbiased}, 
and for couplings of the Bouncy Particle Sampler in \citet[][Proposition 3.1]{corenflos2023debiasing}, under assumptions similar to strong log-concavity and smoothness of the target. It 
was also used in \citet[][Proposition 2]{biswas2021couplingbased} for couplings of a Gibbs sampler for high-dimensional regression with shrinkage priors. Thus Proposition~\ref{prop:vgeometric} can similarly be applied to these settings 
where geometric drift conditions have been shown to hold.

To cover other cases, the reasoning can be extended to polynomially ergodic chains,
for which Assumption \ref{assu:tau-moment-kappa} may hold only for some values of $\kappa$.
We consider the following assumption, which is a variant of a polynomial drift condition.
The assumption is very similar to Condition $D_{\text{sg}}(V,\phi,b,C)$, Definition 16.1.7 in \citet{douc2018MarkovChains}, but includes
the condition $\inf_{{\mathcal C}^\complement} V^{\alpha}>b/\vartheta$. As discussed above, this condition can be met if sub-level sets are small.

\begin{assumption}\label{assumption:drift:subgeom}
  The Markov kernel $P$ is $\pi$-invariant, $\pi$-irreducible, and there exists a measurable function 
     $V:\;\mathbb{X}\to [1,\infty)$, $\alpha\in (0,1)$, $\vartheta,b \in (0,\infty)$ and a small set $\mathcal{C}$ such that we have $\sup_{\mathcal C} V<\infty$ and $\inf_{{\mathcal C}^\complement} V^{\alpha}>b/\vartheta$ and
      \[ \forall x \in \mathbb{X} \quad PV(x) \leq V(x) -\vartheta V^{\alpha}(x)+ b\mathds{1}(x \in \mathcal{C}).\] 
\end{assumption}

Under Assumptions~\ref{assumption:meetingprobafromC}-\ref{assumption:drift:subgeom} we obtain the next result, established in Appendix~\ref{subsec:ass-meeting-times}.

\begin{prop}\label{prop:subgeometric}
  Suppose that $P$ satisfies Assumption \ref{assumption:meetingprobafromC} and Assumption \ref{assumption:drift:subgeom}
  with the same small set ${\mathcal C}$.
Then Assumption \ref{assu:tau-moment-kappa} holds for $\kappa< \alpha/ (1-\alpha)$.
\end{prop}

The above assumptions and proposition can be compared to Assumptions 4--5 and Theorem 2 in \citet{middleton2020unbiased}.
Our assumptions are similar but slightly weaker, namely we do not assume that $P$ is aperiodic since it is implied by the other conditions as stated in Proposition~\ref{prop:int-tv}. 
If we were to employ Theorem 2 in \citet{middleton2020unbiased} to directly verify Assumption~\ref{assu:tau-moment-kappa}, setting $\pi_0$ to $\pi$, we would have to assume that $\pi$ is supported on a compact set. This is the main motivation for our Proposition~\ref{prop:subgeometric}. 
On the other hand, our result establishes less than $\alpha/(1-\alpha)$ moments of $\tau$, whereas \citet[Theorem 2,][]{middleton2020unbiased} obtain
$1/(1-\alpha)$ moments, which is strictly larger.

Proposition~\ref{prop:subgeometric} is useful 
if a polynomial drift condition holds for the MCMC algorithm under consideration. 
\citet{andrieu2015convergence} obtain such drift conditions for pseudo-marginal methods, such
as random walk MRTH where target density evaluations are replaced by unbiased estimates,
see their Corollary 31, Theorem 38 and Theorem 45. Using these results, \citet[Section 2.3,][]{middleton2020unbiased}
proceed to verify the assumptions of their Theorem 2 in the context of pseudo-marginal MRTH,
under conditions on the noise of the target density estimators and the tails of $\pi$. Under these conditions stated
in their Proposition 2, our assumptions are also satisfied. 
In the case of MRTH with independent proposals, we provide details on the verification of the assumptions of Proposition~\ref{prop:subgeometric} in the following toy example.

\begin{example}
  Following \citet[Example~17.2.3]{douc2018MarkovChains}, we consider an independent MRTH kernel $P$ with target $\pi={\rm Uniform}(0,1)$ and proposal density $q(x)=(r+1)x^{r}$ on $(0,1)$, for some $r>0$.
  Since $\pi(x)/q(x)$ is unbounded, the resulting Markov chain cannot be geometrically ergodic \citep{mengersen1996rates} so Assumption~\ref{assumption:drift} cannot hold.
  However, Assumption \ref{assumption:drift:subgeom} does hold.
  Indeed, letting $s\in(\max(r,1),r+1)$ one can take $V(x)=x^{-s}$ and verify that
  \begin{align*}
    PV(x)	&= V(x)-(r+1)V(x)^{1-\frac{r}{s}}+bx^{r-s+1}-\frac{r+1}{s-1}x^{r}, \\
	        &\leq V(x)-(r+1)V(x)^{1-\frac{r}{s}}+bx^{r-s+1},
  \end{align*}
  where $b=\frac{r+1}{r-s+1}+\frac{r+1}{s-1}+r$.
  We define $\mathcal{C}=[x_{0},1)$ with $x_{0}<b^{\frac{1}{r-s}}$. One can verify that
  Assumption \ref{assumption:drift:subgeom} holds with $\vartheta=1$, $\alpha=1-r/s$ and the given $b$,
  noting that $\inf_{x\in\mathcal{C}^\complement}V(x)^{1-\frac{r}{s}}=V(x_{0})^{1-\frac{r}{s}}>b$.
  Assumption \ref{assumption:meetingprobafromC} is simple to verify for the coupled kernel obtained by using a common proposal and a common Uniform random variable in the acceptance step.
  We verify Assumption \ref{assu:tau-moment-kappa} with $\kappa$ arbitrarily close to $r^{-1}$ by taking $s$ arbitrarily close to $r+1$.
\end{example}

The tails of the meeting times can also be studied 
directly, without drift conditions, as is done in \citet{lee2020coupling,karjalainen2025mixingtimeconditionalbackward} for couplings of conditional particle filters,
in \citet{nguyen2022many} for couplings of a Gibbs sampler for Bayesian clustering with Dirichlet process mixtures, 
in \citet{deligiannidis2024importance} for the common random numbers coupling of MRTH with independent proposals, 
and in Section \ref{subsec:fishyillustration}
below.


\subsection{Illustration\label{subsec:fishyillustration}}

We illustrate the fishy function estimator of Definition~\ref{def:Gyx} and the verification of Assumption~\ref{assu:tau-moment-kappa} in the following example.
The target $\pi$ is defined as the posterior in a Cauchy location model
with parameter $\theta$.
We observe $n=3$ measurements $z_{1}=-8,z_{2}=+8,z_{3}=+17$, assumed
to be realizations of $\text{Cauchy}(\theta,1)$. The prior on $\theta$
is $\text{Normal}(0,100)$. Since the prior is Normal and the likelihood is upper-bounded,
$h:x\mapsto x$ is in $L^m(\pi)$ for all $m>0$.
We consider the Gibbs sampler described
in \citet{robert1995convergence}, that introduces auxiliary variables $\eta$ and alternates between Exponential
and Normal draws as follows:
\begin{align*}
\eta_{i}|\theta & \sim\text{Exponential}\left(\frac{1+(\theta-z_{i})^{2}}{2}\right)\quad\forall i=1,\ldots,n,\\
\theta'|\eta_{1},\ldots,\eta_{n} & \sim\text{Normal}\left(\frac{\sum_{i=1}^{n}\eta_{i}z_{i}}{\sum_{i=1}^{n}\eta_{i}+\sigma^{-2}},\frac{1}{\sum_{i=1}^{n}\eta_{i}+\sigma^{-2}}\right),
\end{align*}
where $\sigma^{2}=100$ is the prior variance. The coupling of this
Gibbs sampler is done using common random numbers for the $\eta$-variables,
and a maximal coupling for the update of $\theta$.

We can verify Assumption \ref{assu:tau-moment-kappa} by direct calculations.
Consider any pair $\theta$, $\tilde{\theta}$,
and the pair of next values, $\theta'$, $\tilde{\theta}'$.
Write $\eta$ and $\tilde{\eta}$ for the auxiliary variables in each chain.
First observe that the means of the Normal distributions, being of the form
\[\frac{\sum_{i=1}^{n}\eta_{i}z_{i}}{\sum_{i=1}^{n}\eta_{i}+\sigma^{-2}} = \frac{\sum_{i=1}^{n}\eta_{i}z_{i}}{\sum_{i=1}^{n}\eta_{i}} \frac{\sum_{i=1}^{n}\eta_{i}}{\sum_{i=1}^{n}\eta_{i}+\sigma^{-2}},\]
take values within $(-a, +a)$ with $a = \max|z_i|$, since they are weighted averages of $(z_i)$ multiplied by a value in $(0,1)$.
Therefore, the mean of the next $\theta'$ is within a finite interval that does not depend on the previous $\theta$.
Regarding the variance $(\sum_{i=1}^{n}\eta_{i}+\sigma^{-2})^{-1}$,
note that $\eta_i = -{2}/({1+(\theta-z_{i})^{2}})\log U_i$,
where $U_i$ is Uniform$(0,1)$, thus $\eta_i \leq -2\log U_i$
and finally
$(\sum_{i=1}^{n}\eta_{i}+\sigma^{-2})^{-1} \geq (\sigma^{-2} + \sum_{i=1}^n (-2\log U_i))^{-1}$.
Also, $(\sum_{i=1}^{n}\eta_{i}+\sigma^{-2})^{-1} \leq \sigma^{2}$. 
The distribution of $\sum_{i=1}^n (-2\log U_i)$
does not depend on $\theta$, thus
we can define an interval $(c,d) \subset (0,\infty)$ and $\epsilon\in(0,1)$,
independently of $\theta$ and $\tilde{\theta}$, such that
$(\sum_{i=1}^{n}\eta_{i}+\sigma^{-2})^{-1} \in (c,d)$ and $(\sum_{i=1}^{n}\tilde{\eta}_{i}+\sigma^{-2})^{-1} \in (c,d)$ simultaneously with probability $\epsilon$.
Therefore, with probability $\epsilon$, $(\theta',\tilde{\theta}')$ is
drawn from a maximal coupling of two Normals, which means and variances are
in finite intervals defined independently of $(\theta,\tilde{\theta})$.
Two such Normals have a total variation distance that is bounded away
from one, thus there exists $\delta>0$ such that,
$\mathbb{P}(\theta'=\tilde{\theta}'|\theta,\tilde{\theta})>\epsilon\delta$,
for some $\epsilon>0$, $\delta>0$ and all $(\theta$,$\tilde{\theta})$.
Hence, Assumption \ref{assu:tau-moment-kappa} holds for all $\kappa\geq 1$.

The second algorithm is a Metropolis--Rosenbluth--Teller--Hastings
(MRTH) algorithm with Normal proposal on $\theta$, with standard deviation $10$.
Its coupling employs a reflection-maximal coupling of the proposals as described in Appendix~\ref{appx:couplingrh}.
Verification of Assumption \ref{assu:tau-moment-kappa} can be
done as described in Section \ref{subsec:interpretassumption} via a drift condition \citep{jarner2000geometric}, since
the target is super-exponential.
Thus, Assumption \ref{assu:tau-moment-kappa} holds again for all $\kappa\geq 1$.

\begin{figure}[t]
  \centering \begin{subfigure}[b]{0.45\columnwidth} \includegraphics[width=1\columnwidth]{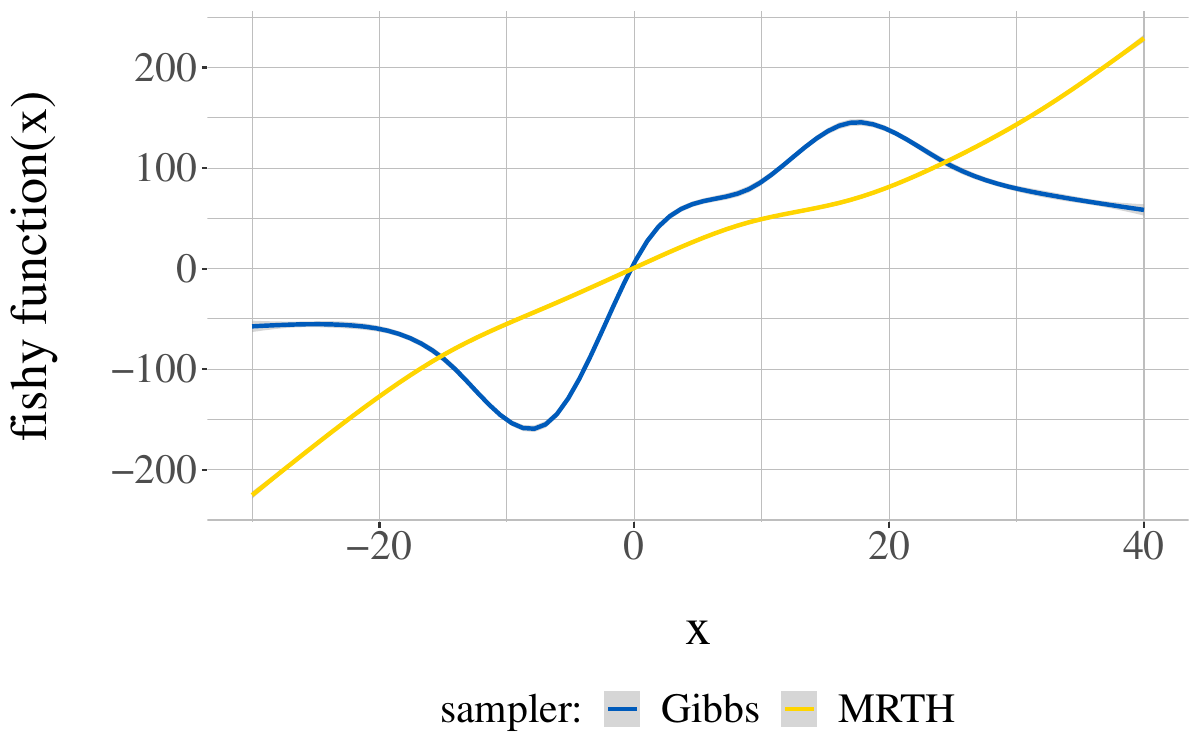}
\caption{}
\label{subfig:cauchynormal:htilde} \end{subfigure} \hspace*{1cm}
\begin{subfigure}[b]{0.45\columnwidth} \includegraphics[width=1\columnwidth]{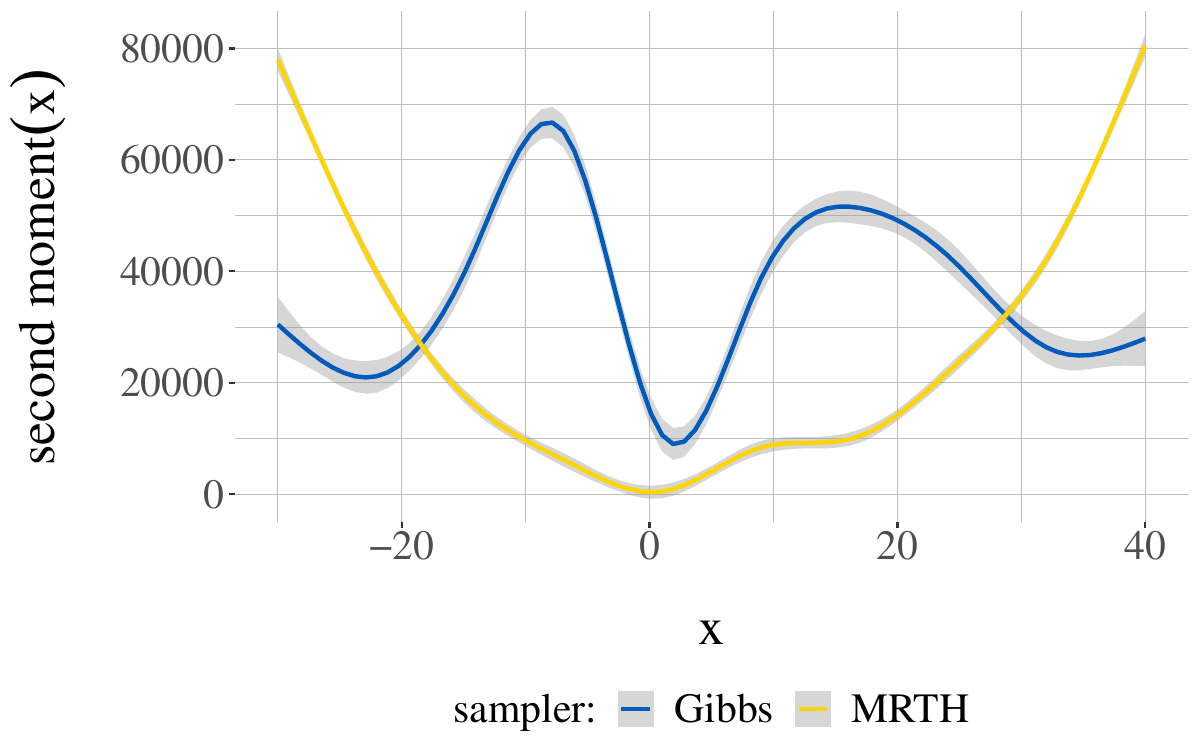}
\caption{}
\label{subfig:cauchynormal:Htildem2} \end{subfigure}
\caption{Cauchy-Normal example: estimation of $\protect\fishytest(x)$ (left) and of second moment
of the estimator, $\mathbb{E}[\protect\fishytestestimator(x)^{2}]$
(right), for two different MCMC algorithms.
\label{fig:cauchynormal:htilde}}
\end{figure}

The state $y$ used to define $\fishytest=g_{y}$ in \eqref{eq:explicitfishy-1}
is set to zero.  Figure \ref{subfig:cauchynormal:htilde} shows the estimated
fishy functions for the Gibbs sampler and MRTH, and Figure
\ref{subfig:cauchynormal:Htildem2} shows the estimated second moments
$\mathbb{E}[\fishytestestimator(x)^{2}]$, for a grid of values of $x$ and using
$10^{3}$ independent repeats of $\fishytestestimator(x)$ for each $x$. 
According to Theorem \ref{thm:fishy-main-text}, $\mathbb{E}[|\fishytestestimator(x)|^p|<\infty$  for all $p\geq 1$, for both algorithms.
Figure \ref{subfig:cauchynormal:htilde} shows that the
fishy functions $\fishytest$ are markedly different for both algorithms. For the Gibbs sampler
the fishy function appears uniformly bounded, whereas it seems to diverge for MRTH.
To understand this difference,
we recall an interpretation of fishy functions due to \citet[Section 4]{kontoyiannis2009notes}.
The average $t^{-1}\sum_{s=0}^{t-1}\test(X_{s})$ has expectation
$t^{-1}\sum_{s=0}^{t-1}P^{s}\test(x)$ given $X_{0}=x$, thus
$\fishytest_{\star}(x)$ in \eqref{def:g-star} represents the leading term in the asymptotic
bias:
\begin{equation}
\fishytest_{\star}(x)=\lim_{t\to\infty}t\left\{ \mathbb{E}_{x}\left[t^{-1}\sum_{s=0}^{t-1}\test(X_{s})\right]-\pi(\test)\right\} .\label{eq:asymptoticbias}
\end{equation}
Thus $g_\star(x)$ measures the asymptotic influence of a starting point $X_0 = x$ in the MCMC estimate.
Figure \ref{subfig:cauchynormal:htilde} shows that this influence
is unbounded for the random walk MRTH algorithm, which is intuitive: more MRTH steps are required for the chain to reach stationarity as $|X_0|$ increases. The same is not true for the Gibbs sampler. Indeed, an arbitrary large initial value for $\theta_0$ would result in $\eta$-variables
arbitrarily close to zero, which in turn would result in $\theta_1$ being distributed approximately as the prior distribution $\text{Normal}(0,\sigma^2)$. Thus, the chain returns
to the center of the space in one step, from an arbitrary initial state.

\section{Glynn--Rhee estimators\label{subsec:recoverGlynnRhee}}

\subsection{Recovering Glynn--Rhee estimators from the Poisson equation}

\citet{glynn2014exact} show how certain couplings of Markov chains could be simulated so as to
construct unbiased estimators of $\pi(h)$, for a class of test functions $h$, that can be computed for a random but finite cost. If the expected cost and the variance of an estimator are finite, then averages of independent copies
converge at the Monte Carlo rate to $\pi(h)$ \citep{glynn1992asymptotic}. Such estimators can be advantageous compared to ergodic average MCMC estimators, as independent copies can be produced by parallel machines. Section 3.3 in \citet{joa2020} provides a review on the use of unbiased estimators generated by parallel machines.
In this section we explain how the unbiased estimators of \citet{glynn2014exact}, as well as the variants of \citet{joa2020}, can
be derived from the Poisson equation, and we provide conditions implying the finiteness on their moments of any order (Theorem~\ref{thm:unbiased-estimator-main-text}). In Proposition~\ref{prop:cltunbiasedmcmc} we provide an asymptotic result facilitating the comparison between standard and unbiased MCMC estimators. Finally, in Section~\ref{subsec:subssampledunbiasedestimators} we discuss the equivalent of ``thinning'' in unbiased MCMC, which is useful for the developments of Section~\ref{sec:Asymptotic-variance}.

First observe
that we may rearrange \eqref{eq:poisson-equation} as
\begin{equation}
\pi(\test)=\test(x)+P\fishytest(x)-\fishytest(x),\qquad x\in\mathbb{X},\label{eq:rearrange-fishy-point}
\end{equation}
where the left-hand side does not depend on $x$ and $g$ is any fishy
function for $h$. It then seems natural to estimate $\pi(h)$ by
estimating the terms on the right-hand side, for any $x$. With $g=g_{\star}$
as in Definition~\ref{def:g-star} we
may write
\begin{equation}
\pi(h)=\test(x)+\sum_{t=0}^{\infty}P^{t+1}\test(x)-P^{t}\test(x),\qquad x\in\mathbb{X},
\end{equation}
where the right-hand side is a familiar quantity in the light of \citet{glynn2014exact}.
In particular, one may run Algorithm \ref{alg:coupledchains} with $L=1$:
starting from $X_{0}=Y_{0}=x$, sample $X_{1}\sim P(X_{0},\cdot)$,
and iteratively sample $(X_{t+1},Y_{t})$ using the transition $\bar{P}$ defined in Section~\ref{subsec:Coupled-Markov-chains}.
The generated $(X,Y)$ process is such that $X_{t+1}=Y_{t}$ almost surely for all $t$ large enough.
Since $P^{t}\test(x)$ is the expectation of both $\test(X_{t})$
and $\test(Y_{t})$, the estimator $h(X_{0})+\sum_{t=0}^{\infty}\{\test(X_{t+1})-\test(Y_{t})\}$
is unbiased for $\pi(\test)$ under suitable assumptions (see Theorem~\ref{thm:unbiased-estimator-main-text}).

This same perspective on \eqref{eq:rearrange-fishy-point} suggests
that for any $x\in\mathbb{X}$ we may define the equivalent and notationally
convenient approximation $h(x)+G_{x}(X_{1})$, where $X_{1}\sim P(x,\cdot)$,
and it is clear that if $\mathbb{E}[G_{y'}(x')]=g_{y'}(x')$ for $\pi$-almost
all $x'$ and $y'=x$, then
\begin{align*}
\mathbb{E}_{x}\left[h(x)+G_{x}(X_{1})\right] & =h(x)+\mathbb{E}_{x}\left[g_{x}(X_{1})\right]\\
 & =h(x)+\mathbb{E}_{x}\left[g_{\star}(X_{1})\right]-g_{\star}(x)  = \pi(h).
\end{align*}
The initialization of the chains can be generalized from a point mass
at $x\in\mathbb{X}$ to a probability distribution $\mu$ (a condition on the initial distribution will be stated in Theorem~\ref{thm:unbiased-estimator-main-text}). Indeed,
a re-arranged and integrated Poisson equation is
\begin{equation}
\pi(\test)=\mu(h)+\mu P(\fishytest)-\mu(\fishytest),
\end{equation}
and this suggests the following estimator of $\pi(h)$, with
the same justification as above.

\begin{defn}
\label{def:unbiased-estimator}For $h\in L^{1}(\pi)$, denote the
approximation of $\pi(h)$ by
\begin{equation}
H=h(X_{0}')+G_{Y_{0}'}(X_{1}'),
\end{equation}
where marginally $X_{0}'\sim\mu$, $Y_{0}'\sim\mu$ and $X_{1}'\sim\mu P$
for some probability measure $\mu$, and $G$ is in Definition~\ref{def:Gyx}. We denote by $\gamma$ the joint
probability measure for $(X_{1}',Y_{0}')$, since this features in
our analysis, noting that this is a coupling of $\mu P$ and $\mu$.
\end{defn}

The estimator in Definition~\ref{def:unbiased-estimator} is identical
to the estimator of \citet{glynn2014exact}, denoted by $H_{0}(X,Y)$ in \citet{joa2020}, if one
chooses $X_{1}'\sim P(X_{0}',\cdot)$ and $(X_{0}',Y_{0}')$ is drawn
from some coupling of $\mu=\pi_{0}$ with itself. In the following we focus on an
independent initialization of the chains.

\subsection{Recovering unbiased MCMC estimators\label{subsec:recoverUMCMC}}

We can also retrieve the more efficient variants proposed in \citet{joa2020}.
By changing the
initial distribution $\gamma$, Definition~\ref{def:unbiased-estimator}
also admits the estimators denoted by $H_{k}(X,Y)$ 
for some $k\in\mathbb{N}$ in \citet{joa2020}, and the estimators denoted by $H_{k:m}(X,Y)$
are obtained as averages of $H_{k}(X,Y)$ over a range of values of $k$.
Unbiased estimators based on chains coupled with a lag $L>1$ \citep{VanettiDoucet2020}
can be retrieved as well by considering the Poisson equation associated
with the iterated Markov kernel $P^{L}$. To make this precise, we
present the following definition.

\begin{defn}
\label{def:lagged-offset}The $L$-lagged and $k$-offset approximation
of $\pi(h)$ is
\[
H_{k}:=h(X_{k})+\sum_{j=1}^{\infty}h(X_{k+jL})-h(Y_{k+(j-1)L}),
\]
where $k\in\mathbb{N}$, $L\geq 1$, $(X_{t+L},Y_{t})_{t\geq0}$
is a time-homogeneous Markov chain with Markov kernel $\bar{P}$,
and $(X_{t})_{t=0}^{L}$ is a Markov chain with transition kernel
$P$ and initial distribution $\pi_{0}$ and $Y_{0}\sim\pi_{0}$
independent of $(X_{0},\ldots,X_{L})$. In particular, $(X_{L},Y_{0})\sim \pi_{0}P^{L}\otimes\pi_{0}$. The coupled chain may be generated with Algorithm~\ref{alg:coupledchains}, starting from two independent draws from $\pi_0$, and the meeting time $\tau$ here is $\inf\{t> L: X_t = Y_{t-L}\}$.
We define for any $k,\ell\in\mathbb{N}$ with $k\leq\ell$, the
average of such estimators as
$H_{k:\ell}:={(\ell-k+1)}^{-1}\sum_{t=k}^{\ell}H_{t}$.
\end{defn}

It will be convenient to view the unbiased MCMC estimator $H_{k:\ell}$
in Definition~\ref{def:lagged-offset} as equivalent to the expectation
of $h$ with respect to an unbiased MCMC signed measure $\hat{\pi}_{k:\ell}$. 
After some routine calculations described in Appendix~\ref{appx:unbiasedmcmc}, and replacing evaluations of the test function $h$
by Dirac delta masses, one obtains the following signed measure
\begin{align}
  \hat{\pi}_{k:\ell}({\rm d}x)&=\frac{1}{\ell-k+1}\sum_{t=k}^{\ell}\delta_{X_{t}}({\rm d}x)+\sum_{t=k+L}^{\tau-1}\frac{v_t}{\ell-k+1}\left\{ \delta_{X_{t}}-\delta_{Y_{t-L}}\right\} ({\rm d}x),\label{eq:pihatmeasure}\\
  \text{with} \quad
  v_t &= \lfloor(t-k) / L\rfloor - \lceil \max(L, t-\ell)/L\rceil + 1.
  \label{eq:weight_in_pihatmeasure}
\end{align}

The following result is established in
Appendix~\ref{subsec:Unbiased-approx-pi-test}, as a particular case of
Proposition~\ref{prop:lag-average-transference}, where upper bounds are given
for $\mathbb{E}[|H_{k:\ell}|^{p}]$. This result can be compared with
\citet[Proposition 1,][]{joa2020} and \citet[Theorem 1,][]{middleton2020unbiased}, which provide only finite
second moments. The latter obtains the same conditions on $\kappa$ and $m$ for
$p=2$. The bounded ${\rm d}\pi_{0}/{\rm d}\pi$ assumption allows one to avoid
the less explicit assumption that $\sup_{n\geq 0} \pi_0 P^n(|h|^{2+\eta}) <
\infty$ and can often be verified as the MCMC user chooses $\pi_0$.
\begin{thm}
\label{thm:unbiased-estimator-main-text}Under Assumption~\ref{assu:tau-moment-kappa},
let $h\in L^{m}(\pi)$ with $m>\kappa/(\kappa-1)$, and ${\rm d}\pi_{0}/{\rm d}\pi\leq M$ with $M<\infty$.
Then, with $H_{k:\ell}$ in Definition~\ref{def:lagged-offset}
for any $k,\ell\in\mathbb{N}$ with $k \leq \ell$, any $L\geq 1$, $\mathbb{E}[H_{k:\ell}]=\pi(h)$
and for $p\geq1$ such that $\frac{1}{p}>\frac{1}{m}+\frac{1}{\kappa}$,
$\mathbb{E}[|H_{k:\ell}|^{p}]<\infty$.
\end{thm}

\begin{rem}
By Theorem~\ref{thm:unbiased-estimator-main-text} it is sufficient
that $\kappa>2$ and $m>2\kappa/(\kappa-2)$ for $H$ to have a finite
variance. On the other hand, Theorem~\ref{thm:clt-kappa} implies
that a CLT holds for $h$ if $\kappa>1$ and $m>2\kappa/(\kappa-1)$,
which is weaker. The stronger condition in Theorem~\ref{thm:unbiased-estimator-main-text}
is because finite second moment of the unbiased estimator is shown
via finite second moment of the approximation of $g_{\star}$, and
this requires $g_{\star}\in L_{0}^{2}(\pi)$.
\end{rem}

Theorem~\ref{thm:unbiased-estimator-main-text} provides conditions for the finiteness of moments of unbiased MCMC estimators,
but does not provide a direct comparison with the moments of ergodic average MCMC estimators.
Comparisons of the second moments
of unbiased and ergodic average MCMC are provided in Proposition 3 of \citet{joa2020}
and Proposition 1 of \citet{middleton2020unbiased}.
In the same spirit, but perhaps more directly, consider the following CLT for the unbiased MCMC estimator $H_{k:\ell}$,
where the asymptotic variance is the same as in the CLT for ergodic average MCMC \eqref{eq:clt}.
\begin{prop}\label{prop:cltunbiasedmcmc}
Assume the conditions of Theorem~\ref{thm:clt-kappa} and let $k\in\mathbb{N}$, and let $L\geq 1$.
Then $H_{k:\ell}$ in Definition~\ref{def:lagged-offset} satisfies, as $\ell\to\infty$,
\[
\sqrt{\ell-k+1}\left( H_{k:\ell}-\pi(h)\right) \overset{d}{\to}{\rm Normal}(0,v(P,h)).
\]
\end{prop}

Proposition \ref{prop:cltunbiasedmcmc} is established in Appendix~\ref{subsec:Unbiased-approx-pi-test}. 
Under the stated conditions, for suitably large $\ell-k$ the
concentration of $H_{k:\ell}$ is similar to that of the ergodic average MCMC
estimator of a similar computational cost, noting that one is simulating
only a single chain after the meeting time. 
Proposition \ref{prop:cltunbiasedmcmc} suggests that the efficiency of unbiased MCMC estimators 
can be made arbitrarily close to that of ergodic average MCMC, by choosing a large enough $\ell$.
For a fixed $\ell$, one can assess the efficiency loss by comparing the product of expected cost
and variance of unbiased MCMC, with the asymptotic variance $v(P,h)$.
However, the numerical approximation of $v(P,h)$ is not easy,
and typically requires long chains. Accordingly, 
\citet{glynn2014exact,agapiou2018unbiased,joa2020} all compare the efficiency of unbiased estimators
with that of ergodic averages obtained from long runs. In Section \ref{sec:Asymptotic-variance}
we propose a new asymptotic variance estimator that is itself unbiased, and thus 
can be generated independently on parallel machines, for a random but finite time, and averaged to obtain a
consistent estimator of $v(P,h)$ converging at the Monte Carlo rate. This alleviates
the need for long chains, not only for the estimation of $\pi(h)$ but also for the efficiency assessment of unbiased MCMC.

\subsection{Subsampled unbiased estimators\label{subsec:subssampledunbiasedestimators}}

To prepare for the description of the estimators in Section~\ref{sec:Asymptotic-variance},
we simplify the notation of the signed measure $\hat{\pi}_{k:\ell}$ in \eqref{eq:pihatmeasure}-\eqref{eq:weight_in_pihatmeasure}.
Let $N=(\ell-k+1) + 2 \max(0, \tau-(k+L))$ be the random number of atoms in  $\hat{\pi}_{k:\ell}$,
and denote these atoms by $Z_1, \ldots, Z_N$.
The first $\ell-k+1$ atoms correspond to $X_k,\ldots,X_\ell$, and 
the following $2 \max(0, \tau-(k+L))$ atoms (if any) correspond to $X_{k+L},Y_k,\ldots,X_{\tau-1},Y_{\tau-1-L}$.
There may be duplicate atoms among $Z_{1},\ldots,Z_{N}$.
We can write
\begin{equation}
\hat{\pi}_{k:\ell}({\rm d}x) = \sum_{i=1}^{N} \omega_i \delta_{Z_i}({\rm d}x),
  \label{eq:pihatmeasure_simplified}
\end{equation}
where $\omega_i$ are the weights, 
of the form $(\ell-k+1)^{-1}$ for the first $\ell-k+1$ atoms, and of the form
$\pm v_t (\ell-k+1)^{-1}$ with $v_t$ defined in \eqref{eq:weight_in_pihatmeasure} and $t\in\{k+L,\ldots,\tau-1\}$ for the remaining atoms.
We further simplify the notation by writing $\hat{\pi}$ instead of $\hat{\pi}_{k:\ell}$.
For a function $h$, $\hat{\pi}(h)=\sum_{i=1}^N \omega_i h(Z_i)= H_{k:\ell}$,
where $H_{k:\ell}$ is the unbiased MCMC estimator in Definition~\ref{def:lagged-offset}.

We can obtain estimators of $\pi(h)$ by subsampling 
the atoms of $\hat{\pi}$ in \eqref{eq:pihatmeasure_simplified}.
This results in estimators with lower
computational cost, for example if evaluations of the test function $h$ are expensive
in comparison to the simulation of $\hat{\pi}$. The computational benefits of subsampling in this context are related to the thinning ideas in \citet{owen2017statistically}.
The following result is established in Appendix~\ref{subsec:Unbiased-approx-pi-test}. 
It demonstrates that the sufficient conditions for
lack-of-bias and finite $p$th moments are identical for $H_{k:\ell}$
and the subsampled estimator $S_{R}$ for any $R\geq1$.
\begin{thm}
\label{thm:subsampled-main-text}Under Assumption~\ref{assu:tau-moment-kappa},
let $h\in L^{m}(\pi)$ for some $m>\kappa/(\kappa-1)$, ${\rm d}\pi_{0}/{\rm d}\pi\leq M$,
$k,\ell\in\mathbb{N}$ with $k\leq\ell$, and $\hat{\pi}=\sum_{i=1}^{N}\omega_{i}\delta_{Z_{i}}$
be the unbiased signed measure in \eqref{eq:pihatmeasure_simplified}.
Define for some integer $R\geq1$,
\[
S_{R}=\frac{1}{R}\sum_{i=1}^{R} \xi_{I_i}^{-1} \omega_{I_{i}}h(Z_{I_{i}}),
\]
where $I_{1},\ldots,I_{R}$ are conditionally independent ${\rm Categorical}\{\xi_{1},\ldots,\xi_{N}\}$
variables with
\[
\frac{a}{N}\leq\min_{i}\xi_{i}\leq\max_{i}\xi_{i}\leq\frac{b}{N},
\]
for some constants $0<a\leq b<\infty$ that may be functions of $k,\ell,L$ but not $\tau$.
Then $\mathbb{E}[S_{R}]=\pi(h)$
and for $p\geq1$ such that $\frac{1}{p}>\frac{1}{m}+\frac{1}{\kappa}$,
$\mathbb{E}\left[\left|S_{R}\right|^{p}\right]<\infty$.
\end{thm}

\begin{rem}
  The more detailed Proposition~\ref{prop:subsample} in Appendix~\ref{subsec:Unbiased-approx-pi-test} provides a bound
on $\mathbb{E}\left[\left|S_{R}\right|^{p}\right]$
but this bound does not depend on the value of $R$. For $p=2$, we
can see that increasing $R$ decreases the variance of $S_{R}$,
through  the law of total variance
\begin{align*}
{\rm var}\left(S_{R}\right) & =\mathbb{E}\left[{\rm var}(S_{R}\mid\hat{\pi})\right]+{\rm var}\left(\mathbb{E}\left[S_{R}\mid\hat{\pi}\right]\right) =\frac{1}{R}\mathbb{E}\left[{\rm var}(S_{1}\mid\hat{\pi})\right]+{\rm var}(\hat{\pi}(h)).
\end{align*}
Hence, there is a tradeoff between computational cost and variance,
with increasing $R$ potentially improving efficiency
up to point where the variance is dominated
by ${\rm var}(\hat{\pi}(h))$.
\end{rem}

\section{Asymptotic variance estimation\label{sec:Asymptotic-variance}}

\subsection{Standard methods and motivation}

The asymptotic variance $v(P,h)$ in \eqref{eq:clt} is commonly estimated
using batch means, spectral variance \citep[see, e.g., ][]{flegal2010batch,vats2018strong,vats2019multivariate,chakraborty2019estimating}, or initial sequence estimators \citep{geyer1992practical,BergSong2022}.
We propose new estimators of $v(P,h)$, one of
which has the distinctive property of being unbiased, computable in a (random) finite time, and with $p$ finite moments under conditions given in Theorem \ref{thm:unbiasedavar_pfinitemoments}. 
Importantly, if their two first moments are finite, unbiased estimators of $v(P,h)$ can be generated independently
in parallel, and their average converges at the Monte Carlo rate. 
This compares favorably to the rate of standard estimators of $v(P,h)$.

Indeed, with batch means and spectral variance estimators, the length of the chain $t$ is decomposed as a product $a_t b_t$. For batch means, $a_t$ refers to a number of batches and $b_t$ to a batch size. For spectral variance estimators, $b_t$ refers to a truncation or bandwidth parameter and $a_t$ denotes $t/b_t$.
In either case, $b_t$ must grow with $t$ for the bias to vanish as $t\to\infty$, 
thus $a_t$ cannot grow linearly in $t$. On the other hand, the variance of these estimators is typically of the order of $a_t^{-1}$. Thus, even with optimal choices for $b_t$, the mean squared error vanishes slower than the Monte Carlo rate $t^{-1}$. 
Appendix~\ref{appx:reviewlongrunvarianceestimation}  provides a review of relevant results on asymptotic variance estimation.
This difference in convergence rates, which we observe in numerical experiments in Section~\ref{subsec:ar}, does not mean that the proposed estimators are always preferable. 
For some fixed computational budget, the variance of the proposed estimators can 
be much larger than that of classical estimators, as we observe in Section \ref{subsec:logitrandom}.

Asymptotic variance estimators are routinely employed to construct confidence intervals in MCMC settings, using the CLT \eqref{eq:clt}; although \citet{atchade2016markov} shows that consistent estimators of $v(P,h)$ may not be necessary for this task. The estimators proposed below require couplings of Markov chains, that can be used to obtain unbiased MCMC estimators as in Section \ref{subsec:recoverUMCMC}. Confidence intervals for unbiased MCMC can be constructed without estimating $v(P,h)$. Nevertheless, estimation of $v(P,h)$ remains a critical task. It enables for example  efficiency comparisons between unbiased and ergodic average MCMC estimators, as argued in Section \ref{subsec:recoverUMCMC}. It may also be useful in comparing MCMC algorithms, as we illustrate in the numerical experiments of Section \ref{subsec:logitrandom},
or in choosing the tuning parameters $k$, $L$ and $\ell$.

\subsection{\label{subsec:consistentavar}Ergodic Poisson asymptotic variance
estimator}

We employ coupled Markov chains, as generated
by Algorithm~\ref{alg:coupledchains}, to define new estimators of
$v(P,\test)$. We start with a consistent estimator in this section, before
developing an unbiased estimator in Section \ref{subsec:unbiasedavar}. We first re-express
\eqref{eq:avar-intro} as
\begin{align}
v(P,\test) & =-v(\pi,\test)+2\pi((\test-\pi(\test))\cdot\fishytest),\label{eq:asymptoticvariance:estimable2}
\end{align}
where $v(\pi,\test)$ is the variance of $\test(X)$ under $X\sim\pi$.
Consider an MCMC approach to estimate \eqref{eq:asymptoticvariance:estimable2}.
Expectations with respect to $\pi$ can be consistently estimated
from long MCMC runs. Using $t$ steps of a  Markov chain evolving with transition $P$, after a burn-in period, and re-indexing iterations
so that $t=0$ corresponds to the first iteration after the discarded burn-in period, we define the
empirical measure $\pi^{\text{MC}}=t^{-1}\sum_{s=0}^{t-1}\delta_{X_{s}}$.
We approximate $\pi(\test)$ and $v(\pi,\test)$ using the
empirical mean and variance, denoted by $\pi^{\text{MC}}(\test)$
and $v^{\text{MC}}(\test)$. 

The difficulty is in the term $\pi((\test-\pi(\test))\cdot\fishytest)$
in \eqref{eq:asymptoticvariance:estimable2}, since $\fishytest$
cannot be evaluated exactly. We employ unbiased estimators $\fishytestestimator(x)$
as in Definition~\ref{def:Gyx} of Section \ref{sec:coupledchainsandfishyfunctions} in place of evaluations $\fishytest(x)$.
This leads to the following estimator of $v(P,\test)$, which we
call the ergodic Poisson asymptotic variance estimator (EPAVE),
\begin{equation}
\hat{v}_{E}(P,h)=-v^{\text{MC}}(\test)+\frac{2}{t}\sum_{s=0}^{t-1}(\test(X_{s})-\pi^{\text{MC}}(\test))\cdot\fishytestestimator(X_{s}),\label{eq:asymptoticvariance:consistentestimator}
\end{equation}
where each $G(X_{s})$ is conditionally independent of all others
given $(X_{0},\ldots,X_{t-1})$. 
We can compute \eqref{eq:asymptoticvariance:consistentestimator}
online by keeping track of the sums
\[
\sum_{s}\test(X_{s}),\sum_{s}\test(X_{s})^{2},\sum_{s}\fishytestestimator(X_{s}),\sum_{s}\test(X_{s})\fishytestestimator(X_{s}).
\]
We can modulate the relative cost of estimating the fishy function
evaluations in \eqref{eq:asymptoticvariance:consistentestimator}
by generating $\fishytestestimator(X_{s})$ at times $s$ such that $s \mod D=0$ for $D\in\mathbb{N}$. 
Preliminary runs can provide an estimate of the average cost 
$C$ of generating $G$ at states that approximately follow $\pi$, in units of transitions from $P$.
Setting $D$ as $\lceil C \rceil$ ensures that the effort of computing \eqref{eq:asymptoticvariance:consistentestimator}
 is split approximately equally between the MCMC run and the estimators $G$.

In Appendix~\ref{subsec:mcmc-estimator-avar}, we show that the estimator $\hat{v}_{E}(P,h)$ is
strongly consistent as $t\to\infty$, and satisfies a $\sqrt{t}$-CLT under
Assumption~\ref{assu:tau-moment-kappa} and moment assumptions on $h$. The
following summarizes Proposition~\ref{prop:epave-as} 
and Theorem~\ref{thm:epave-clt}.
An expression for the asymptotic variance of EPAVE can be extracted from the proof.

\begin{thm}\label{thm:validityEPAVE}
Under Assumption~\ref{assu:tau-moment-kappa}, let $X$ be a Markov
chain with Markov kernel $P$, and $h\in L^{m}(\pi)$ with $m>2\kappa/(\kappa-1)$.
Then for $\pi$-almost all $X_{0}$ and $\pi$-almost all $y$,
\begin{enumerate}
  \item The CLT in \eqref{eq:clt} holds for $h$, and $v(P,h)=-v(\pi,h)+2\pi(h_{0}\cdot g_{y})$.
\item The estimator \eqref{eq:asymptoticvariance:consistentestimator} with
$G=G_{y}$, satisfies $\hat{v}_E(P,h)\to_{{\rm a.s.}}v(P,h)$ as $t\to\infty$.
\item If $m>4\kappa/(\kappa-3)$, the estimator \eqref{eq:asymptoticvariance:consistentestimator}
with $G=G_{y}$ satisfies a $\sqrt{t}$-CLT.
\end{enumerate}
\end{thm}

The last item of Theorem \ref{thm:validityEPAVE} is remarkable because it implies, at least under some conditions, a rate of convergence faster than that of batch means and spectral variance estimators, 
see Appendix~\ref{appx:reviewlongrunvarianceestimation}.
On the other hand, EPAVE requires the implementation of successful couplings,
as well as the length $t$ going to infinity. 
We next propose an unbiased estimator of $v(P,h)$, so that, at least, long runs can be avoided: independent copies can be generated in parallel, and averaged.
Note that EPAVE is biased only because the MCMC chain does not start at stationarity. Batch means and spectral variance estimators computed in finite time would be biased even if the chains started at stationarity.

\subsection{\label{subsec:unbiasedavar}Unbiased Poisson asymptotic variance
estimator}

Starting again from \eqref{eq:asymptoticvariance:estimable2}, we
propose an unbiased estimator of $v(P,h)$ by combining unbiased estimators
$\fishytestestimator(x)$ of $\fishytest(x)$ with unbiased
approximations of $\pi$, as retrieved in Section \ref{subsec:recoverGlynnRhee}.
We thus assume that we can generate
random signed measures as in \eqref{eq:pihatmeasure}, denoted by $\hat{\pi}=\sum_{n=1}^{N}\omega_{n}\delta_{Z_{n}}$
as in \eqref{eq:pihatmeasure_simplified}.
Recall that $\hat{\pi}$ is such that $\mathbb{E}[\hat{\pi}(h)]=\pi(\test)$
for a class of test functions $\test$. 
Combining these measures
$\hat{\pi}$ with $\fishytestestimator(x)$ in \eqref{eq:estimatorpoisson},
each term in \eqref{eq:asymptoticvariance:estimable2} can be estimated
without bias, as we next describe. First, unbiased estimators of $v(\pi,\test)=\mathbb{V}_{\pi}[\test(X)]$
can be obtained using two independent unbiased measures $\hat{\pi}^{(1)}$
and $\hat{\pi}^{(2)}$, by computing
\begin{equation}
\hat{v}(\pi,\test)=\frac{1}{2}\{\hat{\pi}^{(1)}(\test^{2})+\hat{\pi}^{(2)}(\test^{2})\}-\hat{\pi}^{(1)}(\test)\times\hat{\pi}^{(2)}(\test).\label{eq:unbiasedtargetvariance}
\end{equation}

Unbiased estimation of the term $\pi((\test-\pi(\test))\cdot\fishytest)$ in
\eqref{eq:asymptoticvariance:estimable2} is more involved. We first
provide an informal reasoning that motivates the proposed estimator
given below in \eqref{eq:unbiasedasymptvar} and described in pseudocode
in Algorithm~\ref{alg:unbiasedvar}.
Consider the problem of estimating $\pi(h\cdot g)$ without bias,
and assume that we can generate unbiased measures $\hat{\pi}$
of $\pi$ and estimators $G(x)$ with expectation equal to $g(x)$
for all $x$. Then we can generate $\sum_{n=1}^{N}\omega_{n}h(Z_{n})\cdot G(Z_{n})$,
where all $(G(Z_{n}))_{n=1}^{N}$ are conditionally independent given $(Z_{n})$. Conditioning on $\hat{\pi}$, we have
\[
\mathbb{E}\left[\left.\sum_{n=1}^{N}\omega_{n}h(Z_{n})\cdot G(Z_{n})\right|\hat{\pi}\right]=\sum_{n=1}^{N}\omega_{n}h(Z_{n})\cdot g(Z_{n})=\hat{\pi}(h\cdot g),
\]
and then taking the expectation with respect to $\hat{\pi}$ yields
$\pi(h\cdot g)$, under adequate assumptions on $h\cdot g$. However,
the variable $\sum_{n=1}^{N}\omega_{n}h(Z_{n})\cdot G(Z_{n})$ requires
estimators of the fishy function for all $N$ locations, and $N$
could be large. Alternatively, after generating $\hat{\pi}$ we can
sample an index $I\in\{1,\ldots,N\}$ according to a Categorical
distribution with strictly positive probabilities $\xi=(\xi_{1},\ldots,\xi_{N})$,
and given $Z_{I}$ we can generate $G(Z_{I})$. This amounts to subsampling as described in Section \ref{subsec:subssampledunbiasedestimators}. By default, we set $\xi_n = 1/N$ for all $n\in\{1,\ldots,N\}$. Then we observe
that, conditioning on $\hat{\pi}$, integrating out the randomness
in $G(Z_{I})$ given $Z_{I}$, and then the randomness in $I$,
\begin{align*}
\mathbb{E}\left[\left.\omega_{I}\xi_{I}^{-1}h(Z_{I})G(Z_{I})\right|\hat{\pi}\right] & =\mathbb{E}\left[\left.\mathbb{E}\left[\left.\omega_{I}\xi_{I}^{-1}h(Z_{I})G(Z_{I})\right|I,\hat{\pi}\right]\right|\hat{\pi}\right]\\
 & =\mathbb{E}\left[\left.\omega_{I}\xi_{I}^{-1}h(Z_{I})g(Z_{I})\right|\hat{\pi}\right]
  =\sum_{n=1}^{N}\omega_{n}h(Z_{n})g(Z_{n})
  =\hat{\pi}(h\cdot g),
\end{align*}
and therefore $\omega_{I}\xi_{I}^{-1}h(Z_{I})\cdot G(Z_{I})$
is an unbiased estimator of $\pi(h\cdot g)$ that requires only one
estimation of $g$ at $Z_{I}$. The estimator proposed below employs
$R\geq1$ estimators of the fishy function for each signed measure
$\hat{\pi}$, where $R$ is a tuning parameter. Its choice
and the selection probabilities $\xi$ are discussed in Section \ref{subsec:improvements}.

\begin{algorithm}
Input: unbiased signed measures $\hat{\pi}$, unbiased fishy function estimators $G$, method to compute selection
probabilities $\xi$, integer $R$.
\begin{enumerate}
\item Obtain two independent unbiased signed measures, $\hat{\pi}^{(j)}=\sum_{n=1}^{N^{(j)}}\omega_{n}^{(j)}\delta_{Z_{n}^{(j)}}$
for $j\in\{1,2\}$.
\item Compute $\hat{v}(\pi,\test)$ as in \eqref{eq:unbiasedtargetvariance}.
\item For $j\in\{1,2\}$,
\begin{enumerate}
\item Compute selection probabilities $(\xi_{1}^{(j)},\ldots,\xi_{N^{(j)}}^{(j)})$.
\item Draw $I^{(j)}_r$ among $\{1,\ldots,N^{(j)}\}$ with probabilities
$(\xi_{1}^{(j)},\ldots,\xi_{N^{(j)}}^{(j)})$, for $r\in\{1,\ldots,R\}$.
\item Evaluate $\test(Z_{I^{(j)}_r}^{(j)})$ and generate estimator $\fishytestestimator(Z_{I^{(j)}_r}^{(j)})$
in \eqref{eq:estimatorpoisson}, for $r\in\{1,\ldots,R\}$. \label{alg:step:evalestimateatom}
\end{enumerate}
\item Return $\hat{v}(P,\test)$ as in \eqref{eq:unbiasedasymptvar}.
\end{enumerate}
\caption{Unbiased Poisson asymptotic variance estimator (UPAVE). \label{alg:unbiasedvar}}
\end{algorithm}

We gather the above considerations to define the proposed estimator. We write
$\hat{\pi}^{(j)}$ as $\sum_{n=1}^{N^{(j)}}\omega_{n}^{(j)}\delta_{Z_{n}^{(j)}}$
for $j\in\{1,2\}$. Given $\hat{\pi}^{(j)}$, we sample integers $I^{(j)}_r\in\{1,\ldots,N^{(j)}\}$
with probabilities $(\xi_{1}^{(j)},\ldots,\xi_{N^{(j)}}^{(j)})$,
independently for $r\in\{1,\ldots,R\}$. Noting that each $(\test(x)-\pi(\test))\fishytest(x)$
is the expectation of $(\test(x)-\hat{\pi}^{(j)}(\test))\fishytestestimator(x)$
given $x$, we obtain
\begin{equation}
\pi((\test-\pi(\test))\cdot\fishytest)=\mathbb{E}\left[\frac{1}{2R}\sum_{i\neq j\in\{1,2\}}\sum_{r=1}^{R}\frac{\omega_{I^{(j)}_r}^{(j)}}{\xi_{I^{(j)}_r}^{(j)}}(\test(Z_{I^{(j)}_r}^{(j)})-\hat{\pi}^{(i)}(\test))\fishytestestimator(Z_{I^{(j)}_r}^{(j)})\right].\nonumber
\end{equation}
Our proposed unbiased estimator of $v(P,\test)$ is thus
\begin{equation}
\hat{v}(P,\test)=-\hat{v}(\pi,\test)+\frac{1}{R}\sum_{i\neq j\in\{1,2\}}\sum_{r=1}^{R}\frac{\omega_{I^{(j)}_r}^{(j)}}{\xi_{I^{(j)}_r}^{(j)}}(\test(Z_{I^{(j)}_r}^{(j)})-\hat{\pi}^{(i)}(\test))\fishytestestimator(Z_{I^{(j)}_r}^{(j)}),\label{eq:unbiasedasymptvar}
\end{equation}
and its generation is described in Algorithm \ref{alg:unbiasedvar}.
The cost of $\hat{v}(P,\test)$ will typically be dominated
by the cost of obtaining $\hat{\pi}^{(j)}$ for
$j\in\{1,2\}$ and $2R$ estimators of evaluations of $g$. We call
$\hat{v}(P,h)$ the unbiased Poisson asymptotic variance estimator
(UPAVE).

We show that under Assumption~\ref{assu:tau-moment-kappa}, UPAVE
is unbiased and has $p$ finite moments whenever $h$ has sufficiently
many moments. 
The following statement combines Theorem~\ref{thm:avar-unbiased-moments}
and Remark~\ref{rem:average-avar-estimators} 
in Appendix~\ref{subsec:Unbiased-asymptotic-variance},
with the further assumption that $\xi_{n}^{(j)}=1/N^{(j)}$
for all $n\in\{1,\ldots,N^{(j)}\}$, for simplicity.

\begin{thm}\label{thm:unbiasedavar_pfinitemoments}
Under Assumption~\ref{assu:tau-moment-kappa}, let $h\in L^{m}(\pi)$
for some $m>2\kappa/(\kappa-2)$, and ${\rm d}\pi_0/{\rm d} \pi\leq M$. Assume $\xi_{n}^{(j)}=1/N^{(j)}$
for $n\in\{1,\ldots,N^{(j)}\}$. Then for any $R\geq1$ and $\pi$-almost
all $y$, $\mathbb{E}\left[\hat{v}(P,\test)\right]=v(P,h)$ and for
$p\geq1$ such that $\frac{1}{p}>\frac{2}{m}+\frac{2}{\kappa}$, $\mathbb{E}\left[\left|\hat{v}(P,\test)\right|^{p}\right]<\infty$.
\end{thm}

If Assumption~\ref{assu:tau-moment-kappa} holds for all $\kappa\geq1$
(respectively $h\in L^{m}(\pi)$ for all $m\geq1)$, one requires
only slightly more than $2$ moments of $h$ (respectively $\tau$)
to estimate the asymptotic variance consistently and slightly more
than $4$ moments of $h$ (respectively $\tau$) to approximate the
asymptotic variance with a variance in $O(1/M)$, if $M$ is the number
of independent unbiased estimators averaged. This seems close to tight,
since $4$ moments of $h$ are required for the sample variance to
have a finite variance in the setting of i.i.d. variables.

\subsection{Implementation and improvements\label{subsec:improvements}}

We provide some heuristics to guide the tuning and implementation of UPAVE,
leaving more principled approaches to future work.

\emph{Tuning of unbiased MCMC.}
The proposed estimator relies on unbiased signed measures $\hat{\pi}$
of $\pi$, as in \eqref{eq:pihatmeasure}. In
our experiments, we start by generating some lagged chains, with a lag $L=1$ by lack of a better guess, and record the meeting
times $\tau$. We then re-define $L$ as a large quantile of the meeting
times, and we set $k=L$. We choose $\ell$ as a multiple of $k$ such as $5k$, following
suggestions in \citet{joa2020}, to ensure a low proportion
$k/\ell$ of discarded iterations. 
We expect the inefficiency
of the resulting unbiased MCMC estimators to be at least $\ell/(\ell-k)$ that of ergodic average MCMC, due to the proportion $k/\ell$ of discarded iterations.

\emph{Choice of $y$.}
The proposed estimator requires setting $y$ to define $g_y$ as in Definition
\ref{def:gy} and its estimator in Definition \ref{def:Gyx}. Then $G_y(x)$
is generated for various $x$  which are approximately distributed according to $\pi$.
As $G_y(x)$ should preferably have a smaller cost and a smaller variance, we should set $y$
such that two chains starting at $x\sim \pi$ and $y$ are likely to meet quickly. Thus,
$y$ should preferably be central with respect to $\pi$. 

\emph{Selection probabilities.}
To implement UPAVE we need to choose selection probabilities $\xi=(\xi_{1},\ldots,\xi_{N})$
given $N$, the number of atoms in the signed measure $\hat{\pi}$.
We set these probabilities to $1/N$ as a default choice. We
can also try to minimize the variance of the resulting estimators
with respect to $\xi$. This requires information on the variance
of $\fishytestestimator$. Indeed, if we condition on the realization
of $\hat{\pi}$, then the variance of the term
\begin{equation}
\frac{\omega_{I_r}}{\xi_{I_r}}(\test(Z_{I_r})-\hat{\pi}(\test))\fishytestestimator(Z_{I_r}),\;\text{where}\; I_r\sim\text{Categorical}(\xi_{1},\ldots,\xi_{N}),\label{eq:termdependingonxi}
\end{equation}
is minimized over $\xi$ as follows. Since its expectation is independent
of $\xi$, we can equivalently minimize its second moment, thus we
define
\begin{align}
\alpha_{n} & =\{\omega_{n}(\test(Z_{n})-\hat{\pi}(\test))\}^{2}\mathbb{E}\left[\fishytestestimator(Z_{n})^{2}|Z_{n}\right],\nonumber \\
\xi_{n}^{\star} & =\frac{\sqrt{\alpha_{n}}}{\sum_{n'=1}^{N}\sqrt{\alpha_{n'}}},\quad n=1,\ldots,N.\label{eq:optimalselection}
\end{align}
The use of $\xi^{\star}$ leads to a second moment of \eqref{eq:termdependingonxi} equal to $(\sum_{n=1}^{N}\sqrt{\alpha_{n}})^{2}$,
but for any $\xi$ such that $\sum_{n=1}^{N}\xi_{n}=1$, the
Cauchy--Schwarz inequality implies $\sum_{n=1}^{N}\alpha_{n}/\xi_{n}\geq(\sum_{n=1}^{N}\sqrt{\alpha_{n}})^{2}$.
Therefore, $\xi^{\star}$ in \eqref{eq:optimalselection} results in
the smallest variance of the term in \eqref{eq:termdependingonxi}.
We experiment with the estimation of $\xi^\star$ in Appendix H.

\emph{Choice of $R$.}
We need to choose $R$, the integer such that $2R$ is the number of states
at which the fishy function $g$ is estimated.
We can guide the choice of $R$
numerically by monitoring the inefficiency defined as the product of expected
cost and variance, which can be approximated from
independent copies of UPAVE.
In our experiments we observe gains in efficiency
when setting $R$ to a value such that 
the cost of generating the two unbiased signed measures
matches approximately the cost of $2R$ fishy function estimators.
This way, at most half
of the computing budget is allocated suboptimally.
We note that when we run UPAVE for a given choice of $R$, we can
also easily output estimators corresponding to smaller values of $R$,
at no extra cost, which helps in monitoring the effect of $R$.

\emph{Reservoir sampling.}
A naive implementation of UPAVE with Algorithm~\ref{alg:unbiasedvar}
could incur a large memory cost when each state in $\mathbb{X}$ is
large, as in high-dimensional regression \citep{biswas2021couplingbased},
or phylogenetic inference \citep{kelly2021lagged}. Indeed, storing
all the atoms of the generated signed measures might be cumbersome.
However, for UPAVE we only need to select within each measure $R$
atoms at which to evaluate $\test$ and to estimate $\fishytest$;
see Line \ref{alg:step:evalestimateatom} in Algorithm~\ref{alg:unbiasedvar}.
We can address the memory issue by setting $\xi_{n}=1/N$
for all $n$ and by using reservoir sampling \citep{vitter1985random}.
This technique allows sampling $I_r$ uniformly in $\{1,\ldots,N\}$,
$R$ times independently, without knowing $N$ in advance and
keeping only $R$ objects in memory.

We mention other methodological variations that we do not investigate further
in this manuscript.
Instead of sampling $R$ atoms from each signed measure with replacement,
we could sample without replacement.
Also, the $2R$ estimators $\fishytestestimator$ employed in \eqref{eq:unbiasedasymptvar}
could be generated jointly instead of independently. In particular,
we can couple $2R$ chains starting from $(Z_{I^{(j)}_r})_{r=1}^{R}$,
$j\in\{1,2\}$ with a common chain starting from $y$: in other words,
we could simulate a single coupling of $2R+1$ chains instead of simulating
$2R$ couplings of two chains.

\subsection{\label{subsec:multivariate-extension}Multivariate extension}

Using the Cram\'er--Wold theorem \citep[][Theorem 29.4]{Billingsley95}, we can handle test
functions $h$ taking values in $\mathbb{R}^d$:
$h(x)=(h_1(x),\ldots,h_d(x))$. Write $h_0 = h - \pi(h)$.  Consider the
Poisson equation for each $h_i$ and introduce the associated solutions denoted
by $g_i$, and $g=(g_1,\ldots,g_d)$. The sum $\sum_{s=0}^{t-1} (h_i(X_s)-\pi(h_i))$ can be re-written as
$\sum_{s=1}^{t-1} (g_i(X_s)-Pg_i(X_{s-1})) + g_i(X_0) - Pg_i(X_{t-1})$.
Observe that $S_i =  (g_i(X_t) - Pg_i(X_{t-1}))_{t\geq 1}$ is a martingale
difference sequence for which a central limit theorem applies, with asymptotic
variance as in \eqref{eq:avar-intro},
with $g_i$ instead of $g$.
Write $S=(S_1,\ldots,S_d)$.
For any vector $t\in \mathbb{R}^d$, we find that $t^T S$ is a martingale difference sequence as well,
and by the Cram\'er--Wold theorem the multivariate asymptotic variance is
\begin{equation}
  v(P,h) = \mathbb{E}_\pi\left[(g(X_1) - Pg(X_0))(g(X_1) - Pg(X_0))^T\right].
  \label{eq:avar-multivariate}
\end{equation}
Next, the multivariate extension of the alternate representation in \eqref{eq:asymptoticvariance:estimable2}
is obtained by developing the product,
then by using $Pg = g - h_0$ pointwise and elementwise. We obtain
\begin{equation}
  v(P,h) = 
  -\pi(h_0 h_0^T + (h - \pi(h)) g^T + g (h - \pi(h))^T).
  \label{eq:avar-multivariate-estimable}
\end{equation}
The $(i,j)$-th entry of that matrix can be written
\begin{equation}
  -(\pi(h_i\cdot h_j) - \pi(h_i)\pi(h_j)) + \pi\left[(h_i - \pi(h_i)) g_j + g_i (h_j - \pi(h_j))\right].
  \label{eq:avar-multivariate-estimable2}
\end{equation}
Therefore we can estimate multivariate asymptotic variances with the proposed UPAVE,
using pairs of independent unbiased signed measure approximations of $\pi$,
and unbiased estimators of evaluations of each coordinate of the fishy function $g$.

\section{Numerical experiments with asymptotic variance estimators\label{sec:numerical:experiments}}

We implement the proposed estimators of $v(P,h)$ and investigate their distinctive features, 
first in an AR(1) example where $v(P,h)$ is analytically available.
Then we consider MCMC algorithms that are representative of 
methods commonly used in Bayesian data analysis.
In Section~\ref{subsec:logitrandom},
we consider a Hamiltonian Monte Carlo algorithm targeting a posterior distribution in logistic regression
with random effects \citep{heng2019unbiased}.
Supplementary material \citep{aospoisson-supp} includes experiments with
a Gibbs sampler for Bayesian regression with shrinkage prior where the target is high-dimensional,
multimodal and heavy-tailed \citep{biswas2021couplingbased}, and a particle marginal Metropolis--Hastings
algorithm to infer parameters in state space models \citep{middleton2020unbiased}.

In all examples below, we tune unbiased MCMC as suggested in Section \ref{subsec:improvements}, with $L$ obtained from preliminary runs, $k=L$ and $\ell=5k$, 
so that unbiased MCMC is at least 20\% less efficient than ergodic average MCMC.
We compare the proposed UPAVE with batch means (BM) and spectral variance (SV) estimators.
We implement them as in the \texttt{mcmcse} package \citep{mcmcse}, 
with details provided in Appendix~\ref{appx:reviewlongrunvarianceestimation}. In particular, the
batch size is selected automatically with the method of \citet{liuvatsflegal2022}, the bias is reduced using the lugsail variants proposed in \citet{vats:fleg:2022},
with lugsail parameter $r\in\{1,2,3\}$, and each estimator is based parallel chains, as proposed in \citet{agarwal2022globally,gupta2020estimating}.
We have also experimented with the \texttt{spectrum0} function of the \texttt{coda} package \citep{plummer2006coda}, which gave estimates with low bias but prohibitively large variance
in the examples considered below, and the results are not reported.

\subsection{AR(1)\label{subsec:ar}}

We consider the autoregressive process $X_{t}=\phi X_{t-1}+W_{t}$,
where $W_{t}\sim\text{Normal}(0,1)$, and $(W_t)$ are independent. We set $\phi=0.99$. The initial
distribution is $\pi_{0}=\text{Normal}(0,4^{2})$. The target distribution
is $\text{Normal}(0,(1-\phi^{2})^{-1})$, and for $h:x\mapsto x$
the asymptotic variance is analytically available: $v(P,\test)=(1-\phi)^{-2}=10^4$.
We use a reflection-maximal
coupling as described in Appendix~\ref{appx:couplingrh}.
Assumption~\ref{assu:tau-moment-kappa} holds for all $\kappa > 1$ by Proposition~\ref{prop:vgeometric}, taking
$\lambda=1+\frac{\phi}{2}$, $V(x)=\left|x\right|+1$, $b=1+\sqrt{2/\pi}$ and $\bar{v}=\frac{2b}{1-\phi}$ therein.
We also investigate quantitative bounds in Appendix~\ref{appx:AR1}.
We choose $k=500$, $L=500$, $\ell=5k$ for
unbiased MCMC. The state $y\in\mathbb{X}=\mathbb{R}$,
used to define $\fishytest=g_{y}$ in \eqref{eq:explicitfishy-1},
is set as $y=0$. 
We can calculate
that $g_y$ here is the function $x\mapsto (1-\phi)^{-1} x$.

\begin{table}[t]
\centering 
\begin{tabular}{r|l|l|l|l|l}
\hline
R & estimate & cost & fishy cost & variance of estimator & inefficiency\\
\hline
1 & [8049 - 10417] & [5234 - 5262] & [145 - 169] & [2.5e+08 - 4.9e+08] & [1.3e+12 - 2.5e+12]\\
\hline
10 & [9431 - 10264] & [6677 - 6758] & [1585 - 1666] & [4e+07 - 5.6e+07] & [2.7e+11 - 3.7e+11]\\
\hline
50 & [9742 - 10210] & [13155 - 13340] & [8055 - 8247] & [1.2e+07 - 1.5e+07] & [1.6e+11 - 2e+11]\\
\hline
100 & [9829 - 10250] & [21262 - 21578] & [16174 - 16470] & [9.2e+06 - 1.1e+07] & [1.9e+11 - 2.4e+11]\\
\hline
\end{tabular} \caption{AR(1) example: unbiased estimation of the asymptotic variance $v(P,\protect\test)$.
Here $v(P,\test)=10^4$. Each entry provides a $95\%$ confidence interval obtained from $M=10^3$ independent runs.}
\label{tab:ar1:unbiased}
\end{table}

The performance of the proposed estimator of $v(P,\test)$ is shown in Table
\ref{tab:ar1:unbiased}. 
The columns correspond to:
1) $R$: the number of atoms in each measure at which the fishy function is estimated,
2) estimate: overall estimate of $v(P,h)$, obtained by averaging $M$ independent runs,
3) cost: average cost of each run, in units of MCMC transitions,
4) fishy cost: average cost associated with fishy function estimates within each run, 5) empirical variance of the proposed estimators, and 6) inefficiency defined as
the product of variance and average cost (smaller is better).
The results are based on
$M=10^{3}$ independent replicates, 
and each entry shows a $95\%$ confidence interval obtained with the nonparametric bootstrap.
Increasing $R$ leads to a higher average cost but better efficiency. Overall we obtain accurate estimates
of $v(P,\test)$ with parallel runs that each costs of the order of $10^4$
iterations.

Next we compare the proposed estimator with the following naive strategy:
generate a chain of length $T$ (post burn-in), compute the estimate
$T^{-1}\sum_{t=0}^{T-1}h(X_t)$, repeat $M$ times independently and compute the
empirical variance of the $M$ estimates.  We set $T=13,000$ to match the cost of our
estimator with $R=50$. We find that the naive strategy has a bias (as $M\to
\infty$) equal to $-77$, and an asymptotic variance of $2\times 10^8$.  This is
about 10 times larger than the variance of the proposed estimator reported in
Table \ref{tab:ar1:unbiased}.

\begin{figure}[t]
  \centering \begin{subfigure}[b]{0.45\columnwidth} \includegraphics[width=1\columnwidth]{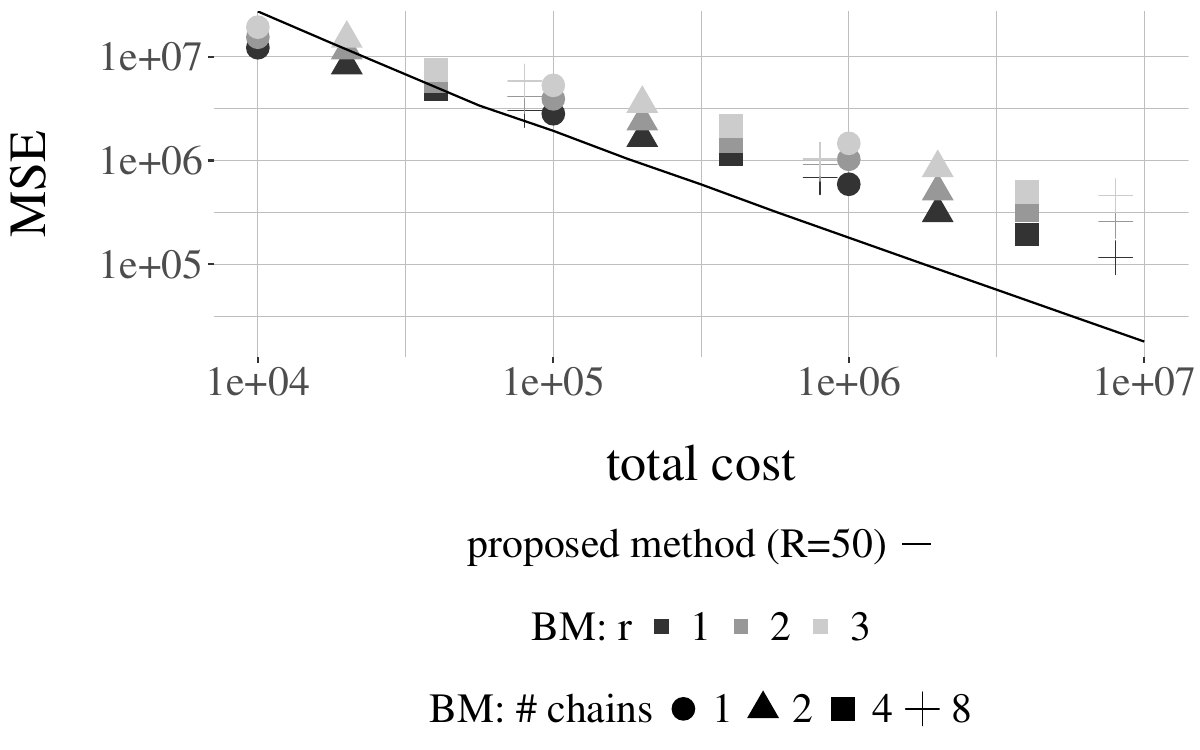}
\caption{}
\label{subfig:ar:bm} \end{subfigure} \hspace*{1cm} \begin{subfigure}[b]{0.45\columnwidth}
\includegraphics[width=1\columnwidth]{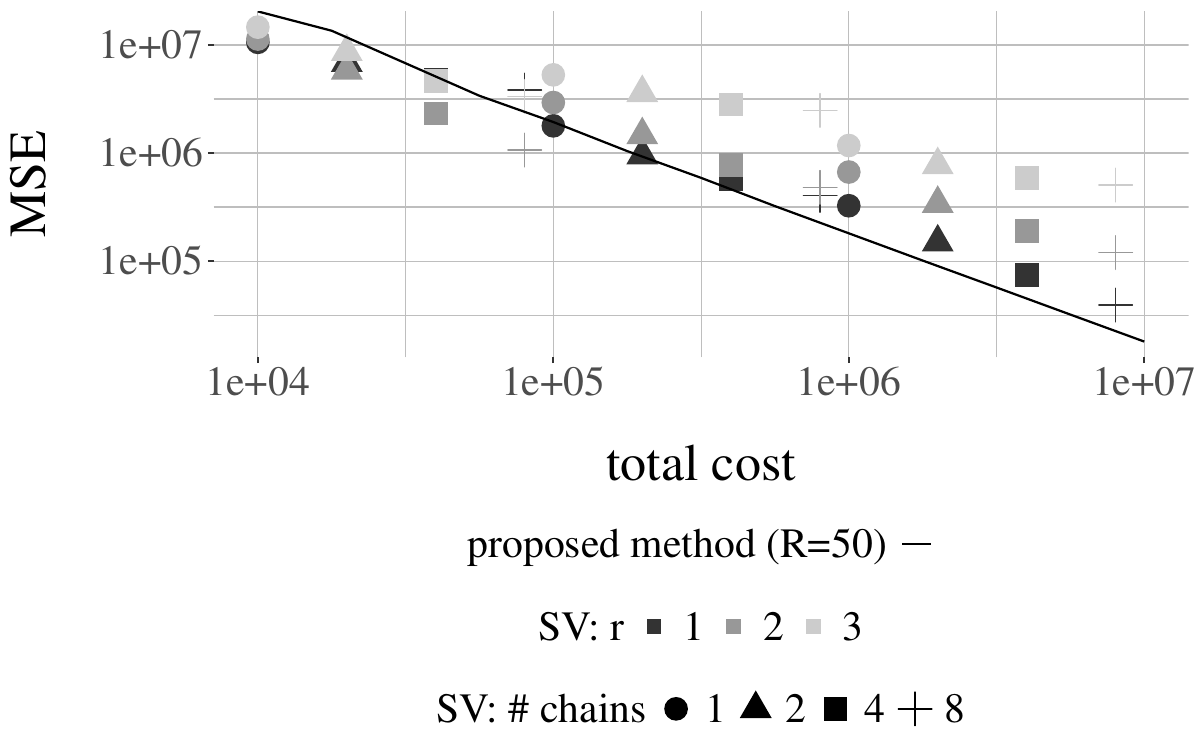}
\caption{}
\label{subfig:ar:sv} \end{subfigure}
\caption{AR(1) example: mean squared error against total cost,
  for batch means (BM, left) and spectral variance (SV, right) estimators
  of $v(P,h)$. The performance of the proposed method (with $R=50$) is indicated with a full line.\label{fig:ar:bmsv}}
\end{figure}

We continue the comparison with batch means and spectral variance estimators.
We compute these estimators of $v(P,\test)$ using chains of
lengths in $\{10^{4},10^{5},10^{6}\}$, and numbers of parallel chains
in $\{1,2,4,8\}$. We base the results in this section on $400$ independent trajectories of
length $10^6$, so for example we obtain $400$ independent estimators based on
one chain, $200$ based on two chains, etc.  For each configuration we
approximate the mean squared error (MSE), and the total cost is equal to number
of chains multiplied by the time horizon.  Finally, we report the MSE that
would be achieved by our proposed method, using $R=50$, if we generated
sequentially a number of independent replicates of UPAVE corresponding to the
given total cost. The comparison here does not account for any potential speed
up on parallel architectures. The results are shown in Figure
\ref{fig:ar:bmsv}, where both axes are on logarithmic scale. The proposed method is worse than standard
estimators when the total computing budget is low. However, we observe that the MSE of the different estimators converge
at different rates, as predicted by the theory reviewed in Appendix~\ref{appx:reviewlongrunvarianceestimation}.
The proposed estimator thus has a smaller MSE when the computing budget is large enough.

\begin{figure}[t]
  \centering \begin{subfigure}[b]{0.45\columnwidth} \includegraphics[width=1\columnwidth]{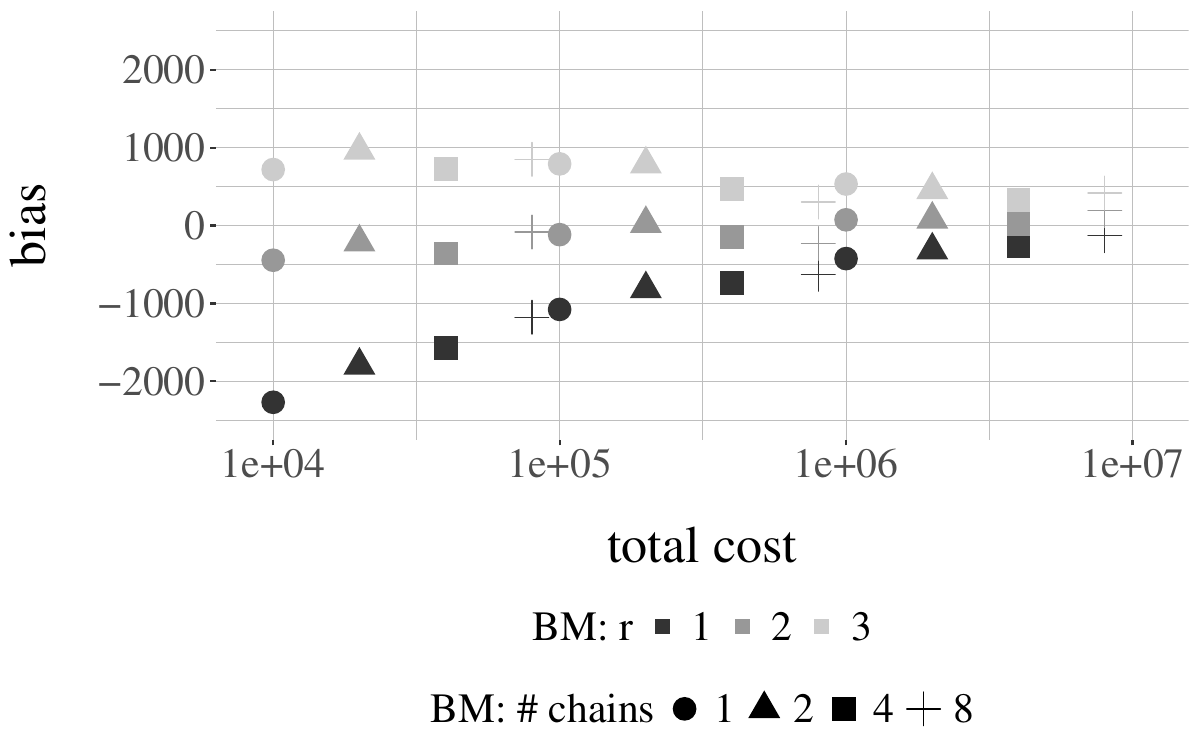}
\caption{}
\label{subfig:ar:bmbias} \end{subfigure} \hspace*{1cm} \begin{subfigure}[b]{0.45\columnwidth}
\includegraphics[width=1\columnwidth]{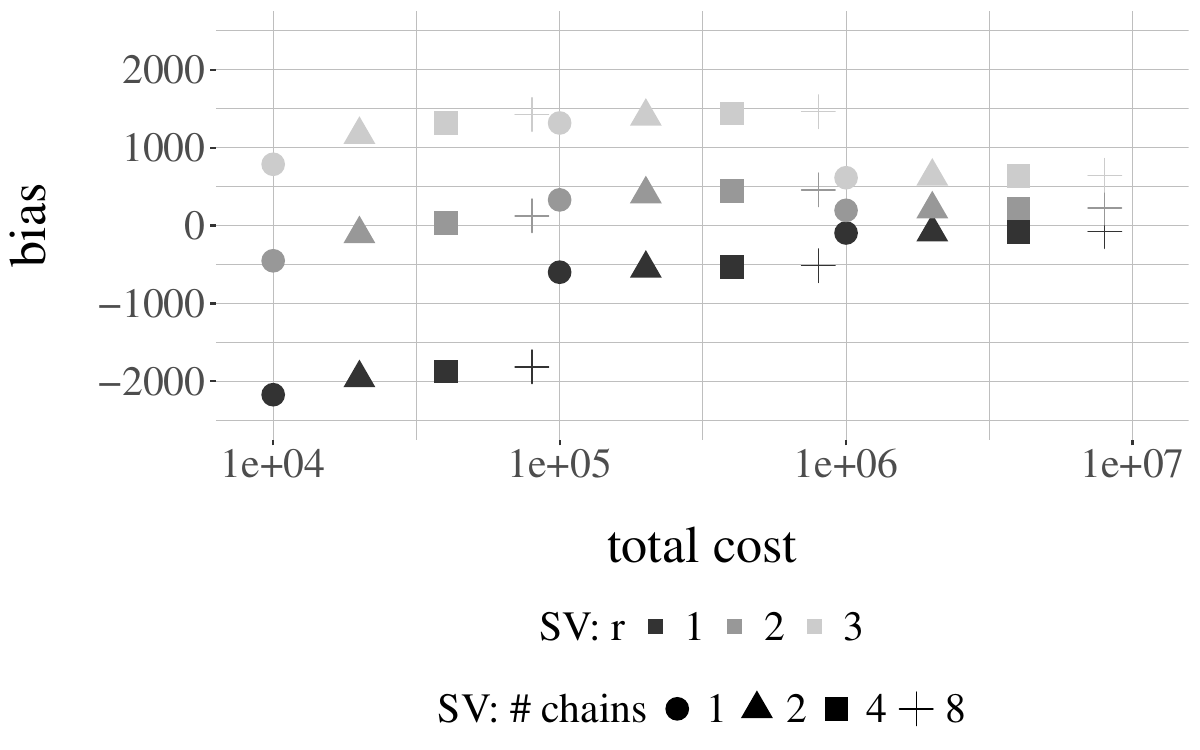}
\caption{}
\label{subfig:ar:svbias} \end{subfigure}
\caption{AR(1) example: bias against total cost,
  for batch means (BM, left) and spectral variance (SV, right) estimators
  of $v(P,h)$. The proposed method is, on the other hand, unbiased.\label{fig:ar:bmsvbias}}
\end{figure}

Finally, we produce similar plots for the bias instead of the mean squared
error, shown in Figure \ref{fig:ar:bmsvbias}. The figure includes only
batch means and spectral variance estimators since the proposed method is
unbiased by design. We see that the $r$ lugsail parameter has a strong effect on the
bias. In particular, the value $r=1$ that resulted in the smallest MSE in Figure
\ref{fig:ar:bmsv} corresponds to a noticeable negative bias.

\subsection{Hamiltonian Monte Carlo for Bayesian logistic regression\label{subsec:logitrandom}}

We consider a logistic regression with random effects \citep{rodriguez2008multilevel}, taken from pedagogical material by Germ\'an Rodr\'iguez.
Here $N$ individuals belong to $G$ groups. The function $\psi:[N]\to[G]$ maps individuals to groups, where $[N]$ denotes the set $\{1,\ldots,N\}$.
Each individual $i$ is associated with a response $Y_i$ in $\{0,1\}$ and a row vector of covariates $x_{i\cdot}\in\mathbb{R}^p$. 
The model is 
\begin{equation}\label{eq:modellogisticrandomeffect}
  \forall i\in [N]\quad Y_i \sim \text{Bernoulli with probability } 1/(1+\exp(-(\alpha + x_{i\cdot} \beta + \gamma_{\psi(i)}))),
\end{equation}
with intercept $\alpha$, regression coefficient $\beta\in\mathbb{R}^p$, and random effects $\gamma_{\psi(i)}$. Each random effect $\gamma_g$, for $g\in[G]$, is assumed $\text{Normal}(0,\sigma^2_a)$. Conditionally on the covariates, the random effects and the responses are independent.  The prior distributions are $\alpha\sim\text{Normal}(0,100^2)$, $\beta_j\sim\text{Normal}(0,100^2)$ and $\sigma_a\sim \text{Uniform}(0,10)$,
but the target distribution also integrates over the random effects $\gamma_g$ for $g\in[G]$: it is the distribution $\pi$ 
of $(\alpha,\beta,\gamma,\sigma_a)$ given the data. 
The state space is $\mathbb{R}^{2+p+G}$. We implement the model in \texttt{Stan} \citep{carpenter2017stan}.

We use data provided with the software of \citet{lillard2000aml} that describe 1060 births to 501 mothers. The response indicates whether the delivery occurred at a hospital or not.
The covariates include the logarithm of the mother's income, the distance between their home and the nearest hospital, and two binary variables indicating the mother's education level: whether it was less than ``high school'', and whether it was ``college'' or more. The births are grouped by mothers, so $\gamma_g$ is the random effect associated with the $g$-th mother. The state space is therefore of dimension 507.
We focus on the coefficient denoted by $\beta_c$ associated with the indicator of college education, and
the test function is $h:(\alpha,\beta,\gamma,\sigma_a)\to\beta_c$.

For the logistic regression model with a prior standard deviation $\sigma_a>0$, 
the target is strongly log-concave, and the gradient of the log density is Lipschitz
(see e.g. supplementary materials of \citet{heng2019unbiased}). Under these conditions, 
Theorem 2 of \citet{heng2019unbiased} states
that their coupling of a variant of Hamiltonian Monte Carlo leads to meeting times with Geometric tails, if the stepsize 
parameter and the number of leapfrog steps  are small enough. 
However, the arguments in \citet{heng2019unbiased} do not hold when $\sigma_a$ is allowed to be arbitrarily close to zero. In that case,
the formal study of the meeting times is an open question. We employ a coupling of Hamiltonian Monte Carlo described in Appendix~\ref{sec:hmc}.
To sample the momentum variables we use a diagonal mass matrix, with entries approximately equal to the inverses of the posterior variances
from 2000 samples generated by \texttt{Stan} \citep{carpenter2017stan}.  
We use a number of leapfrog steps uniformly distributed between $1$ and $20$ at each iteration,
and the stepsize is Exponential with rate $5$. The coupling strategy involves reflection-maximal couplings
when the number of leapfrog steps is equal to one, and common random numbers otherwise. We compare HMC with 
the Metropolis-adjusted Langevin Algorithm (MALA), with the same random stepsize; we do not try to optimize the stepsize for either algorithm.
The HMC algorithm employs here on average ten times more 
gradient evaluations than MALA per iteration. Yet in our implementation the elapsed real time of an HMC iteration is less than twice that of a MALA iteration.

Figure \ref{fig:randomeffectslogit} shows three independent HMC trajectories (left), a histogram of the parameter of interest $\beta_c$ (middle), 
and upper bounds on the total variation (TV) distance to stationarity for HMC and for MALA (see Appendix~\ref{appx:tvupperbounds}). From the plot we set $k=L=500,\ell=5k$ for unbiased HMC
and $k=L=10^4,\ell=5k$ for unbiased MALA.
To implement UPAVE
we need to define $\fishytest=g_{y}$. We draw $y$ (once and for all) from a Normal approximation of the posterior, obtained from 2000 samples from \texttt{Stan} \citep{carpenter2017stan}.

\begin{figure}[t]
  \centering \begin{subfigure}[b]{0.3\columnwidth} \includegraphics[width=1\columnwidth]{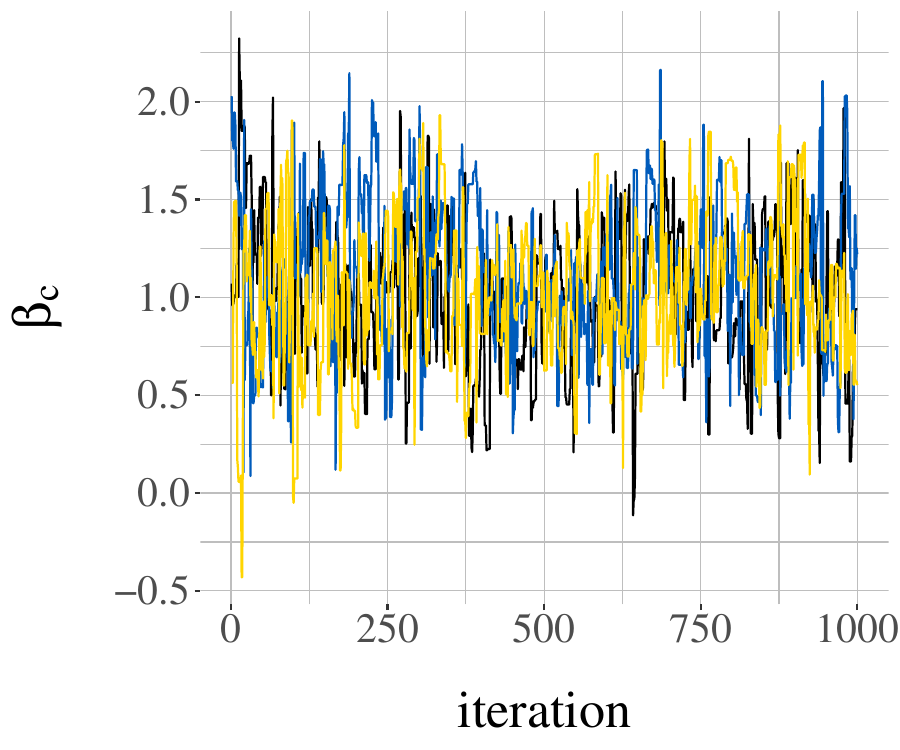}
\caption{}
\label{subfig:randomeffectslogit:trace} \end{subfigure} \hspace*{0.1cm}
\begin{subfigure}[b]{0.3\columnwidth} \includegraphics[width=1\columnwidth]{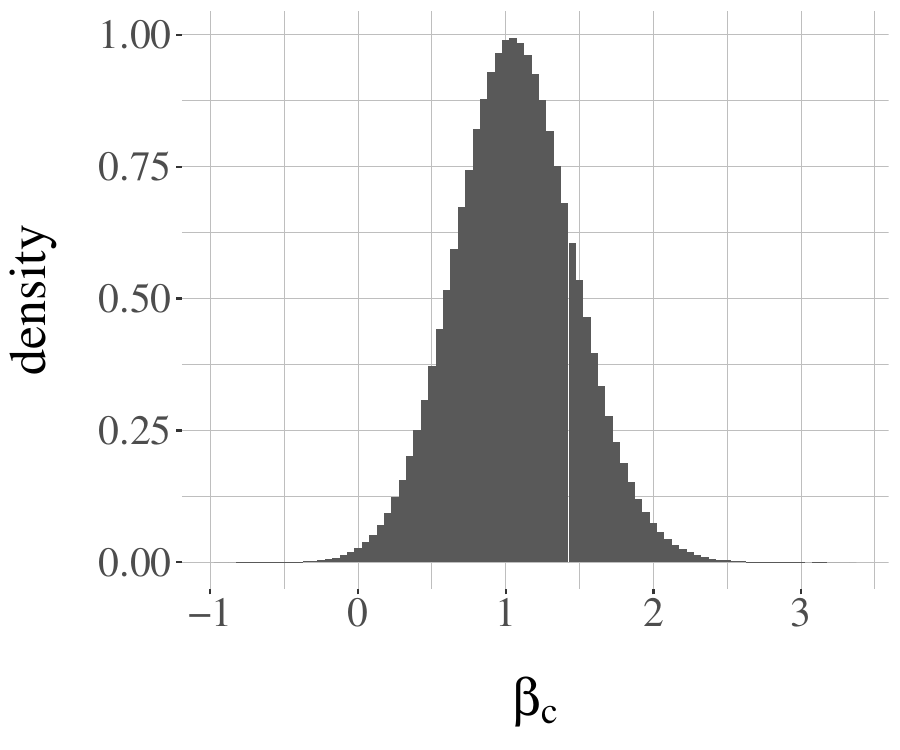}
\caption{}
\label{subfig:randomeffectslogit:hist} \end{subfigure} \hspace*{0.1cm}
\begin{subfigure}[b]{0.3\columnwidth} \includegraphics[width=1\columnwidth]{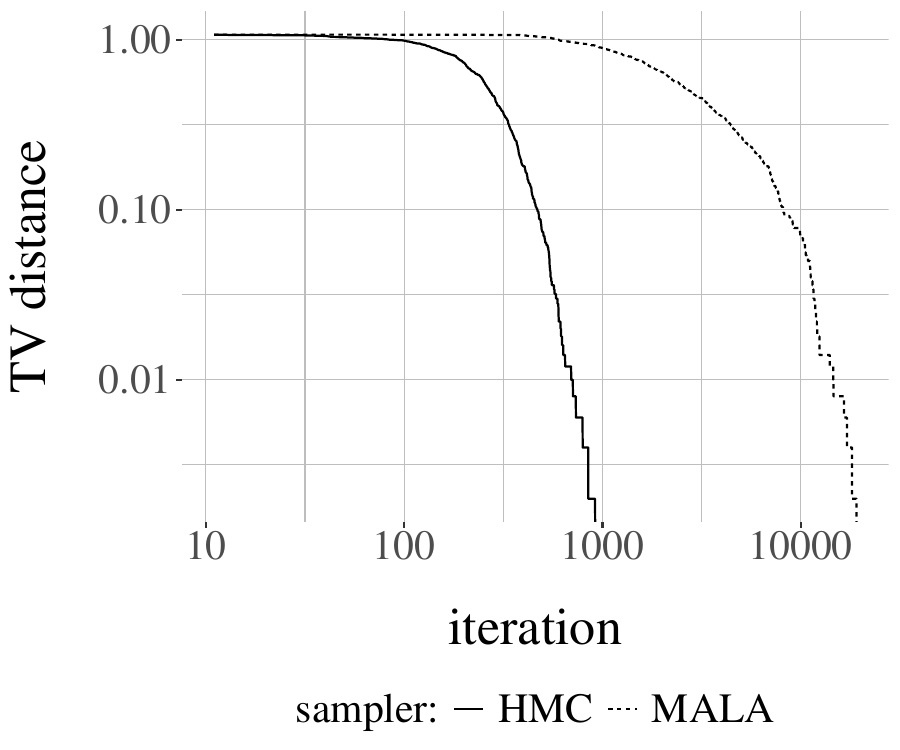}
\caption{}
\label{subfig:randomeffectslogit:tvbounds} \end{subfigure} \caption{Bayesian logistic
regression with random effects. Left: trace of the component $\beta_{c}$ of three
independent HMC chains. Middle: histogram
of $\beta_{c}$, obtained from long HMC runs. Right: upper
bounds on $|\pi_0 P^t - \pi|_{\text{TV}}$, for HMC and MALA.\label{fig:randomeffectslogit}}
\end{figure}

Table \ref{tab:randomeffectslogit:hmc:UPAVE} displays the results obtained for HMC, with $M=5000$ independent runs, and 
for different values of $R$. The average of the $M$ UPAVE estimates of $v(P,h)$ obtained for $R=20$ is $1.76$, which we take as the ground truth. 
This can be interpreted as the inefficiency of ergodic average MCMC, to be compared with the inefficiency of unbiased MCMC which we estimate at 2.47.
This amounts to a 40\% increase in inefficiency for unbiased MCMC relative to ergodic average MCMC, which is arguably a modest price to pay for unbiasedness. The user can decide to increase $k$, $L$ and $\ell$ to 
improve the relative efficiency of unbiased MCMC.

\begin{table}[t]
  \centering 
\begin{tabular}{r|l|l|l|l|l}
\hline
R & estimate & total cost & fishy cost & variance of estimator & inefficiency\\
\hline
1 & [1.57 - 2.12] & [6832 - 6857] & [1252 - 1274] & [8.23e+01 - 9.72e+01] & [5.64e+05 - 6.64e+05]\\
\hline
5 & [1.64 - 1.89] & [11857 - 11908] & [6277 - 6326] & [1.68e+01 - 1.87e+01] & [1.99e+05 - 2.22e+05]\\
\hline
10 & [1.65 - 1.81] & [18176 - 18247] & [12596 - 12666] & [8.24e+00 - 9.09e+00] & [1.5e+05 - 1.65e+05]\\
\hline
20 & [1.7 - 1.82] & [30805 - 30905] & [25223 - 25326] & [4.29e+00 - 4.66e+00] & [1.32e+05 - 1.44e+05]\\
\hline
\end{tabular} \caption{HMC for logistic regression with random effects: unbiased estimation of the asymptotic variance $v(P_{\text{HMC}},\protect\test)$.}
  \label{tab:randomeffectslogit:hmc:UPAVE}
  \end{table}

We perform similar calculations for MALA, as reported in Table \ref{tab:randomeffectslogit:mala:UPAVE}, based on $M=10^3$ independent runs. With $R=10$ we obtain an estimate of $v(P_{\text{MALA}}, \protect\test)$ around 115, which is approximately 65 times larger than that of HMC. Since HMC is not 65 times more costly per iteration, it is advantageous to employ HMC for the estimation of $\pi(h)$. 
We could similarly compare the asymptotic variances resulting from different choices of tuning parameters.

\begin{table}[t]
  \centering 
\begin{tabular}{r|l|l|l|l|l}
\hline
R & estimate & total cost & fishy cost & variance of estimator & inefficiency\\
\hline
1 & [77.5 - 127] & [127801 - 129157] & [20563 - 21913] & [1.45e+05 - 2.04e+05] & [1.86e+10 - 2.59e+10]\\
\hline
5 & [109 - 134] & [213295 - 216188] & [106069 - 108786] & [3.48e+04 - 4.56e+04] & [7.49e+09 - 9.8e+09]\\
\hline
10 & [107 - 125] & [320633 - 325076] & [213506 - 217665] & [1.95e+04 - 2.42e+04] & [6.27e+09 - 7.78e+09]\\
\hline
\end{tabular} \caption{MALA for logistic regression with random effects: unbiased estimation of the asymptotic variance $v(P_{\text{MALA}},\protect\test)$.}
  \label{tab:randomeffectslogit:mala:UPAVE}
\end{table}

We next compare UPAVE with EPAVE in the context of HMC. We run $M=100$ independent chains of length 50,000, discard a burn-in of $500$ iterations, and estimate the fishy function every $D=578$
iterations, chosen so that the cost of each EPAVE run is of the order of 100,000. The average of the $M$ EPAVE runs yields an estimate of $v(P_\text{HMC},h)$ equal to 1.72 and the variance is estimated at 1.86. 
Using the nonparametric bootstrap and the $M$ runs of EPAVE, we obtain a confidence interval [1.45e+05 - 2.4e+05] for the inefficiency of EPAVE, which is aligned with the inefficiency of UPAVE reported in Table~\ref{tab:randomeffectslogit:hmc:UPAVE}.

Finally, we compare UPAVE with batch means and spectral variance estimators. We compute these from $M=25$ independent runs, each of which involves $4$ parallel HMC chains of length $t=10^5$, with the first 1000 discarded as burn-in. Figure \ref{fig:randomeffectslogit:bmsv}
shows the resulting estimates. Comparing with the variances reported in Table \ref{tab:randomeffectslogit:hmc:UPAVE}, we notice that BM and SV estimates have much lower variability
compared to UPAVE. For example, the estimators ``SV r = 2'' with $t=10^5$ have a cost of $4\times 10^5$ and a variance of $0.002$. To match that variance, one would need to average about 2250 UPAVE runs with $R=20$, which would cost about $7\times 10^7$ units of transitions, that is 175 times more than that used to compute ``SV r = 2''.
On the other hand, Figure \ref{fig:randomeffectslogit:bmsv} illustrates that the bias of batch means and spectral variance estimators is noticeable, even in relatively long runs. While all estimates yield plausible values for $v(P_\text{HMC},h)$, many of them fall outside the $95\%$ confidence interval constructed from $5000$ runs of UPAVE, considered as the reference here and represented by horizontal dashed lines. We obtained similar results with MALA: the bias of BM and SV estimators is noticeable, and their variances are much smaller than those of UPAVE for given computing budget. 
For example, the estimator ``SV r = 2'' computed for a cost of $4\times 10^6$ iterations of MALA has the same variance as that of an average of 1000 runs of UPAVE with $R=10$, $k=L=10^4$ and $\ell=5k$, which would cost 45 times more units of Markov transitions. 

  \begin{figure}[t]
    \centering 
  \includegraphics[width=0.8\columnwidth]{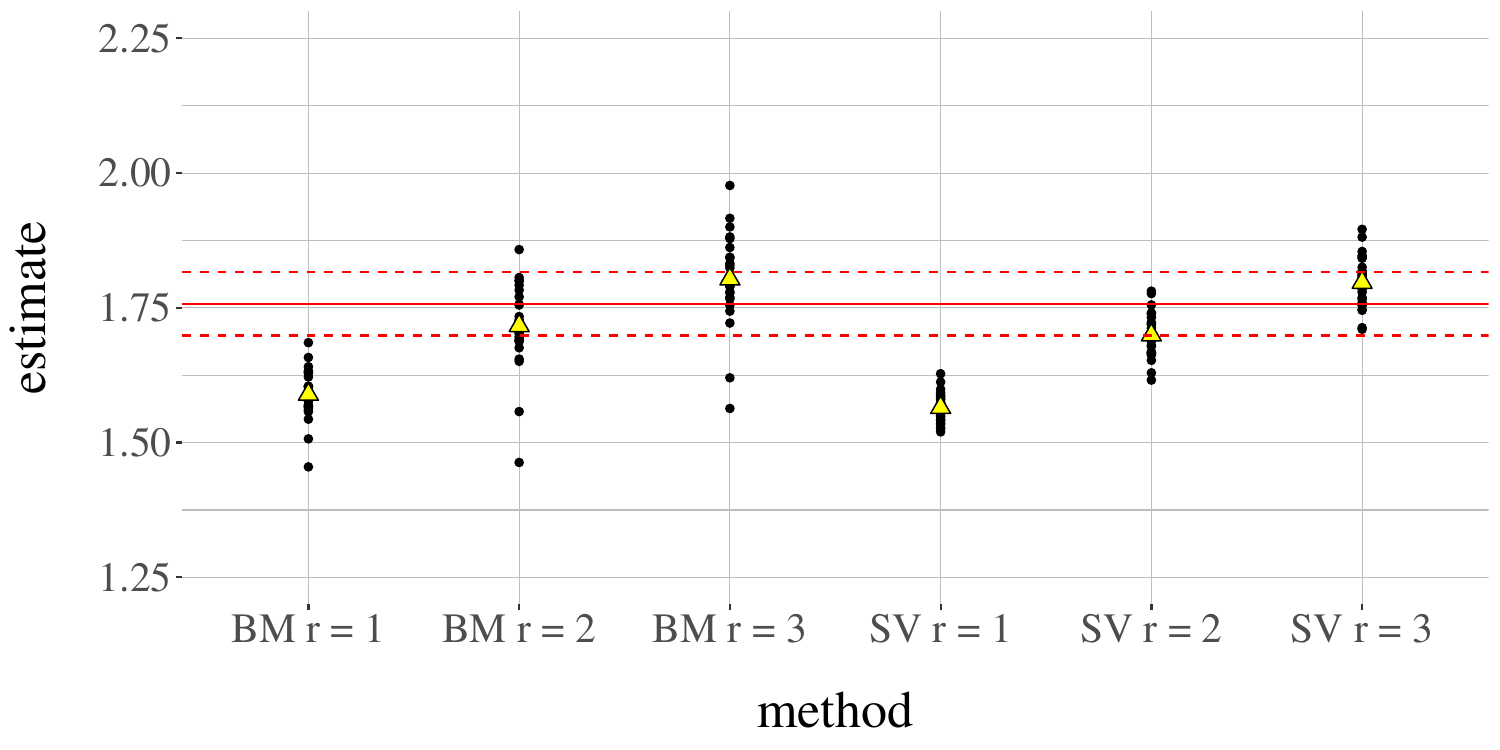}
  \caption{HMC for logistic regression with random effects: batch mean and spectral variance estimators, with tuning parameter $r\in\{1,2,3\}$ in the package \texttt{mcmcse}, obtained from runs of 4 parallel chains of length $10^5$. Each dot represents an estimate of $v(P_\text{HMC},h)$. The yellow triangles represent the means for each method. The horizontal lines represent the estimate of $v(P,h)$ and a $95\%$ confidence interval obtained from 5000 runs of UPAVE with $R=20$. \label{fig:randomeffectslogit:bmsv}}
  \end{figure}

\section{Discussion\label{sec:discussion}}

Our contributions are grounded in the coupling approach that is implementable
in a variety of MCMC settings \citep{atchade2024unbiased}.
The proposed estimators of fishy functions in Section \ref{sec:coupledchainsandfishyfunctions}
could be considered for control variates.
Indeed, approximating solutions of
the Poisson equation is a well-known strategy for variance
reduction via control variates
\citep{andradottir1993variance,henderson1997variance,dellaportas2012control,mijatovic2018poisson,alexopoulos2022variance}.
This involves replacing $h$ by $h-\phi$
in an ergodic average, where $\phi\in L_{0}^{1}(\pi)=\{f \in
L^{1}(\pi):\pi(f)=0\}$, so that $\pi(h-\phi)=\pi(h)$
and the limit of the MCMC estimator is unchanged, but the variance
may be smaller with a judicious choice of $\phi$. A convenient family
of $\phi$ is $\{(I-P)f:f\in L^{1}(\pi)\}$ since $\phi\in L_{0}^{1}(\pi)$
by construction and indeed an optimal choice of $f$ is a fishy function.
Our approach may be the first widely-applicable strategy for
consistent estimation of fishy functions in MCMC settings.

Regeneration methods may provide a viable alternative to couplings.
When a Markov chain admits an accessible atom $\alpha$, a solution
of the Poisson equation \citep[e.g.][]{glynn1996liapounov} is $g(x)=\mathbb{E}_{x}\left[\sum_{k=0}^{\sigma_{\alpha}}h_{0}(X_{k})\right]$,
where $\sigma_{\alpha}=\inf\{n\geq0:X_{n}\in\alpha\}$. This allows
approximation of $g$ pointwise by simulation if one can identify
entries into $\alpha$ and one can approximate $h_{0}=h-\pi(h)$ pointwise.
Proper atoms may not exist for a given general state space Markov
chain, and identification of hitting times of an atom for a suitable
split chain as in \citet{mykland1995regeneration} is not always feasible.
While it is often possible to define a modified Markov chain that
admits an easily identified, accessible atom \citep{brockwell:kadane:05,leedoucetperfectsimulation},
the corresponding solution of the Poisson equation may not be similar
to that of the original chain. Note also that, when atoms can be identified,
one would often use regenerative simulation to approximate the asymptotic
variance \citep{hobert2002applicability}, which can be expressed
as $v(P,h)=\pi(\alpha)\mathbb{E}_{\alpha}[\{ \sum_{k=1}^{\tau_{\alpha}}h_{0}(X_{k})\} ^{2}]$
\citep{bednorz2008regeneration}.

With their reliance on couplings, the proposed estimators of the asymptotic variance $v(P,h)$ in Section \ref{sec:Asymptotic-variance}
require more from MCMC users than batch means or spectral variance estimators, 
but have distinctive features: UPAVE is unbiased, and both EPAVE and UPAVE converge at the Monte Carlo rate under realistic conditions.
The experiments in Section \ref{sec:numerical:experiments}
indicate that they can be implemented in non-trivial settings and display promising performance.
They tend to have a large variance, which, for UPAVE, can be arbitrarily reduced by averaging
over more independent copies.
The lack of bias is in contrast to the somewhat unpredictable bias of existing asymptotic variance estimators,
which stands as a long-standing issue. 
The proposed methods appear useful in settings where accurate estimates of $v(P,h)$ are sought, possibly with confidence intervals, and when parallel processors are available.

As with other works on unbiased MCMC \citep{atchade2024unbiased}, it is worth emphasizing that
the performance depends on the underlying MCMC algorithm,
its initialization and its coupling. We refer to the bimodal target
in Section 5.1 of \citet{joa2020} for a situation where multimodality
in the target distribution combined with a poor choice of MCMC
sampler gives misleading estimates, despite the
lack of bias and finite variance. In our theoretical results,
we have prioritized assumptions on the moments of meeting times of
the coupled Markov chains under strong but reasonable initialization
assumptions, which can be cleanly separated from assumptions on the
moments of functions. Given this emphasis, the results appear to be
fairly strong and provide a sensible relationship between the moments
of the meeting time and moments of the proposed estimators.

The appendices are structured as follows:
\begin{itemize}
\item[Appendix~\ref{appx:unbiasedmcmcwithcouplings}] provides reminders on unbiased MCMC with couplings.
\item[Appendix~\ref{sec:theory}] contains the proofs of our results.
\item[Appendix~\ref{appx:reviewlongrunvarianceestimation}] describes some existing results on the convergence of batch means and spectral variance estimators. 
\item[Appendix~\ref{appx:AR1}] verifies Assumption~\ref{assu:tau-moment-kappa} quantitatively for the AR(1) process considered in Section~\ref{subsec:ar}.
\item[Appendix~\ref{sec:hmc}] describes couplings of Hamiltonian Monte Carlo employed in Section~\ref{subsec:logitrandom}.
\item[Appendix~\ref{appx:highdimreg}] provides experiments with a Gibbs sampler for high-dimensional regression with shrinkage prior.
\item[Appendix~\ref{subsec:pmmh}] provides experiments with a particle marginal Metropolis--Hastings algorithm for parameter inference
in a nonlinear state space model.
\item[Appendix~\ref{appx:twocompetingmcmc}] contains additional experiments in the Cauchy location model introduced in Section~\ref{subsec:fishyillustration}.
\end{itemize}
Code to reproduce the figures
of this article can be found at: \url{https://github.com/pierrejacob/unbiasedpoisson/}.

\section*{Acknowledgements}
Research of AL supported by EPSRC grants EP/R034710/1 and EP/Y028783/1. Research of DV supported by SERB (SPG/2021/001322).
The authors thank Mohamed Keteb for finding a typo in an earlier version of the manuscript.

\bibliographystyle{agsm}
\bibliography{../../Biblio}

\appendix

\section{Unbiased MCMC with couplings\label{appx:unbiasedmcmcwithcouplings}}

\subsection{Couplings of MCMC algorithms\label{appx:couplingrh}}

We recall some implementable
coupling techniques for MCMC. We focus on Metropolis--Rosenbluth--Teller--Hastings
(MRTH) algorithms \citep{hastings:1970}.
The goal is to couple generated trajectories such that exact meetings
can occur. Relevant considerations can be found in \citet{joa2020,wang2021maximal,PappSherlock2024}.

Algorithm \ref{alg:coupledmh} describes a simple way of coupling
$P(x,\cdot)$ and $P(y,\cdot)$ where $P$ is the transition associated
with MRTH, with proposal transition $q$ and acceptance rate $\alpha_{\text{MRTH}}$.
In the algorithmic description, a maximal coupling of two distributions $\mu$ and $\nu$ for random variables $X$ and $Y$, respectively,
refers to a joint distribution $\gamma$ for $(X,Y)$ with marginals $\mu$
and $\nu$, and such that $\mathbb{P}(X=Y)$ is maximized over all
such joint distributions.

\begin{algorithm}
\begin{enumerate}
\item Sample $(X^{\star},Y^{\star})$ from a maximal coupling of $q(x,\cdot)$
and $q(y,\cdot)$.
\item Sample $U\sim\text{Uniform}(0,1)$.
\item If $U<\alpha_{\text{MRTH}}(x,X^{\star})$, set $X'=X^{\star}$, otherwise
$X'=x$.
\item If $U<\alpha_{\text{MRTH}}(y,Y^{\star})$, set $Y'=Y^{\star}$, otherwise
$Y'=y$.
\item Return $(X',Y')$.
\end{enumerate}
\caption{A coupled Metropolis--Rosenbluth--Teller--Hastings kernel \label{alg:coupledmh}}
\end{algorithm}

Next we provide details on how to sample
from maximal couplings of $q(x,\cdot)$ and $q(y,\cdot)$. A possibility,
mentioned in Section 4.5 of \citet{thorisson2000coupling}, and called
$\gamma$-coupling in \citet{johnson1998coupling}, is described in
Algorithm~\ref{alg:maximalcoupling}. The algorithm requires the
ability to sample from $\mu$ and $\nu$. The cost of executing
Algorithm~\ref{alg:maximalcoupling} is random, its expectation is
independent of $\mu$ and $\nu$, and its variance goes to infinity
as $|\mu-\nu|_{\text{TV}}\to0$. A variant of the algorithm for which
the variance is bounded for all $\mu,\nu$ is described in \citet{GerberLee2020}.

\begin{algorithm}
\begin{enumerate}
\item Sample $X\sim\mu$.
\item Sample $W\sim\text{Uniform}(0,1)$.
\item If $W\leq\nu(X)/\mu(X)$, set $Y=X$.
\item Otherwise, sample $Y^{\star}\sim\nu$ and $W^{\star}\sim\text{Uniform}(0,1)$
\\
 until $W^{\star}>\mu(Y^{\star})/\nu(Y^{\star})$, and set $Y=Y^{\star}$.
\item Return $(X,Y)$.
\end{enumerate}
\caption{Sampling from a maximal coupling of $\mu$ and $\nu$. \label{alg:maximalcoupling}}
\end{algorithm}

When $\mu$ and $\nu$ are Normal distributions with the same variance,
an alternative maximal coupling procedure is described in Algorithm~\ref{alg:reflmax}; it was proposed in \citet{bou2020coupling}. Its
cost is deterministic and independent of $\mu$ and $\nu$. In the
algorithmic description, $\varphi$ refers to the probability density
function of the standard Normal distribution. We use
Algorithm~\ref{alg:reflmax} to 
couple the Metropolis-adjusted Langevin algorithm
in Section~\ref{subsec:logitrandom} and Appendix~\ref{sec:hmc}.

\begin{algorithm}
\begin{enumerate}
\item Let $z=\Sigma^{-1/2}(\mu_{1}-\mu_{2})$ and $e=z/|z|$.
\item Sample $\dot{X}\sim\mathcal{N}(0_{d},I_{d})$, and $W\sim\text{Uniform}(0,1)$.
\item If $\varphi(\dot{X})W\leq\varphi(\dot{X}+z)$, set $\dot{Y}=\dot{X}+z$;
else set $\dot{Y}=\dot{X}-2(e^{T}\dot{X})e$.
\item Set $X=\Sigma^{1/2}\dot{X}+\mu_{1},Y=\Sigma^{1/2}\dot{Y}+\mu_{2}$,
and return $(X,Y)$.
\end{enumerate}
\caption{\label{alg:reflmax} Reflection-maximal coupling of $\text{Normal}(\mu_{1},\Sigma)$
and $\text{Normal}(\mu_{2},\Sigma)$.}
\end{algorithm}

In the case of univariate Normal distributions $\text{Normal}(\mu_{1},\sigma^{2})$
and $\text{Normal}(\mu_{2},\sigma^{2})$, the procedure simplifies
to Algorithm \ref{alg:reflmaxunivariate}. We use Algorithm~\ref{alg:reflmaxunivariate}
to couple AR(1) processes in Section~\ref{subsec:ar}.

\begin{algorithm}
\begin{enumerate}
\item Let $z=(\mu_{1}-\mu_{2})/\sigma$.
\item Sample $\dot{X}\sim\text{Normal}(0,1)$, and $W\sim\text{Uniform}(0,1)$.
\item Set $X=\mu_{1}+\sigma\dot{X}$.
\item If $W<\varphi(z+\dot{X})/\varphi(\dot{X})$, set $Y=X$; else set
$Y=\mu_{2}-\sigma\dot{X}$.
\item Return $(X,Y)$.
\end{enumerate}
\caption{\label{alg:reflmaxunivariate} Reflection-maximal coupling of (univariate)
$\text{Normal}(\mu_{1},\sigma^{2})$ and $\text{Normal}(\mu_{2},\sigma^{2})$.}
\end{algorithm}

\subsection{Unbiased estimators of expectations and signed measures\label{appx:unbiasedmcmc}}

\citet{glynn2014exact} show how coupled Markov chains can be employed
to construct unbiased estimators of stationary expectations. We recall here the variations presented in \citet{joa2020,VanettiDoucet2020}, that rely on a coupled Markov kernel $\bar{P}$ that induces
meetings of the two chains. This removes the need to specify a truncation variable
as in \citet{glynn2014exact,agapiou2018unbiased}.

Specifically, we construct the chains with a time lag $L\in\mathbb{N}$,
which is a tuning parameter. The construction is as in Algorithm \ref{alg:coupledchains}.
We also introduce a ``starting time'' $k\in\mathbb{N}$, a ``prospective
end time'' $\ell$, with $k\leq \ell$, which are tuning parameters. Under
assumptions provided in Section~\ref{subsec:Unbiased-approx-pi-test},
the following
random variable is an unbiased estimator of $\pi(\test)$:
\begin{equation}
H_{k}=\test(X_{k})+\sum_{j=1}^{\infty}\{\test(X_{k+jL})-\test(Y_{k+(j-1)L})\}.\label{eq:H_k^L}
\end{equation}
Consider a range of integers $k,\ldots,\ell$ for $k\leq \ell$ and associated
estimators $H_{k},\ldots,H_{\ell}$ obtained from the same
trajectories $(X_{t},Y_{t})_{t\geq0}$. All of these estimators $(H_{t})_{t=k}^{\ell}$
are unbiased, so their average is unbiased and reads
\begin{align*}
H_{k:\ell} & =\frac{1}{\ell-k+1}\sum_{t=k}^{\ell}H_{t}\\
 & =\frac{1}{\ell-k+1}\sum_{t=k}^{\ell}\test(X_{t})+\frac{1}{\ell-k+1}\sum_{s=k}^{\ell}\sum_{j=1}^{\infty}\{\test(X_{s+jL})-\test(Y_{s+(j-1)L})\}.
\end{align*}
We can find a simpler representation for the double sum in the above equation. Denote by $v_t$ the number of times that the term $\Delta_t = \test(X_t) - \test(Y_{t-L})$
appears in the double sum.
Then $v_t$ is the number of terms of the form $s+jL$
equal to
$t$ as $s$ moves in $\{k,\ldots,\ell\}$ and $j\geq 1$. We can focus on $t$ in $\{k+L,\ldots,\tau - 1\}$ since $v_t=0$ for $t$ outside of that range.
Note that, for a given $s$, there can be at most one value of
$j$ such that $s + jL = t$. So
$$v_t = |\{s \in \{k,\ldots,\ell\}: \exists j \geq 1: s+jL=t \}|.$$
We can re-write this as
$$v_t = |\{n \in \{t-\ell,\ldots,t-k\}: \frac{n}{L}\in \mathbb{Z}^+ \}|,$$
where $\mathbb{Z}^+$ is the set of positive integers. In other words we are
counting the positive multiples of $L$ within the range $\{t-\ell,\ldots,t-k\}$,
for any $t\geq k+L$. We can restrict that range to
$\{\max(L,t-\ell),\ldots,t-k\}$, since we cannot find a positive multiple of $L$
smaller than $L$. Now the range is between two positive integers. This yields:
\begin{equation}
v_t = \lfloor(t-k) / L\rfloor - \lceil \max(L, t-\ell)/L\rceil + 1.
  \label{eq:weightdelta}
\end{equation}
Indeed for two positive integers $a\leq b$, the number of multiples of $L$ within $\{a,\ldots,b\}$ is $\lfloor b/L\rfloor - \lceil a/L\rceil + 1$.
Thus we obtain the \emph{unbiased MCMC} estimator 
\begin{equation}
H_{k:\ell}=\frac{1}{\ell-k+1}\sum_{t=k}^{\ell}\test(X_{t})+\text{BC}_{k:\ell},\label{eq:H_kell^L}
\end{equation}
where $\text{BC}_{k:\ell}$ refers to a ``bias cancellation'' term,
\begin{equation}
\text{BC}_{k:\ell}=\sum_{t=k+L}^{\tau-1}\frac{v_t}{\ell-k+1}\left\{ \test(X_{t})-\test(Y_{t-L})\right\} .\label{eq:BC_kell^L}
\end{equation}
Recall that $\tau=\inf\{t>L:X_t = Y_{t-L}\}$, for chains generated e.g. by Algorithm~\ref{alg:coupledchains}.
Instead of the above expressions that involve a test function $\test$,
we can view \emph{unbiased MCMC} as providing the signed measure
\begin{equation}
\hat{\pi}({\rm d}x)=\frac{1}{\ell-k+1}\sum_{t=k}^{\ell}\delta_{X_{t}}({\rm d}x)+\sum_{t=k+L}^{\tau-1}\frac{v_t}{\ell-k+1}\left\{ \delta_{X_{t}}-\delta_{Y_{t-L}}\right\} ({\rm d}x),\label{eq:pihat_kell^L}
\end{equation}
as an unbiased approximation of $\pi$. 

If we count the cost of sampling from the kernel $P$ as one unit,
and the cost of sampling from $\bar{P}$ as one unit if the chains
have met already, and two units otherwise, then the random cost of
obtaining \eqref{eq:H_kell^L}, or \eqref{eq:pihat_kell^L}, equals $\max(L,\ell+L-\tau)+2(\tau-L)$
units.

The above unbiased estimators are to be generated independently in
parallel, and averaged to obtain a final approximation of $\pi$.
Since they are unbiased, the mean squared error is equal to the variance.
To compare unbiased estimators with different cost, e.g. to compare
different configurations of the tuning parameters $k,\ell,L$, we can
compute the asymptotic inefficiency defined as the expected cost multiplied
by the variance, as described in \citet{glynn1992asymptotic}; the
lower value, the better.

\subsection{Upper bounds on the distance to stationarity\label{appx:tvupperbounds}}

A by-product of unbiased MCMC is an upper bound on the total variation
distance between the marginal distribution of the chain at time $t$ and the stationary distribution $|\pi_{t}-\pi|_{\text{TV}}$
\citep{joa2020,biswas2019estimating,craiu2022double}, given by
\begin{equation}
|\pi_{t}-\pi|_{\text{TV}}\leq\mathbb{E}\left[\max\left(0,\left\lceil \frac{\tau-L-t}{L}\right\rceil \right)\right].\label{eq:upperboundtv}
\end{equation}
The upper bound can be estimated for any $t\geq 0$, using independent replications of
the meeting time $\tau=\inf\{t>L:X_t=Y_{t-L}\}$ obtained by running coupled chains with a lag $L$.

\section{Theoretical results\label{sec:theory}}

    \setcounter{thm}{0}
    \renewcommand{\thethm}{\Alph{section}.\arabic{thm}}

\subsection{\label{subsec:ass-meeting-times}Assumption on the meeting time}
We first prove Proposition~\ref{prop:tau-moment-survival}, which
can be used to verify Assumption~\ref{assu:tau-moment-kappa}.
\begin{proof}[Proof of Proposition~\ref{prop:tau-moment-survival}]
The first part follows by Markov's inequality:
\[
\mathbb{P}_{x,y}(\tau>t)=\mathbb{P}_{x,y}\left[\tau^{\kappa}\geq(t+1)^{\kappa}\right]\leq\mathbb{E}_{x,y}[\tau^{\kappa}](t+1)^{-\kappa}.
\]
For the second part, we have $\mathbb{P}_{x,y}(\tau<\infty)=1$. Using
Tonelli's theorem,
\begin{align*}
\mathbb{E}_{\pi\otimes\pi}\left[\tau^{\kappa}\right] & =\mathbb{E}_{\pi\otimes\pi}\left[\int_{0}^{\tau}\kappa u^{\kappa-1}{\rm d}u\right]\\
 & =\mathbb{E}_{\pi\otimes\pi}\left[\int_{0}^{\infty}\kappa u^{\kappa-1}{\mathds{1}}_{(0,\tau)}(u){\rm d}u\right]\\
 & =\int_{0}^{\infty}\kappa u^{\kappa-1}\mathbb{P}_{\pi\otimes\pi}(\tau>u){\rm d}u\\
 & =\sum_{i=0}^{\infty}\mathbb{P}_{\pi\otimes\pi}(\tau>i)\int_{i}^{i+1}\kappa u^{\kappa-1}{\rm d}u\\
 & \leq\sum_{i=0}^{\infty}\mathbb{P}_{\pi\otimes\pi}(\tau>i)\kappa(i+1)^{\kappa-1}\\
 & \leq\pi\otimes\pi(\tilde{C})\kappa\sum_{i=0}^{\infty}(i+1)^{\kappa-s-1},
\end{align*}
which is finite since $s>\kappa$, and we conclude.
\end{proof}

\begin{proof}[Proof of Proposition \ref{prop:subgeometric}]
  Applying Proposition \ref{prop:subgeometric:general} below to $\psi (v) =\phi (v) =\vartheta v^{\alpha}$, there exists a constant $\beta$ such that for all $x,y \in {\mathcal X}$ and all $n\in {\mathbb N}$,
  $$
      {\mathbb P}_{x,y} (\tau>n) \leq  \beta  \frac{V^{\alpha}(x)+V^{\alpha}(y)}{(1+n)^{\alpha/(1-\alpha)}}.
      $$
In particular, with this $\phi$ and any $\sigma \in (0,1)$, denoting
$H_{\phi}(t)=\int_{1}^{t} \frac{1}{\phi(u)} {\rm d}u$,
there exists $c_\sigma>0$ such that $H_\phi^{-1}(\sigma n) \geq c_\sigma H_\phi^{-1}(n)$
and concavity and non-negativity of $\phi$ allows us to conclude that $\phi \circ H_\phi^{-1}(\sigma n) \geq c_\sigma \, \phi \circ H_\phi^{-1}(n)$.
Finally, there exists $c'_{\vartheta,\alpha} > 0$ such that $\phi \circ H_\phi^{-1}(n) \geq c'_{\vartheta,\alpha} (1+n)^{\frac{\alpha}{1-\alpha }}$.
According to Proposition 4.3.2 in \citet{douc2018MarkovChains}, the drift condition $P V(x) \leq V(x) -\vartheta V^{\alpha}(x)+ b\mathds{1}(x \in \mathcal{C})$ implies $ \pi(V^{\alpha})<\infty$, which allows to apply Proposition \ref{prop:tau-moment-survival} and the proof is concluded. 
\end{proof}

\begin{prop} \label{prop:subgeometric:general}
Assume that there exist a measurable function 
     $V:\;\mathbb{X}\to [1,\infty)$, a concave, increasing function $\phi:\;[1,\infty)\to (0,\infty)$ continuously differentiable on $(0,\infty)$ such that $\lim_{v \to \infty} \phi' (v) =0$, $b<\infty$ and a small set $\mathcal{C}$ such that we have $\sup_{\mathcal C} V<\infty$ and $d:=\inf_{{\mathcal C}^\complement} \phi \circ V>b$ and
      for all $x\in \mathbb{X}$,
      \begin{equation} \label{eq:drift:subgeo:general}
        P V(x) \leq V(x) -\phi \circ V(x)+ b\mathds{1}(x \in \mathcal{C}).
      \end{equation}
      Assume in addition that Assumption \ref{assumption:meetingprobafromC} holds with the same small set ${\mathcal C}$.
      Then, for any increasing concave function $\psi:\;[1,\infty)\to (0,\infty)$ and any $\sigma \in (0,1-b/d)$,
      there exists a constant $\beta$ such that for all $x,y \in {\mathcal X}$ and all $n \in {\mathbb N}$, 
      $$
      {\mathbb P}_{x,y} (\tau>n) \leq  \beta \frac{\psi\circ V(x)+\psi \circ V(y)}{\psi \circ H_{\phi}^{-1} (\sigma n) }, 
      $$
      where $H_{\phi}(t)=\int_{1}^{t} \frac{1}{\phi(u)} {\rm d}u$. 
    \end{prop}
    In contrast with other similar (but not identical) existing results, our bounds provide, through the function $\psi$, a tradeoff between the rate $\psi \circ H_{\phi}^{-1} (n)$ and the dependence on the initial condition $\psi\circ V(x)+\psi \circ V(y)$. It allows to integrate the bound with respect to $\pi$, provided  that $\pi(\psi \circ V)<\infty$. For example, choosing $\psi=\phi$, and noting that \eqref{eq:drift:subgeo:general} implies $\pi(\phi \circ V)<\infty$, we can directly obtain for all $n\in {\mathbb N}$, 
    \[
\int \pi({\rm d}x) \pi({\rm d}y) {\mathbb P}_{x,y} (\tau>n) \leq  \frac{\beta_{0}}{\phi \circ H_{\phi}^{-1} (\sigma n) },
\]
for some constant $\beta_{0}$.

    \begin{proof}[Proof of Proposition \ref{prop:subgeometric:general}]
      In this proof, for any function $\phi$, we use the notation $H_{\phi}(t)=\int_{1}^{t} \frac{1}{\phi(u)} {\rm d}u$ and $r_{\phi}(t)=\phi \circ H_{\phi}^{-1} (t)$. 
      By Lemma 19.5.3 in \citet{douc2018MarkovChains}, for any $\sigma\in (0,1-b/d)$, $\bar P \bar V \leq \bar V - \bar \phi \circ \bar V + \bar b \mathds{1}_{\bar {\mathcal C}}$ where $\bar \phi = \sigma \phi$, $\bar b=2b$, $\bar V(x,y)= V(x) + V(y)-1$ and $\bar {\mathcal C}={\mathcal C} \times {\mathcal C}$.

      Moreover, applying Proposition 16.1.11 in \citet{douc2018MarkovChains}, and setting
      $V_{k}=H_{\bar \phi}^{-1} (H_{\bar \phi} \circ \bar V+k)- H_{\bar \phi}^{-1} (k)$, we get 
$$
\bar P V_{k+1} + r_{\bar \phi} (k) \leq V_{k} + b' r_{\bar \phi}(k) \mathds{1}_{\bar {\mathcal C}},
$$
where $b'=\bar b r_{\bar{\phi}}(1)/r_{\bar{\phi}}^{2} (0)$. Set $r^{0}_{\bar\phi}(n)=\sum_{j=0}^{n}r_{\bar\phi}(j)$ and define $M=\sup_{k} -r_{\bar \phi}^{0} (k-1) + b' r_{\bar\phi} (k)/\epsilon$, which is finite since $r_{\bar\phi} (k)/r^{0} _{\bar \phi} (k-1) \to 0$ as $k$ tends to infinity.

Define $S_{k}=V_{k}+r^{0}_{\bar \phi} (k-1)+M $. We first show for any $x,y \in {\mathcal X}$, 
 \begin{equation}
  \label{eq:subgeom:one}
  \bar P[\dif S_{k+1}] (x,y) \leq [\dif S_{k}](x,y),
\end{equation}
where $\dif(x,y)=\mathds{1}(x\neq y)$. The coupling kernel $\bar P$ is assumed to be sticky, i.e. $\bar P \dif (x,y)=0$ if $\dif (x,y)=0$. Hence, if $x=y$, the l.h.s. and r.h.s of \eqref{eq:subgeom:one} are null. We now assume that $x\neq y$. Using $\dif  \leq 1$, we get
\begin{align*}
  \bar P [\dif S_{k+1}] (x,y) & \leq \bar P \lrb{V_{k+1} + (r^{0}_{\bar\phi} (k) + M)\dif} (x,y)\\
                              & \leq \bar P V_{k+1} (x,y)+ r_{\bar\phi} (k)+ \lr{r^{0}_{\bar \phi} (k-1) +M } \bar P \dif (x,y) \\
                              &\leq \lrcb{V_{k}+b' r_{\bar \phi} (k) \mathds{1}_{\bar {\mathcal C}} + \lr{r^{0}_{\bar \phi} (k-1) +M} \lr{1 - \epsilon \mathds{1}_{\bar {\mathcal C}}} } (x,y) \\
  & \leq \lrcb{S_{k} + \lr{b' r_{\bar \phi} (k) - \epsilon (r^{0}_{\bar \phi} (k-1) + M)} \mathds{1}_{\bar {\mathcal C}} } (x,y) \leq S_{k}(x,y),
\end{align*}
where we have used Assumption~\ref{assumption:meetingprobafromC}
and the fact that $b' r_{\bar \phi} (k) - \epsilon (r^{0}_{\bar \phi} (k-1) + M) \leq 0$ by definition of $M$. Since $x \neq y$, we have $\dif (x,y)=1$ and hence $  \bar P[\dif S_{k+1}] (x,y) \leq [\dif S_{k}](x,y)$. At this point, \eqref{eq:subgeom:one} is proved for all $x,y \in {\mathcal C}$. Now write
\begin{align*}
  r^{0}_{\bar \phi} (n-1) \bar P^{n} \dif &\leq \bar P^{n} (\dif S_{n}) \leq \bar P^{n-1} (\dif S_{n-1}) \leq \ldots \leq \dif S_{0}\\
  & \leq \dif \lr{\bar V + M} \leq (1+M)\dif \bar V.
\end{align*}
Finally, we have the two inequalities for all $x,y \in {\mathcal X}$,
\begin{align*}
  \bar P^{n} \dif (x,y) &\leq (1+M) \frac{V(x) + V (y)-1}{r^{0}_{\bar \phi} (n-1) }, \\
  \bar P^{n} \dif (x,y) & \leq 1,
\end{align*}
where the last inequality follows from $\dif \leq 1$. Hence, writing $$A_{n}=\left\{(x,y)\in\mathbb{X}\times\mathbb{X}:\; V (x)+ V (y)> H^{-1}_{\bar\phi} (n)\right\},$$ we have 
\begin{align*}
  \bar P^{n} \dif (x,y) &\leq \mathds{1}_{A_{n}} (x,y) + \frac {1+M}{r^{0}_{\bar \phi} (n-1)} \lr{V (x)+ V (y)-1}\mathds{1}_{A^\complement_{n}}  (x,y)\\
  & \leq \frac{\psi \circ \lr{V (x) + V (y)}} {\psi \circ H^{-1}_{\bar \phi} (n)} + \frac{1+M}{r^{0}_{\bar \phi} (n-1)} \frac{\psi\lr{V (x)+V (y)} - \psi(1)}{\psi \circ H^{-1}_{\bar \phi} (n)- \psi (1)} \lrb{H^{-1}_{\bar \phi} (n)-1},
\end{align*}
where the last inequality holds true since $1\leq V(x)+V (y) \leq H^{-1}_{\bar \phi} (n)$ implies, noting that $\psi$ is concave, 
$$
\frac{\psi \lr{V (x)+ V (y)} - \psi (1)}{V (x)+V (y) -1} \geq \frac{\psi \circ H^{-1}_{\bar \phi} (n)-\psi (1)}{H^{-1}_{\bar \phi} (n)-1}. 
$$
Finally, there exists a constant $\beta$ such that 
\begin{align*}
  \bar P^{n} \dif (x,y) &\leq \frac{\psi \lr{V (x)+V (y)}} {\psi \circ H^{-1}_{\bar \phi} (n)} \lr{1+ \frac{1+M} {1-\psi (1)/\psi \circ H^{-1}_{\bar \phi} (2)} \frac{H^{-1}_{\bar \phi} (n)}{r^{0}_{\bar \phi} (n-1)}}
  \\
  &\leq \beta \frac{\psi\lr{V (x)+V (y)}}{\psi \circ H^{-1}_{\bar \phi} (n)},
\end{align*}
where the last inequality follows from the fact that $\bar \phi \circ H^{-1}_{\bar \phi}=\lrb{H^{-1}_{\bar \phi}}'$ which implies that $H^{-1}_{\bar \phi} (n)/r^{0}_{\bar \phi} (n-1)$ is bounded from above by comparison between series and integrals. In addition, noting that $\psi$ is concave, for any $u,v\geq 0$, $\psi(u+v)-\psi(u)\leq \psi(v) - \psi (0)\leq \psi (v)$. Hence,
$$
{\mathbb P}_{x,y} (\tau>n)=  \bar P^{n} \dif (x,y) \leq  \beta \frac{\psi\circ V (x)+\psi\circ V (y)} {\psi \circ H^{-1}_{\bar \phi} (n)}.
$$
We conclude by observing that $\bar{\phi} = \sigma \phi$ implies $H^{-1}_{\bar{\phi}}(n) = H^{-1}_{\phi}(\sigma n)$.
\end{proof}

The following result provides conditions for the finiteness
of moments of $\fishytest_{\star}$.

\begin{prop}[{\citealt[Proposition 21.2.3]{douc2018MarkovChains}}]
\label{prop:normal-conv-g-star}Let $P$ be a Markov kernel with
unique invariant distribution $\pi$. Let $\test_{0}\in L_{0}^{p}(\pi)$
for some $p\geq1$. If
\[
\sum_{t=0}^{\infty}\left\Vert P^{t}\test_{0}\right\Vert _{L^{p}(\pi)}<\infty,
\]
then $\fishytest_{\star}=\sum_{t=0}^{\infty}P^{t}h_{0}$ is fishy
for $\test_{0}$ and $\fishytest_{\star}\in L^{p}(\pi)$.
\end{prop}

We consider here what Assumption~\ref{assu:tau-moment-kappa} implies
about the corresponding $P$ and its fishy functions. First, we observe
that this assumption implies that $P$ is aperiodic, and also ergodic
of degree $2$ \citep[as defined in][Section 6.4]{nummelinbook},
which implies e.g. that a CLT holds for ergodic averages of bounded
functions.
\begin{prop}
\label{prop:int-tv}If Assumption~\ref{assu:tau-moment-kappa} holds
then
\begin{equation}
\int\pi({\rm d}x)\left\Vert P^{t}(x,\cdot)-\pi\right\Vert _{{\rm TV}}\leq\mathbb{E}_{\pi\otimes\pi}[\tau^{\kappa}](t+1)^{-\kappa},\label{eq:int-tv}
\end{equation}
and $P$ is aperiodic, and ergodic of degree 2.
\end{prop}

\begin{proof}
Using Proposition~\ref{prop:tau-moment-survival},
\begin{align}
\left\Vert P^{t}(x,\cdot)-\pi\right\Vert _{{\rm TV}} & =\left\Vert P^{t}(x,\cdot)-\pi P^{t}\right\Vert _{{\rm TV}}\nonumber \\
 & \leq\int\pi({\rm d}y)\mathbb{P}_{x,y}(\tau>t)\nonumber \\
 & \leq(t+1)^{-\kappa}\mathbb{E}_{x,\pi}[\tau^{\kappa}],\label{eq:Pt-pi-tv}
\end{align}
and since Assumption~\ref{assu:tau-moment-kappa} provides $\kappa>1$
and $\mathbb{E}_{x,\pi}[\tau^{\kappa}]<\infty$ for $\pi$-almost
all $x$, $\left\Vert P^{t}(x,\cdot)-\pi\right\Vert _{{\rm TV}}$
goes to zero as $t\to\infty$
for $\pi$-almost all $x$ and so $P$ is aperiodic by \citet[Lemma~9.3.9]{douc2018MarkovChains}.
We also observe that (\ref{eq:Pt-pi-tv}) implies (\ref{eq:int-tv})
and since $\kappa>1$ we have
\[
\sum_{t=0}^{\infty}\int\pi({\rm d}x)\left\Vert P^{t}(x,\cdot)-\pi\right\Vert _{{\rm TV}}<\infty,
\]
and $P$ is ergodic of degree 2 by \citet[Theorem II.4.1]{chen1999limit}.
\end{proof}
We now turn to the implication of Assumption~\ref{assu:tau-moment-kappa}
on properties of $g_{\star}$. In particular, we see that $\kappa>p$
with sufficiently many moments of $h$ implies that $g_{\star}\in L^{p}(\pi)$.
\begin{thm}
\label{thm:poisson-Lp}Under Assumption~\ref{assu:tau-moment-kappa},
let $h\in L_{0}^{m}(\pi)$ for some $m>\kappa/(\kappa-1)$. For $p\geq1$
such that $\frac{1}{p}>\frac{1}{m}+\frac{1}{\kappa}$, $g_{\star}\in L_{0}^{p}(\pi)$.
\end{thm}

\begin{proof}
For arbitrary $t\in\mathbb{N}$, since $\pi(h)=0$ and $\pi P^{t}=\pi$, for $\pi$-almost all $x$,
\begin{align*}
\left|P^{t}h(x)\right| & =\left|\mathbb{E}_{x,\pi}\left[h(X_{t})-h(Y_{t})\right]\right|\\
 & \leq\mathbb{E}_{x,\pi}\left[\mathds{1}(\tau>t)\left\{ \left|h(X_{t})\right|+\left|h(Y_{t})\right|\right\} \right],
\end{align*}
and hence by Jensen's inequality and $(a+b)^{p}\leq2^{p-1}(a^{p}+b^{p})$
for $a,b\geq0$,
\begin{align*}
\left|P^{t}h(x)\right|^{p} & \leq\mathbb{E}_{x,\pi}\left[\mathds{1}(\tau>t)\left\{ \left|h(X_{t})\right|+\left|h(Y_{t})\right|\right\} \right]^{p}\\
 & \leq\mathbb{E}_{x,\pi}\left[\mathds{1}(\tau>t)\left\{ \left|h(X_{t})\right|+\left|h(Y_{t})\right|\right\} ^{p}\right]\\
 & \leq2^{p-1}\mathbb{E}_{x,\pi}\left[\mathds{1}(\tau>t)\left\{ \left|h(X_{t})\right|^{p}+\left|h(Y_{t})\right|^{p}\right\} \right].
\end{align*}
Using H\"older's inequality with $\delta=p/m$, which is in $(0,1)$
by assumption,
\begin{align*}
\left\Vert P^{t}h\right\Vert _{L^{p}(\pi)}^{p} & =\int\pi({\rm d}x)\left|P^{t}h(x)\right|^{p}\\
 & \leq2^{p-1}\mathbb{E}_{\pi\otimes\pi}\left[\mathds{1}(\tau>t)\left\{ \left|h(X_{t})\right|^{p}+\left|h(Y_{t})\right|^{p}\right\} \right]\\
 & \leq2^{p}\mathbb{P}_{\pi\otimes\pi}(\tau>t)^{1-\delta}\mathbb{E}_{\pi}\left[\left|h(X_{t})\right|^{\frac{p}{\delta}}\right]^{\delta}\\
 & =2^{p}\mathbb{P}_{\pi\otimes\pi}(\tau>t)^{1-\frac{p}{m}}\pi(\left|h\right|^{m})^{\frac{p}{m}}.
\end{align*}
Since $\kappa(m-p)/(mp)>1$ by assumption, using Proposition~\ref{prop:tau-moment-survival},
\[
\sum_{t\geq0}\left\Vert P^{t}h\right\Vert _{L^{p}(\pi)}\leq2\left\Vert h\right\Vert _{L^{m}(\pi)}\sum_{t\geq0}\mathbb{P}_{\pi\otimes\pi}(\tau>t)^{\frac{m-p}{mp}}<\infty,
\]
and we conclude by appealing to Proposition~\ref{prop:normal-conv-g-star}. 
\end{proof}
\begin{rem}
The results above and Theorem~\ref{thm:clt-kappa} rely only on the
existence of meeting times with polynomial survival functions, so
one can deduce that they hold for any Markov kernel $P$ such that
$\left\Vert P^{t}(x,\cdot)-\pi\right\Vert _{{\rm TV}}\leq M(x)(t+1)^{-s}$
with $s>\kappa$ and $\pi(M)<\infty$, since then
\[
\left\Vert P^{t}(x,\cdot)-P^{t}(y,\cdot)\right\Vert _{{\rm TV}}\leq\left\{ M(x)+M(y)\right\} (t+1)^{-s}.
\]
Conversely, for such a Markov kernel there exists a possibly non-Markovian
coupling that would satisfy Assumption~\ref{assu:tau-moment-kappa}
\citep[see, e.g.,][]{griffeath1975maximal}, but we do not pursue
this here.
\end{rem}

\subsection{\label{subsec:clt-poisson}Proof of Theorem~\ref{thm:clt-kappa}}

\begin{lem}[{\citealt[Theorem~1.1]{rio1993covariance}}]
\label{lem:rio-cov}Let $X$ and $Y$ be integrable random variables
such that
\[
\alpha=\sup_{A,B}\mathbb{P}(X\in A,Y\in B)-\mathbb{P}(X\in A)\mathbb{P}(Y\in B)=\sup_{A,B}{\rm cov}({\mathds{1}}_{A}(X),{\mathds{1}}_{B}(Y)),
\]
where the supremum is over all measurable sets. Then
\[
\left|{\rm cov}(X,Y)\right|\leq2\int_{0}^{2\alpha}Q_{X}(u)Q_{Y}(u){\rm d}u\leq4\int_{0}^{\alpha}Q_{X}(u)Q_{Y}(u){\rm d}u,
\]
where for a random variable $Z$, $Q_{Z}$ is the tail quantile function
of $|Z|$, i.e. $Q_{Z}(u)=\inf\{t:\mathbb{P}(\left|Z\right|>t)\leq u\}$.
\end{lem}

\begin{lem}
\label{lem:int-tv-rhok-cov}Assume that, for all $k\in\mathbb{N}$,
\[
\int\pi({\rm d}x)\left\Vert P^{k}(x,\cdot)-\pi\right\Vert _{{\rm TV}}\leq\rho_{k}.
\]
Then for $v$, $h$ measurable functions, and $A$, $B$ measurable
sets,
\[
\left|{\rm cov}\left({\mathds{1}}_{A}(v(X_{0})),{\mathds{1}}_{B}(h(X_{k}))\right)\right|\leq\rho_{k}.
\]
\end{lem}

\begin{proof}
Denote the preimage of $B$ under $h$ as $h^{-1}(B)=\{x:h(x)\in B\}$.
We have

\begin{eqnarray*}
 &  & \left|\mathbb{P}_{\pi}(v(X_{0})\in A,h(X_{k})\in B)-\mathbb{P}_{\pi}(v(X_{0})\in A)\mathbb{P}_{\pi}(h(X_{k})\in B)\right|\\
 & = & \left|\mathbb{E}_{\pi}\left[{\mathds{1}}_{A}(v(X_{0}))\mathbb{E}\left[\left\{ {\mathds{1}}_{B}(h(X_{k}))-\mathbb{P}_{\pi}(h(X_{k})\in B)\right\} \mid\sigma(X_{0})\right]\right]\right|\\
 & = & \left|\mathbb{E}_{\pi}\left[{\mathds{1}}_{A}(v(X_{0}))\left\{ P^{k}(X_{0},h^{-1}(B))-\pi(h^{-1}(B))\right\} \right]\right|\\
 & \leq & \mathbb{E}_{\pi}\left[\left|P^{k}(X_{0},h^{-1}(B))-\pi(h^{-1}(B))\right|\right]\\
 & \leq & \mathbb{E}_{\pi}\left[\left\Vert P^{k}(X_{0},\cdot)-\pi\right\Vert _{{\rm TV}}\right]\\
 & \leq & \rho_{k}.
\end{eqnarray*}
\end{proof}
\begin{lem}
\label{lem:Plkh-tail-quantile}Assume that
\[
\int\pi({\rm d}x)\left\Vert P^{k}(x,\cdot)-\pi\right\Vert _{{\rm TV}}\leq\rho_{k}\leq1,\qquad k\in\mathbb{N}.
\]
Then with $h$ such that $\pi(h)=0$, and $k,\ell\in\mathbb{N}$,
\[
\left|\pi(P^{\ell}h\cdot P^{k}h)\right|\leq4\int_{0}^{\rho_{\ell}\wedge\rho_{k}}Q_{0}(u)^{2}{\rm d}u,
\]
where $Q_{0}$ is the tail quantile function of $|h(X_{0})|$.
\end{lem}

\begin{proof}
We follow the same strategy as in \citet[Lemma~21.4.3]{douc2018MarkovChains}.
By Lemma~\ref{lem:int-tv-rhok-cov} with $v=P^{\ell}h$, we obtain
\[
\alpha(v(X_{0}),h(X_{k}))=\sup_{A,B} {\rm cov}\left({\mathds{1}}_{A}(v(X_{0})),{\mathds{1}}_{B}(h(X_{k}))\right)\leq\rho_{k},
\]
for use in Lemma~\ref{lem:rio-cov}. It follows that
\begin{align*}
\left|\pi(P^{\ell}h\cdot P^{k}h)\right| & =\left|{\rm cov}(v(X_{0}),h(X_{k}))\right|\\
 & \leq4\int_{0}^{\rho_{k}}Q_{v(X_{0})}(u)Q_{h(X_{k})}(u){\rm d}u,
\end{align*}
where $Q_{Z}$ is the tail quantile function of $|Z|$. Since $h(X_{k})$
has the same distribution as $h(X_{0})$, $Q_{h(X_{k})}=Q_{0}$. On
the other hand by \citet[Lemma~21.A.3]{douc2018MarkovChains} we have for all $a\in [0,1]$,
$\int_{0}^{a}Q_{v(X_{0})}(u)^{2}{\rm d}u\leq\int_{0}^{a}Q_{0}(u)^{2}{\rm d}u$.
Hence, by Cauchy--Schwarz, we have
\begin{align*}
\left|\pi(P^{\ell}h\cdot P^{k}h)\right| & \leq4\int_{0}^{\rho_{k}}Q_{v(X_{0})}(u)Q_{h(X_{k})}(u){\rm d}u\\
 & \leq4\left\{ \int_{0}^{\rho_{k}}Q_{v(X_{0})}(u)^{2}{\rm d}u\right\} ^{1/2}\left\{ \int_{0}^{\rho_{k}}Q_{0}(u)^{2}{\rm d}u\right\} ^{1/2}\\
 & \leq4\int_{0}^{\rho_{k}}Q_{0}(u)^{2}{\rm d}u.
\end{align*}
By interchanging the use of $k$ and $\ell$, we obtain the final
bound.
\end{proof}
The following lemma is similar to \citet[Theorem~21.4.4]{douc2018MarkovChains}.
\begin{lem}
\label{lem:rho-sum-g-L2-avar}Assume that
\[
\int\pi({\rm d}x)\left\Vert P^{n}(x,\cdot)-\pi\right\Vert _{{\rm TV}}\leq\rho_{n}\leq1,\qquad n\in\mathbb{N},
\]
and let $h\in L_{0}^{m}(\pi)$. Then
\begin{enumerate}
\item $\left|\pi(P^{k}h\cdot P^{\ell}h)\right|\leq4\left\Vert h\right\Vert _{L^{m}(\pi)}^{2}\frac{m}{m-2}\left(\rho_{k}\wedge\rho_{\ell}\right)^{\frac{m-2}{m}}$.
\item If $g_{\star}\in L_{0}^{1}(\pi)$ and $\sum_{k=0}^{\infty}\rho_{k}^{\frac{m-2}{m}}<\infty$,
then $\pi(h\cdot g_{\star})<\infty$.
\end{enumerate}
\end{lem}

\begin{proof}
We have by Markov's inequality
\[
\mathbb{P}_{\pi}(\left|h(X_{0})\right|>t)\leq\pi(|h|^{m})/t^{m},
\]
from which we may deduce that $Q_{0}(u)\leq\left\Vert h\right\Vert _{L^{m}(\pi)}u^{-1/m}$
and so
\[
\int_{0}^{a}Q_{0}(u)^{2}{\rm d}u\leq\left\Vert h\right\Vert _{L^{m}(\pi)}^{2}\frac{m}{m-2}a^{\frac{m-2}{m}}.
\]
By Lemma~\ref{lem:Plkh-tail-quantile}, it follows that
\[
\left|\pi(P^{k}h\cdot P^{\ell}h)\right|\leq4\int_{0}^{\rho_{k}\wedge\rho_{\ell}}Q_{0}(u)^{2}{\rm d}u\leq4\left\Vert h\right\Vert _{L^{m}(\pi)}^{2}\frac{m}{m-2}\left(\rho_{k}\wedge\rho_{\ell}\right)^{\frac{m-2}{m}}.
\]
Moreover, if $g_{\star}\in L_{0}^{1}(\pi)$,
\begin{align*}
\pi(h\cdot g_{\star}) & =\sum_{k=0}^{\infty}\pi(h\cdot P^{k}h)\\
 & \leq\sum_{k=0}^{\infty}\left|\pi(h\cdot P^{k}h)\right|\\
 & \leq4\left\Vert h\right\Vert _{L^{p}(\pi)}^{2}\frac{m}{m-2}\sum_{k=0}^{\infty}\rho_{k}^{\frac{m-2}{m}},
\end{align*}
from which we may conclude.
\end{proof}
\begin{proof}[Proof of Theorem~\ref{thm:clt-kappa}]
Without loss of generality, assume $\pi(h)=0$. By Theorem~\ref{thm:poisson-Lp},
we have $g_{\star}\in L_{0}^{1}(\pi)$ since $\kappa>1$ and $m>2\kappa/(\kappa-1)>\kappa/(\kappa-1)$.
Then, by Proposition~\ref{prop:int-tv} and Lemma~\ref{lem:rho-sum-g-L2-avar}
we obtain $\pi(h\cdot g_{\star})<\infty$. For the CLT, we appeal
to \citet{maxwell2000central}, for which it is sufficient to show
that
\begin{equation}
\sum_{n\geq1}n^{-3/2}\left\Vert V_{n}h\right\Vert _{L^{2}(\pi)}<\infty,\label{eq:mw-cond}
\end{equation}
where $V_{n}f=\sum_{k=0}^{n-1}P^{k}f$, and $v(P,h)$ is then equal
to $\lim_{n\to\infty}n^{-1}\mathbb{E}_{\pi}\left[\left\{ \sum_{k=1}^{n}h(X_{k})\right\} ^{2}\right]$.
We find using Lemma~\ref{lem:rho-sum-g-L2-avar},
\begin{align*}
\left\Vert V_{n} h\right\Vert _{L^{2}(\pi)}^{2} & =\int\pi({\rm d}x)\left\{ \sum_{k=0}^{n-1}P^{k}h(x)\right\} ^{2}\\
 & =\int\pi({\rm d}x)\sum_{k=0}^{n-1}\sum_{\ell=0}^{n-1}P^{k}h(x)P^{\ell}h(x)\\
 & \leq\sum_{k=0}^{n-1}\sum_{\ell=0}^{n-1}\left|\int\pi({\rm d}x)P^{k}h(x)P^{\ell}h(x)\right|\\
 & \leq4\left\Vert h\right\Vert _{L^{m}(\pi)}^{2}\frac{m}{m-2}\sum_{k=0}^{n-1}\sum_{\ell=0}^{n-1}\left(\rho_{k}\wedge\rho_{\ell}\right)^{\frac{m-2}{m}}\\
 & \leq4\left\Vert h\right\Vert _{L^{m}(\pi)}^{2}\frac{m}{m-2}\sum_{k=0}^{n-1}\sum_{\ell=0}^{n-1}a_{k}\wedge a_{\ell},
\end{align*}
where we define $a_{k}=\min\{1,\mathbb{E}_{\pi\otimes\pi}[\tau^{\kappa}]^{\frac{m-2}{m}}(k+1)^{-\kappa\frac{m-2}{m}}\}$
by Proposition~\ref{prop:int-tv}. Since $(a_{k})$ is non-increasing,
we may deduce that
\[
\sum_{k=0}^{n-1}\sum_{\ell=0}^{n-1}a_{k}\wedge a_{\ell}=\sum_{k=0}^{n-1}(2k+1)a_{k}.
\]
It follows that
\[
\sum_{k=0}^{n-1}(2k+1)a_{k}\leq2\sum_{k=0}^{n-1}(k+1)a_{k}\leq2\mathbb{E}_{\pi\otimes\pi}[\tau^{\kappa}]^{\frac{m-2}{m}}\sum_{k=0}^{n-1}(k+1)^{1-\kappa\frac{m-2}{m}},
\]
and this is $O(n^{1-2\epsilon})$ for some $\epsilon>0$ since $m>2\kappa/(\kappa-1)$.
Hence, $\left\Vert V_{n}h\right\Vert _{L^{2}(\pi)}=O(n^{1/2-\epsilon})$
and \eqref{eq:mw-cond} is satisfied. Since $\pi(h\cdot g_{\star})<\infty$,
this implies that $\lim_{n\to\infty}n^{-1}\mathbb{E}\left[\left\{ \sum_{i=1}^{n}h(X_{i})\right\} ^{2}\right]=\pi(h\cdot g_{\star})-\pi(h^{2})$
by \citet[Lemma~21.2.7]{douc2018MarkovChains}, and we conclude.
\end{proof}

\subsection{\label{subsec:approx-fishy-function}Unbiased approximation of fishy functions}

The following technical lemma will be useful to obtain bounds on the
moments of the fishy function estimator in Definition \ref{def:Gyx}.
For a random variable $X$, we denote $\left\Vert X\right\Vert _{L^{p}}:=\mathbb{E}\left[\left|X\right|^{p}\right]^{\frac{1}{p}}$.
\begin{lem}
\label{lem:tech-new} Let $(U_{i})$, $(V_{i})$ be sequences of random
variables, $N$ a non-negative, integer-valued and almost surely finite
random variable, and $p\geq1$. Then for any $\delta_{0},\delta_{1}\in(0,1)$
with $\delta_{0}+\delta_{1}<1$,
\[
\mathbb{E}\left[\left|\sum_{i=0}^{N-1}U_{i}+V_{i}\right|^{p}\right]^{\frac{1}{p}}\leq\zeta\left(\frac{1-\delta_{1}}{\delta_{0}}\right)\mathbb{E}\left(N^{\frac{p}{\delta_{0}}}\right)^{\frac{1-\delta_{1}}{p}}\left\{ \sup_{i}\left\Vert U_{i}\right\Vert _{L^{\frac{p}{\delta_{1}}}}+\sup_{i}\left\Vert V_{i}\right\Vert _{L^{\frac{p}{\delta_{1}}}}\right\} ,
\]
where $\zeta(s):=\sum_{n=1}^{\infty}n^{-s}$ is finite for $s>1$.
\end{lem}

\begin{proof}
By Minkowski's inequality, H\"older's inequality and Markov's inequality,
\begin{align*}
\mathbb{E}\left[\left|\sum_{i=0}^{N-1}U_{i}+V_{i}\right|^{p}\right]^{\frac{1}{p}} & =\left\Vert \sum_{i=0}^{N-1}U_{i}+V_{i}\right\Vert _{L^{p}}\\
 & \leq\left\Vert \sum_{i=0}^{\infty}\left|U_{i}\right|\mathds{1}(N>i)+\left|V_{i}\right|\mathds{1}(N>i)\right\Vert _{L^{p}}\\
 & \leq\sum_{i=0}^{\infty}\left\Vert \left|U_{i}\right|\mathds{1}(N>i)\right\Vert _{L^{p}}+\left\Vert \left|V_{i}\right|\mathds{1}(N>i)\right\Vert _{L^{p}}\\
 & \leq\sum_{i=0}^{\infty}\mathbb{P}(N>i)^{\frac{1-\delta_{1}}{p}}\left\{ \left\Vert U_{i}\right\Vert _{L^{\frac{p}{\delta_{1}}}}+\left\Vert V_{i}\right\Vert _{L^{\frac{p}{\delta_{1}}}}\right\} \\
 & \leq\sum_{i=0}^{\infty}\left\{ \frac{\mathbb{E}\left(N^{\frac{p}{\delta_{0}}}\right)}{(i+1)^{\frac{p}{\delta_{0}}}}\right\} ^{\frac{1-\delta_{1}}{p}}\left\{ \left\Vert U_{i}\right\Vert _{L^{\frac{p}{\delta_{1}}}}+\left\Vert V_{i}\right\Vert _{L^{\frac{p}{\delta_{1}}}}\right\} \\
 & \leq\zeta\left(\frac{1-\delta_{1}}{\delta_{0}}\right)\mathbb{E}\left[N^{\frac{p}{\delta_{0}}}\right]^{\frac{1-\delta_{1}}{p}}\left\{ \sup_{i}\left\Vert U_{i}\right\Vert _{L^{\frac{p}{\delta_{1}}}}+\sup_{i}\left\Vert V_{i}\right\Vert _{L^{\frac{p}{\delta_{1}}}}\right\} .
\end{align*}
\end{proof}

The following lemma employs dominated convergence to justify the interchange
of expectation and infinite sum, and thereby ensure that $G_{y}(x)$
is an unbiased estimator of $g_{y}(x)$.

\begin{lem}
\label{lem:unbiasedness}Under Assumption~\ref{assu:tau-moment-kappa},
let $h\in L^{m}(\pi)$ for some $m>\kappa/(\kappa-1)$. For $\pi$-almost
all $x$ and $y$, if $\mathbb{P}_{x,y}(\tau<\infty)=1$ and
\begin{equation}
\mathbb{E}_{x,y}\left[\sum_{t=0}^{\tau-1}\left|\test(X_{t})\right|+\left|\test(Y_{t})\right|\right]<\infty,\label{eq:unbiased-condition}
\end{equation}
then $\mathbb{E}\left[G_{y}(x)\right]=\fishytest_{y}(x)$.
\end{lem}

\begin{proof}
Fix $x$ and $y$, let $G_{n}=\sum_{t=0}^{n}h(X_{t})-h(Y_{t})$ with
$(X_{0},Y_{0})=(x,y)$. Then
\[
\mathbb{E}_{x,y}\left[G_{n}\right]=\sum_{t=0}^{n}P^{t}h(x)-P^{t}h(y)=\sum_{t=0}^{n}P^{t}h_{0}(x)-P^{t}h_{0}(y).
\]
Since $m>\kappa/(\kappa-1)$, Theorem~\ref{thm:poisson-Lp} provides
that $g_{\star}\in L_{0}^{1}(\pi)$, and hence for $\pi$-almost all
$x$ and $y$ we have
\[
\lim_{n\to\infty}\mathbb{E}_{x,y}\left[G_{n}\right]=g_{\star}(x)-g_{\star}(y)=g_{y}(x).
\]
Since $\mathbb{P}_{x,y}(\tau<\infty)=1$ we have $G_{n}\to G_{y}(x)$
$\mathbb{P}_{x,y}$-almost surely as $n\to\infty$, so $\mathbb{E}[G_{y}(x)]=\mathbb{E}_{x,y}\left[\lim_{n\to\infty}G_{n}\right]$.
For any $n\in\mathbb{N}$,
\[
\left|G_{n}\right|\leq\sum_{t=0}^{\tau-1}\left|h(X_{t})\right|+\left|h(Y_{t})\right|,
\]
and so the assumed integrability of the right-hand side implies, by
dominated convergence, that

\[
\mathbb{E}\left[G_{y}(x)\right]=\mathbb{E}_{x,y}\left[\lim_{n\to\infty}G_{n}\right]=\lim_{n\to\infty}\mathbb{E}_{x,y}\left[G_{n}\right]=g_{y}(x).
\]
\end{proof}
The following lemma is used several times to ensure that expectations
of functions of $X_{t}$ are uniformly bounded in $t$ under reasonable
conditions. In the first case, the conclusion is a result of stability
properties of $\pi$-invariant Markov chains started at almost all points,
while in the second case regularity is imposed by ensuring that the initial distribution $\mu$
cannot place too much probability in possibly problematic regions;
for example it guarantees that $\phi\in L^{1}(\mu)$.
It is possible that one can weaken the condition in the second statement to
${\rm d}\mu/{\rm d}\pi\in L^{p}(\pi)$ for some $p\geq1$, but this would require
stronger conditions on $\phi$ and hence complicate subsequent results.
\begin{lem}
\label{lem:sup-finite}Let $0\leq\phi\in L^{1}(\pi)$. Under Assumption~\ref{assu:tau-moment-kappa},
\begin{enumerate}
\item for $\pi$-almost all $x$
\[
\sup_{t\geq0}\mathbb{E}_{x}\left[\phi(X_{t})\right]<\infty,
\]
\item if $\mu\ll\pi$ is a probability measure such that ${\rm d}\mu/{\rm d}\pi\leq M<\infty$,
then
\[
\sup_{t\geq0}\mathbb{E}_{\mu}\left[\phi(X_{t})\right]\leq M\pi(\phi).
\]
\end{enumerate}
\end{lem}

\begin{proof}
By assumption, $P$ is $\pi$-irreducible and from Proposition~\ref{prop:int-tv}
it is aperiodic. Since $\pi(\phi)<\infty$, we may define $f=1+\phi\geq1$
and $\pi(f)<\infty$. By the $f$-norm ergodic theorem \citep[Theorem 14.0.1]{meyn:tweedie:1993},
for $\pi$-almost all $x$, $P^{t}f(x)$ is finite for all $t\geq0$
and
\[
\lim_{t\to\infty}\left\Vert P^{t}(x,\cdot)-\pi\right\Vert _{f}=0,
\]
where
\[
\left\Vert P^{t}(x,\cdot)-\pi\right\Vert _{f}=\sup_{g:|g|\leq f}\left|P^{t}g(x)-\pi(g)\right|.
\]
Since $\phi=|\phi|\leq f$,
\[
P^{t}\phi(x)\leq\pi(\phi)+\left\Vert P^{t}(x,\cdot)-\pi\right\Vert _{f},
\]
it follows that for $\pi$-almost all $x$,
\[
\lim_{t\to\infty}\mathbb{E}_{x}\left[\phi(X_{t})\right]=\lim_{t\to\infty}P^{t}\phi(x)=\pi(\phi)<\infty,
\]
and hence $\sup_{t\geq0}\mathbb{E}_{x}\left[\phi(X_{t})\right]<\infty$.
For the second part, we have
\begin{align*}
\mathbb{E}_{\mu}\left[\phi(X_{t})\right] & =\mu P^{t}(\phi)\\
 & =\int\mu({\rm d}x)P^{t}(x,{\rm d}y)\phi(y)\\
 & =\int\pi({\rm d}x)\frac{{\rm d}\mu}{{\rm d}\pi}(x)P^{t}(x,{\rm d}y)\phi(y)\\
 & \leq M\pi P^{t}(\phi)\\
 & =M\pi(\phi),
\end{align*}
which concludes the proof.
\end{proof}
The following theorem provides sufficient conditions for the estimator
$G_{y}(x)$ to be unbiased and have finite $p$th moments.
\begin{lem}
\label{lem:fishy-lemma-new}Under Assumption~\ref{assu:tau-moment-kappa},
let $h\in L^{m}(\pi)$ for some $m>\kappa/(\kappa-1)$. Let $\psi(h,m,\nu)=\sup_{t\geq0}\mathbb{E}_{\nu}\left[\left|h(X_{t})\right|^{m}\right]^{\frac{1}{m}}$,
$\gamma$ be a probability measure on $\mathbb{X}\times \mathbb{X}$,
and let $\gamma_{1}=\gamma(\cdot \times \mathbb{X})$
and $\gamma_{2}=\gamma(\mathbb{X} \times \cdot)$ satisfy $\gamma_{1},\gamma_{2}\ll\pi$.
If $\mathbb{E}_{\gamma}[\tau^{\kappa}]$, $\psi(h,m,\gamma_{1})$
and $\psi(h,m,\gamma_{2})$ are all finite then $\mathbb{E}_{\gamma}\left[G_{Y_{0}}(X_{0})\right]=\gamma_{1}(g_{\star})-\gamma_{2}(g_{\star})$,
and for $p\geq1$ such that $\frac{1}{p}>\frac{1}{m}+\frac{1}{\kappa}$,
\begin{equation}
\mathbb{E}_{\gamma}\left[\left|G_{Y_{0}}(X_{0})\right|^{p}\right]^{\frac{1}{p}}\leq\zeta\left(\frac{(m-p)\kappa}{mp}\right)\mathbb{E}_{\gamma}[\tau^{\kappa}]^{\frac{m-p}{mp}}\left\{ \psi(h,m,\gamma_{1})+\psi(h,m,\gamma_{2})\right\} <\infty,\label{eq:fishy-moment-bound-1}
\end{equation}
where $\zeta(s)=\sum_{n=1}^{\infty}n^{-s}$ for $s>1$.
\end{lem}

\begin{proof}
With $(X_{0},Y_{0})\sim\gamma$, let
\[
\bar{G}=\sum_{t=0}^{\tau-1}\left|h(X_{t})\right|+\left|h(Y_{t})\right|\geq\left|G_{Y_{0}}(X_{0})\right|.
\]
If $p\geq1$ and $\frac{1}{p}>\frac{1}{m}+\frac{1}{\kappa}$, $\delta_{0}=p/\kappa$
and $\delta_{1}=p/m$ are in $(0,1)$ with $\delta_{0}+\delta_{1}<1$. We
may therefore use Lemma~\ref{lem:tech-new} with $U_{t}=\left|h(X_{t})\right|$
and $V_{t}=\left|h(Y_{t})\right|$ to deduce that
\[
\mathbb{E}_{\gamma}\left[\left|\bar{G}\right|^{p}\right]^{\frac{1}{p}}\leq\zeta\left(\frac{(m-p)\kappa}{mp}\right)\mathbb{E}{}_{\gamma}\left[\tau^{\kappa}\right]^{\frac{m-p}{mp}}\left\{ \sup_{t\geq0}\mathbb{E}_{\gamma_{1}}\left[\left|h(X_{t})\right|^{m}\right]^{\frac{1}{m}}+\sup_{t\geq0}\mathbb{E}_{\gamma_{2}}\left[\left|h(Y_{t})\right|^{m}\right]^{\frac{1}{m}}\right\} ,
\]
from which (\ref{eq:fishy-moment-bound-1}) follows. For the lack-of-bias
property, $m>\kappa/(\kappa-1)$ implies $1>\frac{1}{m}+\frac{1}{\kappa}$,
so the RHS of (\ref{eq:fishy-moment-bound-1}) is finite for $p=1$,
and (\ref{eq:unbiased-condition}) holds for $\gamma$-almost all
$(x,y)$. Since $\mathbb{E}_{\gamma}[\tau^{\kappa}]<\infty$ implies
$\mathbb{P}_{x,y}(\tau<\infty)$ for $\gamma$-almost all $(x,y)$,
we deduce by Lemma~\ref{lem:unbiasedness} that $\mathbb{E}\left[G_{y}(x)\right]=g_{y}(x)$
for $\gamma$-almost all $(x,y)$. It follows that
\[
\mathbb{E}_{\gamma}\left[\mathbb{E}\left[G_{Y_{0}}(X_{0})\mid\sigma(X_{0},Y_{0})\right]\right]=\mathbb{E}_{\gamma}\left[g_{\star}(X_{0})-g_{\star}(Y_{0})\right]=\gamma_{1}(g_{\star})-\gamma_{2}(g_{\star}).
\]
\end{proof}

\begin{thm}
\label{thm:fishy-estimator}Under Assumption~\ref{assu:tau-moment-kappa},
let $h\in L^{m}(\pi)$ for some $m>\kappa/(\kappa-1)$. Let $p\geq1$
satisfy $\frac{1}{p}>\frac{1}{m}+\frac{1}{\kappa}$, and $\zeta$
and $\psi$ be as defined in Lemma~\ref{lem:fishy-lemma-new}. Let
$\gamma$ be a probability measure with $\gamma_{1}=\gamma(\cdot \times \mathbb{X})$
and $\gamma_{2}=\gamma(\mathbb{X} \times \cdot)$.
\begin{enumerate}
\item For $\pi$-almost all $x$ and $\pi$-almost all $y$, if $\gamma_{1}=\delta_{x}$
and $\gamma_{2}=\delta_{y}$ and $\mathbb{E}_{x,y}[\tau^{\kappa}]<\infty$
then (\ref{eq:fishy-moment-bound-1}) is finite and $\mathbb{E}\left[G_{y}(x)\right]=g_{\star}(x)-g_{\star}(y)$.
\item For $\pi\otimes\pi$-almost all $(x,y)$, if $\gamma_{1}=\delta_{x}$
and $\gamma_{2}=\delta_{y}$ then (\ref{eq:fishy-moment-bound-1})
is finite and $\mathbb{E}\left[G_{y}(x)\right]=g_{\star}(x)-g_{\star}(y)$.
\item For $\pi$-almost all $y$, if $\gamma_{1}\ll\pi$, ${\rm d}\gamma_{1}/{\rm d}\pi\leq M$
and $\gamma_{2}=\delta_{y}$,
\[
\mathbb{E}_{\gamma}\left[\left|G_{Y_{0}}(X_{0})\right|^{p}\right]^{\frac{1}{p}}\leq\zeta\left(\frac{(m-p)\kappa}{mp}\right)\mathbb{E}_{\gamma}[\tau^{\kappa}]^{\frac{m-p}{mp}}\left\{ M^{\frac{1}{m}}\left\Vert h\right\Vert _{L^{m}(\pi)}+\psi(h,m,\delta_{y})\right\} <\infty,
\]
and $\mathbb{E}_{\gamma}\left[G_{Y_{0}}(X_{0})\right]=\gamma_{1}(g_{\star})-g_{\star}(y)$.
\item If $\gamma_{i}\ll\pi$, ${\rm d}\gamma_{i}/{\rm d}\pi\leq M$ for
$i\in\{1,2\}$ and $\mathbb{E}_{\gamma}\left[\tau^{\kappa}\right]<\infty$
then
\[
\mathbb{E}_{\gamma}\left[\left|G_{Y_{0}}(X_{0})\right|^{p}\right]^{\frac{1}{p}}\leq2\zeta\left(\frac{(m-p)\kappa}{mp}\right)\mathbb{E}_{\gamma}[\tau^{\kappa}]^{\frac{m-p}{mp}}M^{\frac{1}{m}}\left\Vert h\right\Vert _{L^{m}(\pi)}<\infty,
\]
and $\mathbb{E}_{\gamma}\left[G_{Y_{0}}(X_{0})\right]=\gamma_{1}(g_{\star})-\gamma_{2}(g_{\star})$.
\end{enumerate}
\end{thm}

\begin{proof}
All the parts are deduced from Lemma~\ref{lem:fishy-lemma-new}.
Recall that $\zeta$ is defined as in Lemma \ref{lem:tech-new}.
For the first part, $\psi(h,m,\delta_{z})$ is finite for $\pi$-almost
all $z$ by Lemma~\ref{lem:sup-finite}. For the second part, we
add to this that $\mathbb{E}_{x,y}[\tau^{\kappa}]<\infty$ for $\pi\otimes\pi$-almost
all $(x,y)$ by Assumption~\ref{assu:tau-moment-kappa}. For the
third part, $\psi(h,m,\delta_{y})$ is finite for $\pi$-almost all
$y$ and $\psi(h,m,\gamma_{1})^{m}\leq M\pi(\left|h\right|^{m})$
by Lemma~\ref{lem:sup-finite}, while
\[
\mathbb{E}_{\gamma}[\tau^{\kappa}]=\mathbb{E}_{\gamma_{1},y}[\tau^{\kappa}]\leq M\int\pi({\rm d}x)\mathbb{E}_{x,y}[\tau^{\kappa}]<\infty,
\]
for $\pi$-almost all $y$ by Assumption~\ref{assu:tau-moment-kappa}.
For the fourth part, we add to this that $\psi(h,m,\gamma_{2})^{m}\leq M\pi(\left|h\right|^{m})$
by Lemma~\ref{lem:sup-finite}.
\end{proof}

\subsection{\label{subsec:Unbiased-approx-pi-test}Unbiased approximation of \texorpdfstring{$\pi(h)$}{pi(h)}}

We next demonstrate that the approximation in Definition~\ref{def:unbiased-estimator}
is indeed unbiased and has finite $p$th moments under suitable conditions.
The proof of Theorem~\ref{thm:ub-simple} is an application of Theorem~\ref{thm:fishy-estimator}
and its statement can be compared with \citet[Theorem~1]{middleton2020unbiased},
which treats the case $p=2$, for which we essentially arrive at the
same condition for $\kappa$ and $m$. The lack-of-bias condition
here is less demanding, and we deduce finiteness of higher moments
of $H$ for $m$ sufficiently large.

In Propositions~\ref{prop:bound-L-k}--\ref{prop:lag-average-transference}
we establish that the properties obtained for $H$ may be deduced
also for averages of lagged and offset estimators, in Definition~\ref{def:lagged-offset}, that are used in
practice. Finally, in Proposition~\ref{prop:subsample}
we show that subsampled estimators are also unbiased and have finite
$p$th moments under the same conditions. Subsampling is an
important aspect of the proposed asymptotic variance estimator (UPAVE, Section
\ref{subsec:unbiasedavar}).
In this section we use $\mu$ as notation for a possible initial distribution
of a chain evolving according to $P$. Later on this initial distribution
will be set to either the user-specified distribution $\pi_0$ or to the distribution
$\pi_k = \pi_0 P^k$ for some integer $k$.

The following two lemmas guarantee that if an independent initialization from $\mu$
is used for coupled lagged chains, i.e. $\gamma = \mu P \otimes \mu$, and ${\rm d}\mu/{\rm d}\pi\leq M$, then $\mathbb{E}_{\gamma}[\tau^{\kappa}]<\infty$
is guaranteed by Assumption~\ref{assu:tau-moment-kappa}, and similarly
lack-of-bias results for $\gamma$-almost all $(x,y)$ may be deduced
from lack-of-bias results for $\pi\otimes\pi$-almost all $(x,y)$. 
\begin{lem}
\label{lem:gamma-M-sq-independent-init}Let $\mu$ be a probability measure on $(\mathbb{X},\mathcal{X})$
and $P$ a $\pi$-invariant Markov kernel. Then if ${\rm d}\mu/{\rm d}\pi\leq M$
then ${\rm d}\mu P/{\rm d}\pi\leq M$ and ${\rm d} (\mu P\otimes \mu)/{\rm d}(\pi\otimes\pi)\leq M^{2}$.
\end{lem}

\begin{proof}
For $A\in\mathcal{X}$,
\[
\mu P(A)=\int\mu({\rm d}x)P(x,A)\leq M\int\pi({\rm d}x)P(x,A)=M\pi(A),
\]
from which we may deduce that ${\rm d}\mu P/{\rm d}\pi\leq M$. It
then follows that for $A,B\in\mathcal{X}$,
\[
(\mu P\otimes\mu)(A\times B)=\int\mu P(A)\mu(B)\leq M^{2}(\pi\otimes\pi)(A\times B),
\]
from which we may conclude.
\end{proof}
\begin{lem}
\label{lem:unbiasedness-exp}If $h,g_{\star}\in L^{1}(\pi)$ and
$\gamma$ is a coupling of $\mu P$ and $\mu$
for which $\mathbb{E}\left[G_{y}(x)\right]=g_{y}(x)$ for $\gamma$-almost all $(x,y)$,
then
\[
\mathbb{E}\left[H\right]=\mathbb{E}\left[h(X_{0}')+G_{Y_{0}'}(X_{1}')\right]=\pi(h),
\]
where $H$ is in Definition~\ref{def:unbiased-estimator} with $(X_1',Y_0') \sim \gamma$ and $X_0'\sim \mu$.
\end{lem}

\begin{proof}
We have
\begin{align*}
\mathbb{E}\left[H\right] & =\mathbb{E}\left[h(X_{0}')+G_{Y_{0}'}(X_{1}')\right]\\
 & =\mathbb{E}\left[h(X_{0}')\right]+\mathbb{E}\left[G_{Y_{0}'}(X_{1}')\right]\\
 & =\mathbb{E}\left[h(X_{0}')\right]+\mathbb{E}\left[g_{\star}(X_{1}')-g_{\star}(Y_{0}')\right]\\
 & =\mu(h)+\mu P(g_{\star})-\mu(g_{\star})\\
 & =\mu(h)-\mu(h-\pi(h)) = \pi(h).
\end{align*}
\end{proof}
\begin{thm}
\label{thm:ub-simple}Under Assumption~\ref{assu:tau-moment-kappa},
let $H$ and $\gamma$ be defined
as in Definition~\ref{def:unbiased-estimator},
assume ${\rm d}\mu/{\rm d}\pi\leq M$, $\mathbb{E}_{\gamma}[\tau^{\kappa}]<\infty$,
and  $h\in L^{m}(\pi)$ for some $m>\kappa/(\kappa-1)$. 
Then, $\mathbb{E}\left[H\right]=\pi(h)$,
and for $p\geq1$ such that $\frac{1}{p}>\frac{1}{m}+\frac{1}{\kappa}$,
\begin{equation}
\mathbb{E}\left[\left|H\right|^{p}\right]^{\frac{1}{p}}\leq M^{\frac{1}{p}}\left\Vert h\right\Vert _{L^{m}(\pi)}\zeta\left(\frac{(m-p)\kappa}{mp}\right)\left\{ 1+2M^{\frac{p-m}{mp}}\mathbb{E}_{\gamma}[\tau^{\kappa}]^{\frac{m-p}{mp}}\right\} <\infty.\label{eq:ub-simple-bound}
\end{equation}
\end{thm}

\begin{proof}
Since $m>\kappa/(\kappa-1)$ and $\mathbb{E}_{\gamma}[\tau^{\kappa}]<\infty$,
Theorem~\ref{thm:fishy-estimator} gives $\mathbb{E}\left[G_{y}(x)\right]=g_{y}(x)$
for $\gamma$-almost all $(x,y)$ and Lemma~\ref{lem:unbiasedness-exp}
then implies that $\mathbb{E}[H]=\pi(h)$. By Minkowski's inequality,
Theorem~\ref{thm:fishy-estimator}, ${\rm d}\mu/{\rm d}\pi\leq M$,
$m>p$, and $1\leq\zeta\left(\frac{(m-p)\kappa}{mp}\right)$ we obtain
\begin{align*}
\mathbb{E}\left[\left|H\right|^{p}\right]^{\frac{1}{p}} & \leq\mathbb{E}\left[\left|h(X_{0}')\right|^{p}\right]^{\frac{1}{p}}+\mathbb{E}\left[\left|G_{Y_{0}'}(X_{1}')\right|^{p}\right]^{\frac{1}{p}}\\
 & \leq\left\Vert h\right\Vert _{L^{p}(\mu)}+2\zeta\left(\frac{(m-p)\kappa}{mp}\right)\mathbb{E}_{\gamma}[\tau^{\kappa}]^{\frac{m-p}{mp}}M^{\frac{1}{m}}\left\Vert h\right\Vert _{L^{m}(\pi)}\\
 & \leq M^{\frac{1}{p}}\left\Vert h\right\Vert _{L^{p}(\pi)}+2\zeta\left(\frac{(m-p)\kappa}{mp}\right)\mathbb{E}_{\gamma}[\tau^{\kappa}]^{\frac{m-p}{mp}}M^{\frac{1}{m}}\left\Vert h\right\Vert _{L^{m}(\pi)}\\
 & \leq M^{\frac{1}{p}}\left\Vert h\right\Vert _{L^{m}(\pi)}\left\{ 1+2\zeta\left(\frac{(m-p)\kappa}{mp}\right)M^{\frac{p-m}{mp}}\mathbb{E}_{\gamma}[\tau^{\kappa}]^{\frac{m-p}{mp}}\right\} \\
 & \leq M^{\frac{1}{p}}\left\Vert h\right\Vert _{L^{m}(\pi)}\zeta\left(\frac{(m-p)\kappa}{mp}\right)\left\{ 1+2M^{\frac{p-m}{mp}}\mathbb{E}_{\gamma}[\tau^{\kappa}]^{\frac{m-p}{mp}}\right\} .
\end{align*}
\end{proof}
We now show how Theorem~\ref{thm:ub-simple} can be used to extend
the results to the approximations in Definition~\ref{def:lagged-offset}.

\begin{rem}
\label{rem:L-k-approx-Poisson} The approximation $H_k$ in Definition~\ref{def:lagged-offset} may be viewed as
the sum of $h(X_{k})$ and an unbiased approximation $G_{Y_{k}}^{(L)}(X_{k+L})$
of $g_{Y_{k}}^{(L)}(X_{k+L})=g_{\star}^{(L)}(X_{k+L})-g_{\star}^{(L)}(Y_{k})$,
where $g_{\star}^{(L)}$ is the mean-zero solution of the Poisson
equation for $(P^{L},h)$, i.e.
\[
(I-P^{L})g_{\star}^{(L)}=h-\pi(h).
\]
If $g_{\star}^{(L)}\in L^{p}(\pi)$ then, noting that
\[
(I-P^{L})g_{\star}^{(L)}=(I-P)(\sum_{k=0}^{L-1}P^{k})g_{\star}^{(L)},
\]
we obtain $g_{\star}^{(1)}=(\sum_{k=0}^{L-1}P^{k})g_{\star}^{(L)}\in L^{p}(\pi)$,
since $P$ is a bounded linear operator in $L^{p}(\pi)$.
\end{rem}

The following shows that Theorem~\ref{thm:ub-simple} holds for general
$L\geq1$ and $k\geq0$, and in fact increasing either of these decreases
the upper bound on the moments of the estimator.
\begin{prop}
\label{prop:bound-L-k}Under Assumption~\ref{assu:tau-moment-kappa},
let $h\in L^{m}(\pi)$ for some $m>\kappa/(\kappa-1)$, ${\rm d}\pi_{0}/{\rm d}\pi\leq M$.
For any $L\geq1$, $k\geq0$, $\mathbb{E}[H_{k}]=\pi(h)$ and
for $p\geq1$ such that $\frac{1}{p}>\frac{1}{m}+\frac{1}{\kappa}$,
\begin{equation}
\mathbb{E}\left[\left|H_{k}\right|^{p}\right]^{\frac{1}{p}}\leq M^{\frac{1}{p}}\left\Vert h\right\Vert _{L^{m}(\pi)}\zeta\left(\frac{(m-p)\kappa}{mp}\right)\left\{ 1+2M^{\frac{m-p}{mp}}\mathbb{E}_{\pi\otimes\pi}\left[\left(0 \vee \left\lceil \frac{\tau-k}{L}\right\rceil\right)^{\kappa}\right]^{\frac{m-p}{mp}}\right\} <\infty,\label{eq:HkLp-bound}
\end{equation}
where $\tau=\inf\{t\geq0:X_{t}=Y_{t}\}$ for the Markov chain $(X,Y)$
with Markov kernel $\bar{P}$.
\end{prop}

\begin{proof}
Recall that $a\vee b$ stands for the maximum and $a\wedge b$ for
the minimum of $a$ and $b$. Since $\bar{P}^{L}$ is a coupling of
$P^{L}$ with itself and $P^{L}$ is $\pi$-invariant, and following
Remark~\ref{rem:L-k-approx-Poisson}, we seek to apply Theorem~\ref{thm:ub-simple}
to the approximation
\[
h(X_{k})+G_{Y_{k}}^{(L)}(X_{k+L}),
\]
where the second term is obtained by considering a chain $(X^{(L,k)},Y^{(L,k)})$
with Markov transition kernel $\bar{P}^{L}$, with $(X_{0}^{(L,k)},Y_{0}^{(L,k)})\sim\gamma_{k}^{(L)}$,
and $\gamma_{k}^{(L)}=\gamma^{(L)}\bar{P}^{k}=(\pi_{0}P^{L}\otimes\pi_{0})\bar{P}^{k}$.
This is analogous to the coupled chain described in Section~\ref{subsec:Coupled-Markov-chains}
which is used in Definition~\ref{def:Gyx}. We define
\[
\tau_{L,k}=\inf\{t\geq0:X_{t}^{(L,k)}=Y_{t}^{(L,k)}\},
\]
and we seek to verify that Assumption~\ref{assu:tau-moment-kappa}
holds for $\tau_{L,k}$ and obtain a finite bound for $\mathbb{E}_{\gamma_{k}^{(L)}}^{(L)}[\tau_{L,k}^{\kappa}]$,
where $\mathbb{E}^{(L)}_{\gamma_{k}^{(L)}}$ denotes expectation w.r.t. to the law of
$(X^{(L,k)},Y^{(L,k)})$. In Theorem~\ref{thm:ub-simple}, we take
$\mu=\pi_{0}P^{k}$.

If we define $(X,Y)$ with Markov transition kernel $\bar{P}$ and
$(X_{0},Y_{0})\sim\gamma^{(L)}$, and define $\tau=\inf\{t\geq0:X_{t}=Y_{t}\}$,
we observe that $(X^{(L,k)},Y^{(L,k)})$ may be taken as a skeleton
of this chain, i.e.
\[
(X_{t}^{(L,k)},Y_{t}^{(L,k)})\overset{d}{=}(X_{k+tL},Y_{k+tL}),\qquad t\geq0.
\]
Therefore, we may deduce that
\[
\tau_{L,k}\overset{d}{=}0\vee\left\lceil \frac{\tau-k}{L}\right\rceil ,
\]
and therefore $\mathbb{E}_{\pi\otimes\pi}\left[\tau^{\kappa}\right]<\infty$
implies $\mathbb{E}_{\pi\otimes\pi}^{(L)}\left[\tau_{L,k}^{\kappa}\right]<\infty$.
By Lemma~\ref{lem:gamma-M-sq-independent-init}, we have ${\rm d}\mu/{\rm d}\pi={\rm d}\pi_{0}P^{k}/{\rm d}\pi\leq M$
and ${\rm d}\gamma^{(L)}/{\rm d}(\pi\otimes\pi)\leq M^{2}$, and hence
\[
\mathbb{E}_{\gamma_{k}^{(L)}}^{(L)}\left[\tau_{L,k}^{\kappa}\right]=\mathbb{E}_{\gamma^{(L)}}\left[\left(0\vee\left\lceil \frac{\tau-k}{L}\right\rceil \right)^{\kappa}\right]\leq M^{2}\mathbb{E}_{\pi\otimes\pi}\left[\left(0\vee\left\lceil \frac{\tau-k}{L}\right\rceil \right)^{\kappa}\right].
\]
It follows from (\ref{eq:ub-simple-bound}) that for $p\geq1$ such
that $\frac{1}{p}>\frac{1}{m}+\frac{1}{\kappa}$,
\begin{align*}
\mathbb{E}\left[\left|H_{k}\right|^{p}\right]^{\frac{1}{p}} & \leq M^{\frac{1}{p}}\left\Vert h\right\Vert _{L^{m}(\pi)}\zeta\left(\frac{(m-p)\kappa}{mp}\right)\left\{ 1+2M^{\frac{p-m}{mp}}\mathbb{E}_{\gamma^{(L)}}^{(L)}[\tau_{L,k}^{\kappa}]^{\frac{m-p}{mp}}\right\} \\
 & \leq M^{\frac{1}{p}}\left\Vert h\right\Vert _{L^{m}(\pi)}\zeta\left(\frac{(m-p)\kappa}{mp}\right)\left\{ 1+2M^{\frac{m-p}{mp}}\mathbb{E}_{\pi\otimes\pi}\left[\left(0 \vee \left\lceil \frac{\tau-k}{L}\right\rceil\right)^{\kappa}\right]^{\frac{m-p}{mp}}\right\} \\
 & <\infty,
\end{align*}
and that $m>\kappa/(\kappa-1)$ is sufficient for $\mathbb{E}[H_{k}]=\pi(h)$.
\end{proof}
The following shows that the unbiased signed measure \eqref{eq:pihatmeasure} 
is indeed unbiased
for functions with suitably large moments when $\kappa$ is large
enough, and that moments of averaged unbiased estimators are finite
under the same conditions as for $H$.
\begin{prop}
\label{prop:lag-average-transference}Under Assumption~\ref{assu:tau-moment-kappa},
let $h\in L^{m}(\pi)$ for some $m>\kappa/(\kappa-1)$, ${\rm d}\pi_{0}/{\rm d}\pi\leq M$.
Then for any $k,\ell\in\mathbb{N}$ with $\ell\geq k$, $\mathbb{E}[H_{k:\ell}]=\pi(h)$
and for $p\geq1$ such that $\frac{1}{p}>\frac{1}{m}+\frac{1}{\kappa}$,
\[
\mathbb{E}\left[\left|H_{k:\ell}\right|^{p}\right]^{\frac{1}{p}}\leq M^{\frac{1}{p}}\left\Vert h\right\Vert _{L^{m}(\pi)}\zeta\left(\frac{(m-p)\kappa}{mp}\right)\left\{ 1+2M^{\frac{m-p}{mp}}\mathbb{E}_{\pi\otimes\pi}\left[\left(0 \vee \left\lceil \frac{\tau-k}{L}\right\rceil\right)^{\kappa}\right]^{\frac{m-p}{mp}}\right\} <\infty.
\]
\end{prop}

\begin{proof}
By Proposition~\ref{prop:bound-L-k}, if $m>\kappa/(\kappa-1)>1$
then $\mathbb{E}[H_{t}]=\pi(h)$ for all $t\in\{k,\ldots,\ell\}$
and so $\mathbb{E}[H_{k:\ell}]=\pi(h)$. Using Minkowski's inequality,
we have
\[
\mathbb{E}\left[\left|H_{k:\ell}\right|^{p}\right]^{\frac{1}{p}}=\frac{1}{\ell-k+1}\mathbb{E}\left[\left|\sum_{t=k}^{\ell}H_{t}\right|^{p}\right]^{\frac{1}{p}}\leq\frac{1}{\ell-k+1}\sum_{t=k}^{\ell}\mathbb{E}\left[\left|H_{t}\right|^{p}\right]^{\frac{1}{p}},
\]
and we may conclude using the fact that the upper bound in
(\ref{eq:HkLp-bound}) is non-increasing in $k$.
\end{proof}

\begin{proof}[Proof of Proposition~\ref{prop:cltunbiasedmcmc}]
We may write $H_{t}=h(X_{t})+B_{t}$ for a random variable $B_t$
defined on the right-hand side of \eqref{eq:H_k^L} as  $\sum_{j=1}^{\infty} h(X_{t+jL}) - h(Y_{t+(j-1)L})$.
Here $(X_{t+L},Y_{t})$ is a Markov chain with transition kernel $\bar{P}$, as generated by Algorithm~\ref{alg:coupledchains}.
It follows that we may write
\[
\sqrt{\ell-k+1}\left( H_{k:\ell}-\pi(h)\right) =\frac{1}{\sqrt{\ell-k+1}}\sum_{t=k}^{\ell}h_{0}(X_{t})+\frac{1}{\sqrt{\ell-k+1}}\sum_{t=k}^{\ell}B_{t}.
\]
Since $B_{t}=0$ for $t\geq\tau$, we have $\frac{1}{\sqrt{\ell-k+1}}\sum_{t=k}^{\ell}B_{t}\to0$
in probability as $\ell\to\infty$, and by Theorem~\ref{thm:clt-kappa}
and Slutsky, we may conclude.
\end{proof}

We now demonstrate that estimators associated with subsampling the
unbiased signed measure $\hat{\pi}$ in \eqref{eq:pihatmeasure_simplified}
are unbiased and have finite $p$th moments under the same conditions
as the standard estimator. We start with a technical lemma stating that
the weights in $\hat{\pi}$ are bounded by expressions that depend on $k,\ell,L$ but not on $\tau$.

\begin{lem}
\label{lem:signed-measure-weight-bounds}Let $\hat{\pi}=\sum_{i=1}^{N}\omega_{i}\delta_{Z_{i}}$
be the unbiased signed measure in \eqref{eq:pihatmeasure_simplified}.
Then
\[
\frac{1}{\ell-k+1}\leq\min_{i}\left|\omega_{i}\right|\leq\max_{i}\left|\omega_{i}\right|\leq\frac{1}{\ell-k+1}\left(1+\frac{\ell-k}{L}\right)
\]
\end{lem}

\begin{proof}
For the lower bound it suffices to note that if a weight is non-zero,
its absolute value is necessarily greater than or equal to $1/(\ell-k+1)$. For the
upper bound, we find that $\left|\omega_{i}\right|\leq1/(\ell-k+1)$
for the first $\ell-k+1$ points, and for the remaining points,
\begin{align*}
\left|\omega_{i}\right| & \leq\frac{\left\lfloor (t-k)/L\right\rfloor -\left\lceil \max(L,t-\ell)/L\right\rceil +1}{\ell-k+1}\\
 & \leq\frac{(t-k)/L-\max(L,t-\ell)/L+1}{\ell-k+1}\\
 & \leq\frac{(t-k)/L-(t-\ell)/L+1}{\ell-k+1}\\
 & =\frac{(\ell-k)/L+1}{\ell-k+1} =\frac{1}{\ell-k+1}\left(1+\frac{\ell-k}{L}\right).
\end{align*}
\end{proof}

\begin{prop}
\label{prop:subsample}Under Assumption~\ref{assu:tau-moment-kappa},
let $h\in L^{m}(\pi)$ for some $m>\kappa/(\kappa-1)$, ${\rm d}\pi_{0}/{\rm d}\pi\leq M$,
$k,\ell\in\mathbb{N}$ with $k\leq\ell$, and $\hat{\pi}=\sum_{i=1}^{N}\omega_{i}\delta_{Z_{i}}$
be the unbiased signed measure in \eqref{eq:pihatmeasure_simplified}.
Define for some
$R\geq1$,
\[
S_{R}=\frac{1}{R}\sum_{i=1}^{R} \xi_{I_i}^{-1} \omega_{I_{i}}h(Z_{I_{i}}),
\]
where $I_{1},\ldots,I_{R}$ are conditionally independent ${\rm Categorical}\{\xi_{1},\ldots,\xi_{N}\}$
variables satisfying
\[
\frac{a}{N}\leq\min_{i}\xi_{i}\leq\max_{i}\xi_{i}\leq\frac{b}{N},
\]
for some constants $0<a\leq b<\infty$ that may be functions of $k,\ell,L$ but not $\tau$. Then $\mathbb{E}[S_{R}]=\pi(h)$
and for $p\geq1$ such that $\frac{1}{p}>\frac{1}{m}+\frac{1}{\kappa}$,
\begin{eqnarray*}
\mathbb{E}\left[\left|S_{R}\right|^{p}\right]^{\frac{1}{p}} &\leq& a^{-1}b^{\frac{\kappa-p}{\kappa p}}\left\{ \frac{1}{\ell-k+1}\left(1+\frac{\ell-k}{L}\right)\right\} \mathbb{E}\left[\left|N\right|^{\kappa}\right]^{\frac{m-p}{mp}} \\
&& \times M^{\frac{1}{m}}\left\Vert h\right\Vert _{L^{m}(\pi)}\left\{ \sum_{i=1}^{\infty}\frac{1}{i^{1+\varepsilon}}\right\} ^{\frac{\kappa-p}{\kappa p}}<\infty,
\end{eqnarray*}
where $\varepsilon=\frac{\kappa}{m(\kappa-p)}\cdot(\kappa m-\kappa p-mp)$.
\end{prop}

\begin{proof}
The signed measure $\hat{\pi}$ is such that $\hat{\pi}(h)=H_{k:\ell}$.
Hence, Proposition~\ref{prop:lag-average-transference} may be applied
to show that $\mathbb{E}[\hat{\pi}(h)]=\pi(h)$ when $m>\kappa/(\kappa-1)$.
It follows that
\[
\mathbb{E}[S_{R}]=\mathbb{E}\left[\mathbb{E}[S_{R}\mid\hat{\pi}]\right]=\mathbb{E}\left[\hat{\pi}(h)\right]=\pi(h).
\]
To determine that $\mathbb{E}\left[\left|S_{R}\right|^{p}\right]<\infty$
for $1\geq\frac{1}{p}>\frac{1}{m}+\frac{1}{\kappa}$ we define $S_{1}=\omega_{I}\xi_{I}^{-1}h(Z_{I})$
where $I\mid\hat{\pi}\sim{\rm Categorical}\{\xi_{1},\ldots,\xi_{N}\}$.
Then $S_{R}$ is less than $S_{1}$ in the convex order, i.e. $S_{R}\leq_{{\rm cx}}S_{1}$
\citep[see, e.g.,][]{shaked2007stochastic} for any $R\geq2$ since
one may define $S_{1}$ by drawing $I\sim {R}^{-1}\sum_{r=1}^{R}\delta_{I_{r}}$
and it is then clear that $\mathbb{E}[S_{1}\mid\sigma(I_{1:R},\omega_{1:N},Z_{1:N})]=S_{R}$.
It follows that $\mathbb{E}\left[\left|S_{R}\right|^{p}\right]\leq\mathbb{E}\left[\left|S_{1}\right|^{p}\right]$
since $x\mapsto\left|x\right|^{p}$ is convex for $p\geq1$. Using Lemma~\ref{lem:signed-measure-weight-bounds},
the fact that $N\leq\ell-k+1+2\tau$ and H\"older's inequality,
\begin{align*}
\mathbb{E}\left[\left|S_{1}\right|^{p}\right] & =\mathbb{E}\left[\left|\omega_{I}\xi_{I}^{-1}h(Z_{I})\right|^{p}\right]\\
 & \leq a^{-p}\mathbb{E}\left[\left|N\omega_{I}h(Z_{I})\right|^{p}\right]\\
 & \leq a^{-p}\left\{ \frac{1}{\ell-k+1}\left(1+\frac{\ell-k}{L}\right)\right\} ^{p}\mathbb{E}\left[\left|N\right|^{\kappa}\right]^{\frac{p}{\kappa}}\mathbb{E}\left[\left|h(Z_{I})\right|^{\frac{\kappa p}{\kappa-p}}\right]^{1-\frac{p}{\kappa}},
\end{align*}
and since $\mathbb{E}\left[\left|N\right|^{\kappa}\right]<\infty$
by Assumption~\ref{assu:tau-moment-kappa} it remains to show that
$\mathbb{E}\left[\left|h(Z_{I})\right|^{\frac{\kappa p}{\kappa-p}}\right]<\infty$.
Now, since $\mathbb{P}(I=i\mid\sigma(N,Z_{1},\ldots,Z_{N}))\leq bi^{-1}$
for all $1\leq i\leq N$, and taking any $\varepsilon\in(0,\kappa)$, we
obtain for $q=\kappa p/(\kappa-p)$,
\begin{align*}
\mathbb{E}\left[\left|h(Z_{I})\right|^{q}\right] 
 & =\sum_{i=1}^{\infty}\mathbb{E}\left[\mathds{1}(N\geq i,I=i)\left|h(Z_{i})\right|^{q}\right]\\
 & \leq b\sum_{i=1}^{\infty}\frac{1}{i}\mathbb{E}\left[\mathds{1}(N\geq i)\left|h(Z_{i})\right|^{q}\right]\\
 & \leq b\sum_{i=1}^{\infty}\frac{1}{i}\mathbb{E}\left[\left(\frac{N}{i}\right)^{\varepsilon}\left|h(Z_{i})\right|^{q}\right]\\
 & \leq b\mathbb{E}\left[N^{\kappa}\right]^{\frac{\varepsilon}{\kappa}}\sum_{i=1}^{\infty}\frac{1}{i^{1+\varepsilon}}\mathbb{E}\left[\left|h(Z_{i})\right|^{\frac{\kappa q}{\kappa-\varepsilon}}\right]^{1-\frac{\varepsilon}{\kappa}}.
\end{align*}
For $\frac{1}{p}>\frac{1}{m}+\frac{1}{\kappa}$, it follows that we
may take $\varepsilon=\kappa(m-q)/m$ so that $\kappa q/(\kappa-\varepsilon)=m$.
Moreover, since $Z_{i}\sim\pi_{0}P^{t}$ for some $t\in\mathbb{N}$,
$\mathbb{E}\left[\left|h(Z_{i})\right|^{m}\right]\leq M\pi(\left|h\right|^{m})$
by Lemma~\ref{lem:sup-finite} and so using the fact that $(1-\frac{\varepsilon}{\kappa})(1-\frac{p}{\kappa})=\frac{p}{m}$,
we obtain
\begin{align*}
&\mathbb{E}\left[\left|S_{1}\right|^{p}\right] \\
& \leq a^{-p}\left\{ \frac{1}{\ell-k+1}\left(1+\frac{\ell-k}{L}\right)\right\} ^{p}\mathbb{E}\left[\left|N\right|^{\kappa}\right]^{\frac{p}{\kappa}}\left\{ b\mathbb{E}\left[N^{\kappa}\right]^{\frac{\varepsilon}{\kappa}}M^{1-\frac{\varepsilon}{\kappa}}\pi(\left|h\right|^{m})^{1-\frac{\varepsilon}{\kappa}}\sum_{i=1}^{\infty}\frac{1}{i^{1+\varepsilon}}\right\} ^{1-\frac{p}{\kappa}}\\
 & =a^{-p}b^{\frac{\kappa-p}{\kappa}}\left\{ \frac{1}{\ell-k+1}\left(1+\frac{\ell-k}{L}\right)\right\} ^{p}\mathbb{E}\left[\left|N\right|^{\kappa}\right]^{1-\frac{p}{m}}M^{\frac{p}{m}}\left\Vert h\right\Vert _{L^{m}(\pi)}^{p}\left\{ \sum_{i=1}^{\infty}\frac{1}{i^{1+\varepsilon}}\right\} ^{1-\frac{p}{\kappa}},
\end{align*}
which is finite. We thus conclude.
\end{proof}
\begin{example}
Natural choices of $\xi_{i}$ are to take $\xi_i = 1/N$ or $\xi_{i}\propto\left|\omega_{i}\right|$.
In the latter case, it follows from Lemma~\ref{lem:signed-measure-weight-bounds} that
\[
\max_{i}\frac{\left|\omega_{i}\right|}{\sum_{j=1}^{N}\left|\omega_{j}\right|}\leq\frac{\frac{1}{\ell-k+1}\left(1+\frac{\ell-k}{L}\right)}{N\frac{1}{\ell-k+1}}=\frac{1+(\ell-k)/L}{N},
\]
and
\[
\min_{i}\frac{\left|\omega_{i}\right|}{\sum_{j=1}^{N}\left|\omega_{j}\right|}\geq\frac{\frac{1}{\ell-k+1}}{N\frac{1}{\ell-k+1}\left(1+\frac{\ell-k}{L}\right)}=\frac{1}{N(1+\frac{\ell-k}{L})},
\]
and so one may take $a=1/(1+\frac{\ell-k}{L})$ and $b=1+(\ell-k)/L$
in Proposition~\ref{prop:subsample} for this choice.
\end{example}

\subsection{\label{subsec:estimator-pi-h-g}Unbiased approximation of \texorpdfstring{$\pi(h_{0}\cdot g)$}{pi(h0g)}}

We now look at combinations of unbiased fishy function estimation
and unbiased estimation of $\pi(h)$. This will involve estimators
$G_{y}(x)$ at random points $x$. We first introduce
an alternative representation to avoid ambiguity in the
following developments. We define a probability measure $Q$ such
that, with $U\sim Q$,
\[
\bar{g}_{y}(x,U)\overset{d}{=}G_{y}(x),
\]
with $G_y(x)$ in Definition~\ref{def:Gyx}, and we define an extended distribution $\check{\pi}({\rm d}x,{\rm d}u)=\pi({\rm d}x)Q({\rm d}u)$
with $\check{\pi}$-invariant Markov kernel $T(x,u;{\rm d}y,{\rm d}v)=P(x,{\rm d}y)Q({\rm d}v)$,
and coupled Markov kernel
\[
\begin{aligned}
\bar{T}(x,u,y,v;{\rm d}x',{\rm d}u',{\rm d}y',{\rm d}v') =\;&
\bar{P}(x,y;{\rm d}x',{\rm d}y') \, Q({\rm d}u') \\
&\times \left\{ 
  \mathds{1}(x'=y')\, \delta_{u'}({\rm d}v') 
  + \mathds{1}(x'\neq y')\, Q({\rm d}v') 
\right\}.
\end{aligned}
\]
Thus, if $(X',U',Y',V')\sim\bar{T}(x,y,u,v;\cdot)$
then $(X',Y')\sim\bar{P}(x,y;\cdot)$ and if $X'=Y'$ then $V'=U'\sim Q$
but if $X'\neq Y'$ then $U',V'\sim Q$ independently.

We denote by $\check{G}_{y,v}^{f}(x,u)$
an unbiased approximation of a fishy function associated with the transition $T$ and
test function $f$,
as opposed to $P$ and $h$, i.e. 
\[
\mathbb{E}\left[\check{G}_{y,v}^{f}(x,u)\right]=\check{g}_{y,v}^{f}(x,u):=\check{g}_{\star}^{f}(x,u)-\check{g}_{\star}^{f}(y,v),
\]
where $\check{g}_{\star}^{f}(x,u):=\sum_{t=0}^{\infty}T^{t}f_{0}(x,u)$,
with $f_{0}:=f-\check{\pi}(f)$  for
$f$ in $L_{0}^{1}(\check{\pi})$.

The Markov chain $(X,U,Y,V)$ has the same meeting time as the Markov
chain $(X,Y)$ by construction, so Assumption~\ref{assu:tau-moment-kappa}
holds for this chain with $\pi$ replaced by $\check{\pi}$. Similarly
if ${\rm d}\mu/{\rm d}\pi\leq M$ then with $\check{\mu}=\mu\otimes Q$
and $\check{\pi}=\pi\otimes Q$ we have ${\rm d}\check{\mu}/{\rm d}\check{\pi}\leq M$.
Hence, we may apply Proposition~\ref{prop:lag-average-transference}
or Proposition~\ref{prop:subsample} to deduce lack-of-bias of an
appropriate approximation of $\check{\pi}(\phi)$ and finite $p$th
moments if $\phi\in L^{m}(\check{\pi})$ and $1\geq\frac{1}{p}>\frac{1}{\kappa}+\frac{1}{m}$.

The following two lemmas provide conditions for $\bar{g}_{y}\in L^{q}(\check{\pi})$
and $h\cdot \bar{g}_{y}:(x,u)\mapsto h(x)\bar{g}_y(x,u)\in L^{s}(\check{\pi})$, which are used to analyze
both the MCMC estimator of $v(P,h)$ of Section \ref{subsec:consistentavar} (EPAVE)
and the unbiased estimators of Section \ref{subsec:unbiasedavar} (UPAVE).
\begin{lem}
\label{lem:check-gy-moment} Let $h\in L^{m}(\pi)$ for some $m>1$.
Under Assumption~\ref{assu:tau-moment-kappa}, for $\pi$-almost
all $y$, $\bar{g}_{y}\in L^{q}(\check{\pi})$ for $q\geq1$ such that $\frac{1}{q}>\frac{1}{m}+\frac{1}{\kappa}$.
\end{lem}

\begin{proof}
We observe that $\check{\pi}(\left|\bar{g}_y\right|^{q})^{\frac{1}{q}}=\mathbb{E}_{\pi}\left[\left|G_{y}(X_{0})\right|^{q}\right]^{\frac{1}{q}}$
and then we apply  Theorem~\ref{thm:fishy-estimator}, part 3.,
with $\gamma=\pi\otimes\delta_{y}$ for $\pi$-almost all $y$. 
\end{proof}
\begin{lem}
\label{lem:check-phi-1-moment}Let $h\in L^{m}(\pi)$ for some $m>1$.
Under Assumption~\ref{assu:tau-moment-kappa}, with $\phi_{1}(x,u)=h(x)\bar{g}_{y}(x,u)$,
for $\pi$-almost all $y$, $\phi_{1}\in L^{s}(\check{\pi})$ for
$s\geq1$ such that $\frac{1}{s}>\frac{2}{m}+\frac{1}{\kappa}$.
\end{lem}

\begin{proof}
Let $f(x,u)=h(x)$. By H\"older's inequality with $\delta\in(0,1)$,
\begin{align*}
\check{\pi}(\left|\phi_{1}\right|^{s}) & \leq\check{\pi}(\left|f\right|^{\frac{s}{\delta}})^{\delta}\check{\pi}(\left|\bar{g}_y\right|^{\frac{s}{1-\delta}})^{1-\delta}=\pi(\left|h\right|^{\frac{s}{\delta}})^{\delta}\check{\pi}(\left|\bar{g}_y\right|^{\frac{s}{1-\delta}})^{1-\delta}.
\end{align*}
With $q=s/(1-\delta)$, we deduce by Lemma~\ref{lem:check-gy-moment}
that if $\frac{1}{q}>\frac{1}{m}+\frac{1}{\kappa}$ then $\check{\pi}(\left|\bar{g}_{y}\right|^{q})<\infty$
for $\pi$-almost all $y$. On the other hand, if $s<m\delta$, then
$\pi(\left|h\right|^{\frac{s}{\delta}})<\infty$. Taking $\delta=\kappa/(2\kappa+m)$
we find that $m\delta=(1-\delta)(\frac{1}{m}+\frac{1}{\kappa})^{-1}$,
and this implies that $\frac{1}{s}>\frac{2}{m}+\frac{1}{\kappa}$
is sufficient for $\check{\pi}(\left|\phi_{1}\right|^{s})<\infty$.
\end{proof}

\begin{defn}[Estimator of $\pi(h_{0}\cdot g_{y})$]
\label{def:zeta-basic}Let $H$ be an unbiased estimator of $\pi(h)$
as defined in Definition~\ref{def:lagged-offset}. Let $\phi_{H}(x,u)=(h(x)-H)\bar{g}_{y}(x,u)$
and, with random variables independent to those used to define $H$,
let $\Phi_{R}$ (resp. $\Phi_{0}$) be the approximation corresponding
to $S_{R}$ (resp. $H_{k:\ell}$) of $\check{\pi}(\phi_{H})$
in Proposition~\ref{prop:subsample} (resp. Proposition~\ref{prop:lag-average-transference})
with the unbiased signed measure approximating $\check{\pi}$ involving
random variables independent  to those used to define $H$.
\end{defn}

\begin{prop}
\label{prop:Phi-R-unbiased-moments}Under Assumption~\ref{assu:tau-moment-kappa}
with $\kappa>2$, let $h\in L^{m}(\pi)$ for some $m>2\kappa/(\kappa-2)$, and ${\rm d}\pi_0/{\rm d}\pi\leq M$.
For any $R\in\{0,1,\ldots\}$, $\Phi_R$ in Definition~\ref{def:zeta-basic} satisfies,
\begin{enumerate}
\item $\mathbb{E}\left[\Phi_{R}\right]=\pi(h_{0}\cdot g_{y})$.
\item For $p\geq1$ such that $\frac{1}{p}>\frac{2}{\kappa}+\frac{2}{m}$,
$\mathbb{E}\left[\left|\Phi_{R}\right|^{p}\right]^{\frac{1}{p}}<\infty$.
\end{enumerate}
\end{prop}

\begin{proof}
We first determine $s\geq1$ such that $\phi_{H}\in L^{s}(\check{\pi})$.
Let $\phi_{1}(x,u)=h(x)\bar{g}_{y}(x,u)$ and $\phi_{2,c}(x,u)=c\bar{g}_{y}(x,u)$,
so that $\phi_{H}=\phi_{1}-\phi_{2,H}$. We find that for a fixed
$H$ and $s\geq1$,
\[
\check{\pi}(\left|\phi_{H}\right|^{s})^{\frac{1}{s}}\leq\check{\pi}(\left|\phi_{1}\right|^{s})^{\frac{1}{s}}+\check{\pi}(\left|\phi_{2,H}\right|^{s})^{\frac{1}{s}}=\check{\pi}(\left|\phi_{1}\right|^{s})^{\frac{1}{s}}+\left|H\right|\check{\pi}(\left|\bar{g}_{y}\right|^{s})^{\frac{1}{s}}.
\]
By Lemma~\ref{lem:check-gy-moment}, $\check{\pi}(\left|\bar{g}_{y}\right|^{s})<\infty$
if $\frac{1}{s}>\frac{1}{m}+\frac{1}{\kappa}$. By Lemma~\ref{lem:check-phi-1-moment},
$\check{\pi}(\left|\phi_{1}\right|^{s})<\infty$ if $\frac{1}{s}>\frac{2}{m}+\frac{1}{\kappa}$.
Hence, $\phi_{H}\in L^{s}(\check{\pi})$ for $\frac{1}{s}>\frac{2}{m}+\frac{1}{\kappa}$.
It follows that if $\frac{1}{p}>\frac{2}{m}+\frac{2}{\kappa}$ then
there exists $s\in[1,(\frac{2}{m}+\frac{1}{\kappa})^{-1})$ such that
by Proposition~\ref{prop:lag-average-transference},
\[
\mathbb{E}\left[\left|\Phi_{0}\right|^{p}\mid H\right]\leq M\left\Vert \phi_{H}\right\Vert _{L^{s}(\check{\pi})}^{p}\zeta\left(\frac{(s-p)\kappa}{sp}\right)^{p}\left\{ 1+2M^{\frac{p-s}{sp}}\mathbb{E}_{\gamma}\left[\tau^{\kappa}\right]^{\frac{s-p}{sp}}\right\} ^{p}<\infty,
\]
and for $R\geq1$, by Proposition~\ref{prop:subsample},
\[
\mathbb{E}\left[\left|S_{R}\right|^{p}\mid H\right]\leq\mathbb{E}\left[\left|N\right|^{\kappa}\right]^{\frac{s-p}{s}}M^{\frac{p}{s}}\left\Vert \phi_{H}\right\Vert _{L^{s}(\check{\pi})}^{p}\left\{ \sum_{i=1}^{\infty}\frac{1}{i^{1+\varepsilon}}\right\} ^{\frac{\kappa-p}{\kappa}}<\infty,
\]
where $\varepsilon=\frac{\kappa}{s(\kappa-p)}\cdot(\kappa s-\kappa p-sp)$.
It follows that if $\mathbb{E}\left[\left\Vert \phi_{H}\right\Vert _{L^{s}(\check{\pi})}^{p}\right]<\infty$
then $\mathbb{E}\left[\left|\Phi_{R}\right|^{p}\right]=\mathbb{E}\left[\mathbb{E}\left[\left|\Phi_{R}\right|^{p}\mid H\right]\right]<\infty$
for $R\geq0$. We have
\begin{align*}
\mathbb{E}\left[\left\Vert \phi_{H}\right\Vert _{L^{s}(\check{\pi})}^{p}\right] & \leq\mathbb{E}\left[\left\{ \check{\pi}(\left|\phi_{1}\right|^{s})^{\frac{1}{s}}+\left|H\right|\check{\pi}(\left|\bar{g}_{y}\right|^{s})^{\frac{1}{s}}\right\} ^{p}\right]\\
 & \leq2^{p-1}\left\{ \check{\pi}(\left|\phi_{1}\right|^{s})^{\frac{p}{s}}+\mathbb{E}\left[\left|H\right|^{p}\right]\check{\pi}(\left|\bar{g}_{y}\right|^{s})^{\frac{p}{s}}\right\} ,
\end{align*}
and $\mathbb{E}\left[\left|H\right|^{p}\right]<\infty$ by Proposition~\ref{prop:lag-average-transference}
if $\frac{1}{p}>\frac{1}{m}+\frac{1}{\kappa}$ , which imposes no
additional constraints on $p$ or $s$. Hence, we may conclude that
for $\frac{1}{p}>\frac{2}{m}+\frac{2}{\kappa}$, $\mathbb{E}\left[\left|\Phi_{R}\right|^{p}\right]<\infty$
for $R\geq0$. For the lack-of-bias property, we consider $p=1$ and
if $m>2\kappa/(\kappa-2)$ then, by Proposition~\ref{prop:lag-average-transference}
and Proposition~\ref{prop:subsample},
\[
\mathbb{E}[\Phi_{R}\mid H]=\check{\pi}(\phi_{H})=\pi(h\cdot g_{y})-H\pi(g_{y}),
\]
and since $\mathbb{E}[H]=\pi(h)$, we may conclude.
\end{proof}

\subsection{\label{subsec:Unbiased-asymptotic-variance}Unbiased Poisson asymptotic
variance estimator}

We show that the basic and subsampled unbiased Poisson asymptotic variance
estimators are indeed unbiased, and have finite $p$th moments under
the same conditions.
\begin{defn}
\label{def:avar-estimator}Let $H_{1}$ and $H_{2}$ be two independent
unbiased estimators of $\pi(h)$ and $S$ be an unbiased estimator
of $\pi(h^{2})$, all of the type described in Proposition~\ref{prop:lag-average-transference}.
Let $R\in\{0,1,\ldots\}$, $\Phi_{R}$ be as in Definition~\ref{def:zeta-basic}
and define
\[
V_{R}=-(S-H_{1}H_{2})+2\Phi_{R}.
\]
\end{defn}

\begin{lem}
\label{lem:v-pi-h-estimator}Under Assumption~\ref{assu:tau-moment-kappa},
let $h\in L^{m}(\pi)$ for some $m>2\kappa/(\kappa-1)$, and ${\rm d}\pi_0/{\rm d} \pi\leq M$. Let $H_{1}$,
$H_{2}$ and $S$ be as in Definition~\ref{def:avar-estimator}.
Then
\begin{enumerate}
\item $\mathbb{E}\left[S-H_{1}H_{2}\right]=\pi(h^{2})-\pi(h)^{2}=v(\pi,h)$.
\item For $p\geq1$ such that $\frac{1}{p}>\frac{2}{m}+\frac{1}{\kappa}$,
$\mathbb{E}\left[\left|S-H_{1}H_{2}\right|^{p}\right]<\infty$.
\end{enumerate}
\end{lem}

\begin{proof}
We note that if $h\in L^{m}(\pi)$ then $h^{2}\in L^{m/2}(\pi)$.
Hence, we deduce by Proposition~\ref{prop:lag-average-transference}
that if $m/2>\kappa/(\kappa-1)$, i.e. $m>2\kappa/(\kappa-1)$ then
$\mathbb{E}[S]=\pi(h^{2})$ and also $\mathbb{E}[H_{1}H_{2}]=\mathbb{E}[H_{1}]\mathbb{E}[H_{2}]=\pi(h)^{2}$.
By Minkowski's inequality, for $p\geq1$:
\begin{align*}
\mathbb{E}\left[\left|S-H_{1}H_{2}\right|^{p}\right]^{\frac{1}{p}} & \leq\mathbb{E}\left[\left|S\right|^{p}\right]^{\frac{1}{p}}+\mathbb{E}\left[\left|H_{1}H_{2}\right|^{p}\right]^{\frac{1}{p}}\\
 & =\mathbb{E}\left[\left|S\right|^{p}\right]^{\frac{1}{p}}+\mathbb{E}\left[\left|H_{1}\right|^{p}\right]^{\frac{1}{p}}\mathbb{E}\left[\left|H_{2}\right|^{p}\right]^{\frac{1}{p}},
\end{align*}
and the terms on the right-hand side are finite by Proposition~\ref{prop:lag-average-transference}
if $\frac{1}{p}>\frac{2}{m}+\frac{1}{\kappa}$.
\end{proof}
\begin{thm}
\label{thm:avar-unbiased-moments}Under Assumption~\ref{assu:tau-moment-kappa},
let $h\in L^{m}(\pi)$ for some $m>2\kappa/(\kappa-2)$, and ${\rm d}\pi_0/{\rm d} \pi\leq M$. Let $V_{R}$
be as in Definition~\ref{def:avar-estimator} with $R\geq0$. For
$\pi$-almost all $y$:
\begin{enumerate}
\item $\mathbb{E}\left[V_{R}\right]=v(P,h)$.
\item For $p\geq1$ such that $\frac{1}{p}>\frac{2}{m}+\frac{2}{\kappa}$,
$\mathbb{E}\left[\left|V_{R}\right|^{p}\right]<\infty$.
\end{enumerate}
\end{thm}

\begin{proof}
By Lemma~\ref{lem:v-pi-h-estimator}, $\mathbb{E}[S-H_{1}H_{2}]=v(\pi,h)$
if $m>2\kappa/(\kappa-2)$ and $\mathbb{E}\left[\left|S-H_{1}H_{2}\right|^{p}\right]<\infty$
for the range of $p$ given. By Proposition~\ref{prop:Phi-R-unbiased-moments},
for $\pi$-almost all $y$, $\mathbb{E}[\Phi_{R}]=\pi(h_{0}\cdot g_{y})$
if $m>2\kappa/(\kappa-2)$ and $\mathbb{E}[\left|\Phi_{R}\right|^{p}]$
is finite for the range of $p$ given. Hence, for $m>2\kappa/(\kappa-2)$
we have $\mathbb{E}[V_{R}]=v(P,h)$ and we may conclude the finiteness
of $\mathbb{E}\left[\left|V_{R}\right|^{p}\right]$ by Minkowski's
inequality.
\end{proof}
\begin{rem}
\label{rem:average-avar-estimators}By Minkowski's inequality we may
similarly conclude that any average of estimators of the form given
in Definition~\ref{def:avar-estimator} also has lack-of-bias and
moments implied by Theorem~\ref{thm:avar-unbiased-moments}, and
hence that the results apply to (\ref{eq:unbiasedasymptvar}) when
$\xi_{n}^{(j)}=1/N^{(j)}$ for $n\in\{1,\ldots,N^{(j)}\}$.

\end{rem}

\subsection{\label{subsec:mcmc-estimator-avar}Ergodic Poisson asymptotic variance
estimator}

The asymptotic variance estimator in (\ref{eq:asymptoticvariance:consistentestimator})
can be analyzed using somewhat standard convergence theorems. For
the CLT in particular, it is helpful to view the Markov chain on the
extended space introduced in Section~\ref{subsec:estimator-pi-h-g}.
\begin{prop}
\label{prop:epave-as}Under Assumption~\ref{assu:tau-moment-kappa},
let $X$ be a Markov chain with Markov kernel $P$, and $h\in L^{m}(\pi)$
with $m>2\kappa/(\kappa-1)$. For $\pi$-almost all $X_{0}$, the
CLT holds for $h$ and for $\pi$-almost all $y$, $v(P,h)=-v(\pi,h)+2\pi(h_{0}\cdot g_{y})$.
The estimator (\ref{eq:asymptoticvariance:consistentestimator}) with
$G=G_{y}$, satisfies $\hat{v}_E(P,h)\to_{{\rm a.s.}}v(P,h)$ as $t\to\infty$.
\end{prop}

\begin{proof}
We have
\[
v^{{\rm MC}}(h)=\frac{1}{t}\sum_{s=0}^{t-1}h(X_{i})^{2}-\left\{ \frac{1}{t}\sum_{s=0}^{t-1}h(X_{i})\right\} ^{2},
\]
and these terms converge almost surely to $\pi(h^{2})$ and $\pi(h)^{2}$,
respectively, as $t\to\infty$ by the Markov chain law of large numbers
\citep[see, e.g.,][Theorem~5.2.9]{douc2018MarkovChains} and continuous
mapping, and hence $v^{{\rm MC}}(h)\to_{{\rm a.s.}}\pi(h^{2})-\pi(h)^{2}={\rm var}_{\pi}(h)$.
Now consider
\[
\frac{1}{t}\sum_{s=0}^{t-1}\left\{ h(X_{s})-\pi^{{\rm MC}}(h)\right\} G_{y}(X_{s})=\frac{1}{t}\sum_{s=0}^{t-1}h(X_{s})G_{y}(X_{s})-\frac{1}{t}\sum_{s=0}^{t-1}\pi^{{\rm MC}}(h)G_{y}(X_{s}).
\]
The assumptions guarantee by Theorem~\ref{thm:clt-kappa} that for
$\pi$-almost all $y$, $h\cdot g_{y}\in L^{1}(\pi)$ and since $\mathbb{E}\left[G_{y}(x)\right]=g_{y}(x)$
for $\pi$-almost all $x$ by Theorem~\ref{thm:fishy-estimator},
the first term on the right-hand side converges almost surely to $\pi(h\cdot g_{y})$
by the Markov chain law of large numbers, while similarly for the
second term we have $\pi^{{\rm MC}}(h)\to_{{\rm a.s.}}\pi(h)$ and
$\frac{1}{t}\sum_{s=0}^{t-1}G_{y}(X_{s})\to_{{\rm a.s.}}\pi(g_{y})$,
so the second term converges almost surely to $\pi(h)\pi(g_{y})$.
Hence, the left-hand side converges almost surely to $\pi(h_{0}\cdot g_{y})$,
and we conclude.
\end{proof}

\begin{thm}
\label{thm:epave-clt}Under Assumption~\ref{assu:tau-moment-kappa},
let $X$ be a Markov chain with Markov kernel $P$, and $h\in L^{m}(\pi)$
with $m>4\kappa/(\kappa-3)$. For $\pi$-almost all $X_{0}$, the
estimator (\ref{eq:asymptoticvariance:consistentestimator}) with
$G=G_{y}$ satisfies a $\sqrt{t}$-CLT for $\pi$-almost all $y$.
\end{thm}

\begin{proof}
Recall the notation of Appendix~\ref{subsec:estimator-pi-h-g} 
with $\bar{g}_y$ and $\check\pi$.
Define $f=(f_{1},\ldots,f_{4})$ with $f_{1}:(x,u)\mapsto h(x)$,
$f_{2}:(x,u)\mapsto h(x)^{2}$, $f_{3}=f_{1}\cdot \bar{g}_{y}$ and $f_{4}=\bar{g}_{y}$.
We observe that \eqref{eq:asymptoticvariance:consistentestimator} can be rewritten as
\begin{align*}
\hat{v}_{E}(P,h) & =\check{\pi}_{t}(f_{1})^{2}-\check{\pi}_{t}(f_{2})+2\check{\pi}_{t}(f_{3})-2\check{\pi}_{t}(f_{1})\check{\pi}_{t}(f_{4}),
\end{align*}
where $\check{\pi}_{t}=\frac{1}{t}\sum_{i=0}^{t-1}\delta_{(X_{i},U_{i})}$
is the empirical measure of the Markov chain with transition $T$ introduced in Section~\ref{subsec:estimator-pi-h-g}.
Hence, we define
\[
Z_{t}=(\check{\pi}_{t}(f_{1}),\check{\pi}_{t}(f_{2}),\check{\pi}_{t}(f_{3}),\check{\pi}_{t}(f_{4})),
\]
and $\mu_{Z}=\check{\pi}(f)=(\pi(h),\pi(h^{2}),\pi(h\cdot g_{y}),\pi(g_{y}))$.
For each $i\in\{1,\ldots,4\}$, we need to check that $f_{i}\in L^{s}(\check{\pi})$
for appropriately large $s$. We find $f_{1}\in L^{s}(\check{\pi})$
for $s\leq m$, $f_{2}\in L^{s}(\check{\pi})$ for $s\leq m/2$, $f_{3}\in L^{s}(\check{\pi})$
for $s<(\frac{2}{m}+\frac{1}{\kappa})^{-1}$ by Lemma~\ref{lem:check-phi-1-moment}
and $f_{4}\in L^{s}(\pi)$ for $s<(\frac{1}{m}+\frac{1}{\kappa})^{-1}$
by Lemma~\ref{lem:check-gy-moment}. It follows that all of these
are in $L^{s}(\check{\pi})$ if $s<(\frac{2}{m}+\frac{1}{\kappa})^{-1}$.
The CLT then holds for all $f_{i}$ individually, i.e. there exists
$\sigma_{i}^{2}<\infty$ such that
\[
\sqrt{t}(Z_{t,i}-\check{\pi}(f_{i}))=\sqrt{t}(\check{\pi}_{t}(f_{i})-\check{\pi}(f_{i}))\to_{d}\text{Normal}(0,\sigma_{i}^{2}),
\]
if $s>2\kappa/(\kappa-1)$ by Theorem~\ref{thm:clt-kappa}, and combining
these inequalities leads to the condition $\kappa>3$ and $m>4\kappa/(\kappa-3)$.
Using the Cram\'er--Wold device, we may then deduce that $\sqrt{n}(Z_{n}-\mu_{Z})\to_{d}\text{Normal}(0,\Sigma)$,
where
\[
\Sigma_{ij}={\rm cov}_{\pi}(f_{i},f_{j})+\sum_{t=1}^{\infty}{\rm cov}_{\pi}(f_{i}(X_{0}),f_{j}(X_{t}))+{\rm cov}_{\pi}(f_{j}(X_{0}),f_{i}(X_{t}))<\infty.
\]
Taking $\phi(z)=z_{1}^{2}-z_{2}+2z_{3}-2z_{1}z_{4}$, we obtain $\phi(\mu_{Z})=v(P,h)$
and by the delta method,
\[
\sqrt{t}(\hat{v}_{t}(P,h)-v(P,h))=\sqrt{t}(\phi(Z_{t})-\phi(\mu_{Z}))\to_{d}\text{Normal}(0,\sigma^{2}),
\]
where $\sigma^{2}=\nabla\phi(\mu_{Z})^{T}\Sigma\nabla\phi(\mu_{Z})<\infty$.
\end{proof}

\section{Comparison with long-run variance estimators\label{appx:reviewlongrunvarianceestimation}}

Estimation of the long-run variance $v(P, h)$ in \eqref{eq:avarfamiliar} has been a long-standing problem in MCMC, stochastic simulation, and time-series analysis. Two common estimators are the spectral variance (SV) and batch means (BM) estimators. We employ parallel chain versions of these long-run variance estimators, which we describe with references to single-chain estimators as well. In this section we assume that an adequate portion of the trajectories has been discarded as burn-in, without referring to it in the notation.

For $k = 1, \dots, m$, let $(X_{k,t})_{t\geq 0}$ denote the $k$-the Markov chain from a $\pi$-invariant kernel $P$. Long-run variance estimators for parallel chain implementations employ ``global-centering'' estimators in \eqref{eq:avarfamiliar}. For $k = 1, \dots, m$,  let $\pi^{\text{MC},k}(h) = t^{-1} \sum_{s=0}^{t-1} h(X_{k,s})$ denote the MCMC estimator of $\pi(h)$ from the $k$th chain and let the global average be denoted by $\bar{\pi}^{\text{MC}}(h) = m^{-1} \sum_{k=1}^{m} \pi^{\text{MC},k}(h)$.  Each $\text{cov}_{\pi}(h(X_0), h(X_k))$ appearing in \eqref{eq:avarfamiliar} denotes the lag-$k$ autocovariance under stationarity. \cite{agarwal2022globally} propose to estimate the lag-$k$ autocovariance using this global mean:
\begin{align}
    \hat{\gamma}(u) = \dfrac{1}{m} \sum_{k=1}^{m}\dfrac{1}{t} \sum_{s=0}^{t-u-1} (h(X_{k,s}) - \bar{\pi}^{\text{MC}}(h))(h(X_{k,s+u}) - \bar{\pi}^{\text{MC}}(h)),
\end{align}
The SV estimator of \cite{agarwal2022globally,ande:1971, hannan:1970, priest:1981} is a truncated and weighted sum of these sample autocovariances. Let $b_t$ denote the truncation or bandwidth and let $w: \mathbb{R} \to \mathbb{R}$ be a weight function, then the SV estimator is
\begin{equation}\label{eq:avar:SV}
    v^{\text{sv}}_{b_t}(P, h):=  \sum_{u=-(t-2)}^{t-2} w \left( \frac{u}{b_t} \right)\hat{\gamma}(u)\,.
\end{equation}
\cite{ande:1971} presents a variety of weight functions, including the popular Bartlett, Tukey--Hanning, and quadratic spectral weights. 
In our experiments we employ the Tukey--Hanning weight function, and the calculations are done as in \citet{heberle2017fast}.

As an alternative, BM estimators are constructed as follows. Let $t = a_tb_t$ where $b_t$ is the size of a batch and $a_t$ is the number of batches. For each of the $a_t$ batches, define the mean of the $s$th batch in the $k$th chain as $\bar{h}_{k,s} = b_t^{-1}\sum_{i=0}^{b_t - 1} h(X_{k,sb_t +i})$, for $s = 0, \dots, a_t-1$ and $k = 1, \dots, m$. Then the parallel-chain version of the BM estimator as defined by \cite{argon2006replicated,gupta2020estimating} is
\begin{equation}\label{eq:avar:BM}
    v^{\text{BM}}_{b_t}(P,h):= \dfrac{b_t}{ma_t-1} \sum_{k=1}^{m} \sum_{s=0}^{a_t-1}\left(\bar{h}_{k,s} - \bar{\pi}^{\text{MC}}(h) \right)^2.
\end{equation}
The single-chain BM estimator dates back to \cite{schm:1982} and its asymptotic properties for MCMC were first studied in \cite{flegal2010batch}. For both SV and BM we select $b_t$
with the method of \citet{liuvatsflegal2022}, implemented in the function \texttt{batchSize} of \texttt{mcmcse} \citep{mcmcse}.

Both SV and BM estimators can exhibit significant bias, specifically negative bias for positively-correlated Markov chains.  \cite{vats:fleg:2022} proposed a jackknife-like strategy to control the bias of BM and SV estimators, yielding \textit{lugsail} versions of BM and SV. For $r > 0$, a lugsail version of the estimator is
\begin{equation}\label{eq:avar:lugsail}
	v^{\circ}_{r}(P,h):= 2 v^{\circ}_{b_t}(P,h) - v^{\circ}_{\lfloor b_t/r \rfloor}(P,h)\,,
\end{equation}
where $v^{\circ}$ represents either $v^{\text{BM}}$ or $v^{\text{SV}}$. When $r = 1$, $v^{\circ}_{r}(P,h)$ is the original BM or SV estimator. For larger values of $r$, the estimator exhibits reduced bias, as will be presented in the results summarized in this section.  The choice of $b_t$ (bandwidth for SV and batch-size for BM) directly impacts the bias; large $b_t$ implies small bias. Of course, if $b_t$ is large, $t/b_t$ (or $a_t$ in BM) is small, which yields high variance of the variance estimator. This trade-off is often  balanced by seeking mean-square-optimality.

Expressions of bias and variance for single-chain BM and SV were first obtained by \cite{damerdji:1995} under strong mixing conditions of the underlying process. These expressions were later considered by \cite{flegal2010batch}, and the conditions were weakened by \cite{vats2019multivariate,vats2018strong}. \cite{agarwal2022globally} and \cite{gupta2020estimating} present bias and variance expressions for the parallel chain versions of SV and BM, respectively. We present a synergized version of the results here. But first, the result below allows the existence of a strong invariance principle for $h$ and $P$. Let $S_{1,t} = \sum_{i=0}^{t-1} h(X_{1,t})$.

\begin{thm}[\cite{kuel:phil:1980,vats2018strong}]
\label{thm:sip}
Suppose $\mathbb{E}_{\pi}[|h(X)|^{2 + \delta}] < \infty$ for some $\delta > 0$, and let $P$ be polynomially ergodic of order $\xi > (q + 1 + \epsilon)/(1 + 2/\delta)$ for some $q \geq 1$ and $\epsilon > 0$. Then, without changing its distribution, process $\{S_{1,t}\}_{t\geq 0}$ can be redefined on a richer probability space together with a standard Brownian motion, $\{B(t)\}_{t\geq 0}$ such that   there exist $\sigma > 0$, $\lambda > 0$, and a finite random variable $D$, such that
\[
\left| S_{1,t} - t \pi(h) - \sigma B(t) \right| < D \,t^{1/2 - \lambda}\,.
\]
\end{thm}

Theorem~\ref{thm:sip} yields a strong invariance principle result under mild conditions on the process, with $\sigma^2 = v(P,h)$. The existence of a strong invariance principle is one of the sufficient conditions for obtaining the variance of SV and BM estimators.
For $q \geq 1$, define
\begin{equation}\label{eq:phiq}
  \Phi^{(q)} = \sum_{u=-\infty}^{\infty}|u|^{q} {\rm cov}_{\pi}\left(h(X_{1,0}), h(X_{1,u}) \right)\,.
\end{equation}
The term $\Phi^{(q)}$ appears in the asymptotic biases below.

\begin{thm}[Bias and variance of SV estimators \citep{agarwal2022globally,damerdji:1995,liuvatsflegal2022}]
\label{thm:mse_sv}
Suppose Theorem~\ref{thm:sip} holds with $q$ such that
\[
  k_q := \lim_{x \to 0} \dfrac{1 - w(x)}{|x|^q}  < \infty\,,
\]
and $b_t^{q+1}/t \to 0$ as $t \to \infty$. Then, with $\Phi^{(q)}$ as in \eqref{eq:phiq},
\begin{equation}
    \lim_{t\to \infty} b_t^{q} \mathbb{E} \left(v^{\text{sv}}_r(P,h) - v(P,h)  \right) = - (2 - r^q) k_{q} \Phi^{(q)}\,.
\end{equation}
Additionally, if $\mathbb{E}(D^4) < \infty$, and $\mathbb{E}_{\pi}[h(X)^4] <\infty$, then
\begin{equation}
    \lim_{t \to \infty} \dfrac{t}{b_t}{\rm var}\left( v^{\text{sv}}_r(P,h) \right) = \dfrac{2}{m} v(P,h)^2 \int_{-\infty}^{\infty} w_r(x)^2 {\rm d}x\,,
\end{equation}
where $w_r(x)$ is a known function of $w$ and $r$. 
\end{thm}

The bias witnessed in any implementation then depends on the choice of $r$. For $r = 1$ (regular SV) the bias is significantly negative, whereas higher choices of $r$ remove the negative bias. \cite{vats:fleg:2022} recommend going up to $r = 3$ to control the effects on the variance of the estimator. Although the rate on the bias depends on the weight function chosen, the rate on the variance does not depend on the weight function. Theorem~\ref{thm:mse_sv} implies that the variance of the SV estimator vanishes at $t/b_t$ rate. In the context of MCMC, the asymptotic distribution of SV estimators have not been studied yet.
The bias and variance of BM estimators are similar to SV for $q = 1$. 

\begin{thm}[Bias and variance of BM estimators \citep{flegal2010batch,gupta2020estimating,vats:fleg:2022}]
\label{thm:mse_BM}
Suppose Theorem~\ref{thm:sip} holds with $q = 1$, then
\begin{equation}
    \lim_{t\to \infty} b_t \mathbb{E} \left(v^{\text{BM}}_r(P,h) - v(P,h)  \right) = -(2 - r) \Phi^{(1)}\,.
\end{equation}
Additionally, if $\mathbb{E}(D^4) < \infty$, and $\mathbb{E}_{\pi}[h(X)^4] <\infty$, then
\begin{equation}
    \lim_{t \to \infty} a_t{\rm var}\left( v^{\text{BM}}_r(P,h) \right) = \dfrac{2}{m}\left[\dfrac{1}{r} + \dfrac{4(r-1)}{r}   \right]v(P,h)^2\,.
\end{equation}
\end{thm}
Here again, for $r = 1$ we get a significant negative bias (in the presence of positive autocorrelations), with larger values of $r$ correcting for this underestimation. In our simulations we consider $r = 1, 2, 3$.

For the BM estimators as well, the variance has an $a_t$ rate of convergence. Given Theorem~\ref{thm:mse_BM}, \cite{damerdji:1995,flegal2010batch} argue that the mean-square optimal choice of $b_t$ is $b_t \propto t^{1/3}$, where the proportionality constant may be estimated using parametric techniques discussed in \cite{liuvatsflegal2022}. Recently, \cite{chakraborty2019estimating} obtained asymptotic normality of single-chain BM estimators $(m = 1)$ for geometrically ergodic, reversible chains.
\begin{thm}[\cite{chakraborty2019estimating}]
\label{thm:norm_bm}
Suppose $P$ is geometrically ergodic and $\pi$-reversible, with $\mathbb{E}_{\pi} [h(X)^8] < \infty$. If $b_t/t^{1/3} \to \infty$, then,
\[
\sqrt{a_t} \left(v^{\text{BM}}(P,h) - v(P,h) \right) \overset{d}{\to} \text{Normal}(0, 2v(P,h)^2).
\]
\end{thm}

Thus, in general we see that SV and BM estimators exhibit a $\sqrt{a_t}$ rate of convergence, which is slower than the Monte Carlo rate of EPAVE (as $t\to\infty$) and UPAVE (as the number of independent copies go to infinity).

\section{\label{appx:AR1}Verifying the assumptions in the AR(1)
case}

Assumption~\ref{assu:tau-moment-kappa} may be verified,
using a geometric Lyapunov drift condition as described in Section \ref{subsec:interpretassumption}. The assumption holds for all $\kappa\geq 1$.
Below we directly establish this, and we describe efforts to obtain bounds on the survival probabilities of the meeting times that have explicit dependencies 
on the parameters of the AR(1) process.

A Markov chain $X$ is an ${\rm AR}(1,\phi,\sigma^{2})$ chain if
for a sequence of independent $\text{Normal}(0,1)$ random variables
$Z=(Z_{n})$, $X_{n}=\phi X_{n-1}+\sigma Z_{n}$ for all $n\geq1$.
We shall assume here that $\phi\in(0,1)$. We use the reflection-maximal
coupling in Algorithm~\ref{alg:reflmaxunivariate} to construct a
Markovian coupling of two AR(1,$\phi$,$\sigma^{2}$) chains $X$
and $Y$ started from $x_{0}$ and $y_{0}$, respectively. For convenience
in this section, we describe the reflection-maximal coupling of $\text{Normal}(x,\sigma^{2})$
and $\text{Normal}(y,\sigma^{2})$ as follows:
\begin{enumerate}
\item Set $z\leftarrow(x-y)/\sigma$.
\item Sample $W\sim\text{Normal}(0,1)$.
\item Sample $B\sim{\rm Bernoulli}(1\wedge\frac{\varphi(z+W)}{\varphi(W)})$,
where $\varphi$ is the standard Normal probability density function.
\item If $B=1$, output $(x+\sigma W,x+\sigma W)$.
\item If $B=0$, output $(x+\sigma W,y-\sigma W)$.
\end{enumerate}
\begin{algorithm}
\caption{\label{alg:max-refl-coupling-ar1}Simulating pairs of ${\rm AR}(1,\phi,\sigma^{2})$
chains with reflection-maximal coupling.}

\begin{enumerate}
\item Set $X_{0}\leftarrow x_{0}$, $Y_{0}\leftarrow y_{0}$.
\item For $n=1,2,\ldots$ sample $(X_{n},Y_{n})$ from the reflection-maximal
coupling of $\text{Normal}(\phi X_{n-1},\sigma^{2})$ and $\text{Normal}(\phi Y_{n-1},\sigma^{2})$.
\end{enumerate}
\end{algorithm}

\begin{lem}
\label{lem:ar1-equiv}For Algorithm~\ref{alg:max-refl-coupling-ar1},
the meeting time $\tau=\inf\{n\geq1:X_{n}=Y_{n}\}$ satisfies
\[
\mathbb{P}_{x_{0},y_{0}}(\tau>n)=\mathbb{E}_{d_{0}}\left[\prod_{i=0}^{n-1}G(D_{i},D_{i+1})\right],
\]
where $D=(D_{n})$ is an ${\rm AR}(1,\phi,1)$ chain with $d_{0}=(x_{0}-y_{0})/(2\sigma)$
and $G(x,x')=\left(1-\exp\left\{ -2\phi xx'\right\} \right)_{+}$.
\end{lem}

\begin{proof}
We consider the following equivalent construction of $X$ and $Y$
as in Algorithm~\ref{alg:max-refl-coupling-ar1}. Let $(W_{n})$
be a sequence of independent $\text{Normal}(0,1)$ random variables.
Set $X_{0}=x_{0}$ and define
\[
X_{n}=\phi X_{n-1}+\sigma W_{n},\qquad n\geq1,
\]
and then let $Y'_{0}=y_{0}$ and define
\[
Y'_{n}=\phi Y'_{n-1}-\sigma W_{n},\qquad n\geq1.
\]
We can then define $(B_{n})$ to be a sequence of conditionally independent
Bernoulli random variables with
\[
B_{n}\sim{\rm Bernoulli}\left(1\wedge g(X_{n-1},Y_{n-1},W_{n})\right),
\]
where $g(x,y,w)=\varphi\left(\phi\frac{(x-y)}{\sigma}+w\right)/\varphi(w)$.
$X$ is the Markov chain described by the algorithm, while $Y'$ is
the Markov chain associated with the $Y$ chain when the Bernoulli
random variables $B_{n}$ take the value $0$ for all $n$. In particular,
if $\tau=\inf\{n\geq1:B_{n}=1\}$ then $Y_{n}=Y'_{n}$ for $n\in\{0,\ldots,\tau-1\}$
and $Y_{n}=X_{n}$ for $n\geq\tau$. Hence, the meeting time $\tau$
can be determined by analyzing the $X$ and $Y'$ chains, and in particular
\[
\mathbb{P}_{x_{0},y_{0}}(\tau>n)=\mathbb{E}_{x_{0},y_{0}}\left[\prod_{i=0}^{n-1}(1-1\wedge g(X_{i},Y_{i},W_{i+1}))\right]=\mathbb{P}_{d_{0}}(\tau>n).
\]
Now define $D_{n}=(X_{n}-Y'_{n})/(2\sigma)$. We notice that this
is an ${\rm AR}(1,\phi,1)$ chain with
\[
D_{n}=\phi D_{n-1}+W_{n},
\]
and that we may re-express the conditional distribution of $B_{n}$
as
\[
B_{n}\sim{\rm Bernoulli}\left(1\wedge\frac{\varphi\left(2\phi D_{n-1}+W_{n}\right)}{\varphi(W_{n})}\right),
\]
where
\begin{align*}
\frac{\varphi\left(2\phi D_{n-1}+W_{n}\right)}{\varphi(W_{n})} & =\exp\left\{ -2\phi D_{n-1}(\phi D_{n-1}+W_{n})\right\} \\
 & =\exp\left\{ -2\phi D_{n-1}D_{n}\right\} .
\end{align*}
It follows that, with $d_{0}=(x_{0}-y_{0})/(2\sigma)$,
\begin{align*}
\mathbb{P}_{x_{0},y_{0}}(\tau>n) & =\mathbb{P}_{x_{0},y_{0}}\left(B_{1}=0,\ldots,B_{n}=0\right)\\
 & =\mathbb{P}_{d_{0}}\left(B_{1}=0,\ldots,B_{n}=0\right)\\
 & =\mathbb{E}_{d_{0}}\left[\prod_{i=1}^{n}\left(1-\exp\left\{ -2\phi D_{i-1}D_{i}\right\} \right)_{+}\right]\\
 & =\mathbb{E}_{d_{0}}\left[\prod_{i=0}^{n-1}G(D_{i},D_{i+1})\right].
\end{align*}
\end{proof}
This equivalence suggests the relevance of that expectation with respect
to an AR$(1,\phi,1)$ chain, which we bound below. The proof is delayed
to the end of this appendix as it requires several intermediate results.
\begin{prop}
\label{prop:ar1-d-bound}Let $X$ be ${\rm AR(1,\phi,1)}$ and $P$
its corresponding Markov kernel. Then
\[
\mathbb{E}_{x}\left[\prod_{i=0}^{n-1}G(X_{i},X_{i+1})\right]\leq\left(\frac{2}{\tilde{\beta}}+|x|+3\right)\left\{ \tilde{\beta}^{\frac{\log(\phi)}{\log(\tilde{\beta})+\log(\phi)}}\right\} ^{n},
\]
where $\tilde{\beta}=\beta^{\delta}$, with $\beta=(1+\phi^{2})/2$
, $b=2-\phi^{2}$, $h=1-\frac{1}{\sqrt{2}}\exp\left\{ -\frac{3\phi^{2}}{1-\phi^{2}}\right\} $,
$\delta=\frac{\log h}{\log h+\log\beta-\log b}$.
\end{prop}

Combining Lemma~\ref{lem:ar1-equiv} with Proposition~\ref{prop:ar1-d-bound},
we obtain the following.
\begin{cor}
For Algorithm~\ref{alg:max-refl-coupling-ar1}, the meeting time
$\tau=\inf\{n\geq1:X_{n}=Y_{n}\}$ satisfies
\[
\mathbb{P}_{x_{0},y_{0}}(\tau>n)=\tilde{C}(x_{0},y_{0})\bar{\beta}^{n},
\]
where with the same constants as in Proposition~\ref{prop:ar1-d-bound},
$\tilde{C}(x,y)=\frac{2}{\tilde{\beta}}+|\frac{x-y}{2\sigma}|+3$
and $\bar{\beta}=\tilde{\beta}^{\frac{\log(\phi)}{\log(\tilde{\beta})+\log(\phi)}}$,
and satisfies $\pi\otimes\pi(\tilde{C})<\infty$.
\end{cor}

We see that the dependence on $\sigma$ and $|x-y|$ is fairly mild.
On the other hand, if one calculates the dependence of $\bar{\beta}$
on $\phi$, one finds that it deteriorates quickly as $\phi\nearrow1$,
even though it remains less than $1$. In contrast, numerical experiments
suggest that the true geometric rate is in fact $\phi$, but we are
not aware of a proof technique that is able to capture such a rate.
Indeed, the calculations we have used to provide a rigorous bound
are similar to those used to provide quantitative convergence rates
for Markov chains more generally and these are often loose in practice.
\begin{lem}
\label{lem:holder-survival}Let $X$ be a Markov chain with Markov
kernel $P$. Assume there exists $V\geq1$, $(\beta,b)\in(0,1)\times[1,\infty)$
such that for some set $C\subset\mathbb{X}$,
\[
PV(x)\leq\beta V(x){\mathds{1}}_{C^{\complement}}(x)+bV(x){\mathds{1}}_{C}(x),
\]
where $C^{\complement}$ is the complement $\mathbb{X}\setminus C$.
Then for $G:\mathbb{X}\times\mathbb{X}\to[0,1]$,
\[
A_{n}=\mathbb{E}_{x}\left[\prod_{i=0}^{n-1}G(X_{i},X_{i+1})\right]\leq V(x)^{\delta}\beta^{\delta n}\leq V(x)\beta^{\delta n},
\]
where we may take $\delta=\log h/(\log h+\log\beta-\log b)\in(0,1)$
for any $(0,1)\ni h\geq\sup_{x\in C}\mathbb{E}_{x}\left[G(x,X_{1})\right]$.
\end{lem}

\begin{proof}
By H\"older's inequality and the assumptions we have for any $\delta\in(0,1)$,
\begin{align*}
P\left\{ G(x,\cdot)^{1-\delta}V(\cdot)^{\delta}\right\} (x) & \leq\left\{ PG(x,\cdot)(x)\right\} ^{1-\delta}\left\{ PV(x)\right\} ^{\delta}\\
 & \leq{\mathds{1}}_{C^{\complement}}(x)\beta^{\delta}V(x)^{\delta}+{\mathds{1}}_{C}(x)h^{1-\delta}b^{\delta}V(x)^{\delta}\\
 & =\beta^{\delta}V(x)^{\delta},
\end{align*}
where the equality is due to the specific choice of $\delta$. Now,
since $0\leq G\leq1$ and $V\geq1$, we have
\[
A_{n}\leq\mathbb{E}_{x}\left[\left\{ \prod_{i=0}^{n-1}G(X_{i},X_{i+1})^{1-\delta}\right\} V(X_{n})^{\delta}\right]=:B_{n}.
\]
It follows that
\begin{align*}
B_{n} & =\mathbb{E}_{x}\left[\left\{ \prod_{i=0}^{n-2}G(X_{i},X_{i+1})^{1-\delta}\right\} P\left\{ G(X_{n-1},\cdot)^{1-\delta}V(\cdot)^{\delta}\right\} (X_{n-1})\right]\\
 & \leq\mathbb{E}_{x}\left[\left\{ \prod_{i=0}^{n-2}G(X_{i},X_{i+1})^{1-\delta}\right\} \beta^{\delta}V(X_{n-1})^{\delta}\right]\\
 & =\beta^{\delta}B_{n-1},
\end{align*}
and hence that $A_{n}\leq\beta^{n\delta}V(x)^{\delta}.$
\end{proof}
\begin{prop}
\label{prop:ar1-drift-survival-bound}Let $X$ be ${\rm AR(1,\phi,1)}$
and $P$ its corresponding Markov kernel. Then
\[
PV(x)\leq\beta V(x){\mathds{1}}_{C^{\complement}}(x)+bV(x){\mathds{1}}_{C}(x),
\]
where $V(x)=1+(1-\phi^{2})x^{2}$, $\beta=(1+\phi^{2})/2$, $C=\left\{ x:x^{2}\leq\frac{3}{1-\phi^{2}}\right\} $,
$b=2-\phi^{2}$, and
\[
\sup_{x\in C}\mathbb{E}_{x}[G(x,X_{1})]\leq h=1-\frac{1}{\sqrt{2}}\exp\left\{ -\frac{3\phi^{2}}{1-\phi^{2}}\right\} .
\]
Hence,
\[
\mathbb{E}_{x}\left[\prod_{i=0}^{n-1}G(X_{i},X_{i+1})\right]\leq V(x)\tilde{\beta}^{n},
\]
where
\[
\delta=\frac{\log h}{\log h+\log\beta-\log b}\in(0,1),\qquad\tilde{\beta}=\beta^{\delta}.
\]
\end{prop}

\begin{proof}
Let $a=1-\phi^{2}$. We have
\begin{align*}
PV(x) & =1+a\mathbb{E}\left[\left(\phi x+W\right)^{2}\right]\\
 & =1+a\phi^{2}x^{2}+a\\
 & =\phi^{2}V(x)+1-\phi^{2}+a.
\end{align*}
Now take $\beta=(1+\phi^{2})/2$. Then we find
\[
PV(x)\leq\beta V(x){\mathds{1}}_{C^{\complement}}(x)+bV(x){\mathds{1}}_{C}(x),
\]
where
\[
C=\left\{ x:x^{2}\leq\frac{3}{1-\phi^{2}}\right\} ,
\]
and $b=1+a=2-\phi^{2}$. With $G(x,x')=\left(1-\exp\left\{ -2\phi xx'\right\} \right)_{+}$,
we find
\begin{align*}
\mathbb{E}_{x}[G(x,X_{1})] & =\int\left\{ 1-\frac{\varphi\left(2\phi x+w\right)}{\varphi(w)}\right\} _{+}\varphi(w){\rm d}w\\
 & =\int\left\{ \varphi(w)-\varphi\left(2\phi x+w\right)\right\} _{+}{\rm d}w\\
 & =\left\Vert \text{Normal}(0,1)-\text{Normal}(-2\phi x,1)\right\Vert _{{\rm TV}}\\
 & =2\Phi(\phi|x|)-1\\
 & \leq1-\frac{1}{\sqrt{2}}\exp\left\{ -\left(\phi x\right)^{2}\right\} ,
\end{align*}
and so we may take
\begin{align*}
h & =\sup\left\{ 1-\frac{1}{\sqrt{2}}\exp\left\{ -\left(\phi x\right)^{2}\right\} :x^{2}\leq\frac{3}{1-\phi^{2}}\right\} \\
 & =1-\frac{1}{\sqrt{2}}\exp\left\{ -\frac{3\phi^{2}}{1-\phi^{2}}\right\} .
\end{align*}
We conclude by Lemma~\ref{lem:holder-survival}.
\end{proof}
In combination with Lemma~\ref{lem:ar1-equiv}, we obtain the following.
\begin{cor}
For Algorithm~\ref{alg:max-refl-coupling-ar1}, the meeting time
$\tau=\inf\{n\geq1:X_{n}=Y_{n}\}$ satisfies
\[
\mathbb{P}_{x_{0},y_{0}}(\tau>n)=\tilde{C}(x_{0},y_{0})\tilde{\beta}^{n},
\]
where $\tilde{\beta}\in(0,1)$ is as in Proposition~\ref{prop:ar1-drift-survival-bound},
and $\tilde{C}(x,y)=1+(1-\phi^{2})\left(\frac{x_{0}-y_{0}}{2\sigma}\right)^{2}$
satisfies $\pi\otimes\pi(\tilde{C})<\infty$.
\end{cor}

In the above, the dependence of $\tilde{C}$ on $(x_{0}-y_{0})^{2}$
is suboptimal, and we can improve this via the following result.
\begin{lem}
\label{lem:combine-rates-x-mu}Let $X$ be a Markov chain and
\[
\mathbb{P}_{\nu}(\tau>n)=\mathbb{E}_{\nu}\left[\prod_{i=0}^{n-1}G(X_{i},X_{i+1})\right],\qquad n\in\mathbb{N},
\]
for some $G:\mathbb{X}^{2}\to[0,1]$. Then for any distribution $\mu$,
$m\in\mathbb{N}$ and $k\in\{0,\ldots,m\}$,
\[
\mathbb{P}_{x}(\tau>m)\leq\mathbb{P}_{\mu}(\tau>m-k)+2\left\Vert P^{k}(x,\cdot)-\mu\right\Vert _{{\rm TV}}.
\]
\end{lem}

\begin{proof}
Let $k\in\{0,\ldots,m\}$. With $\mu_{k}(\cdot)=P^{k}(x,\cdot)$,
we have
\begin{align*}
\mathbb{P}_{x}(\tau>m) & =\mathbb{E}_{x}\left[\prod_{i=0}^{m-1}G(X_{i},X_{i+1})\right]\\
 & \leq\mathbb{E}_{x}\left[\prod_{i=k}^{m-1}G(X_{i},X_{i+1})\right]\\
 & =\mathbb{P}_{\mu_{k}}(\tau>m-k).
\end{align*}
Now $f=x\mapsto\mathbb{P}_{x}(\tau>m-k)$ takes values in $[0,1]$,
and
\[
\mathbb{P}_{\nu}(\tau>m-k)=\int\nu({\rm d}x)\mathbb{P}_{x}(\tau>m-k)=\nu(f).
\]
Hence, by the definition of TV:
\[
\left\Vert \mu-\nu\right\Vert _{{\rm TV}}=\frac{1}{2}\sup_{f:\mathbb{R}\to[-1,1]}|\mu(f)-\nu(f)|,
\]
we conclude that
\[
\mathbb{P}_{x}(\tau>m)\leq\mathbb{P}_{\mu_{k}}(\tau>m-k)\leq\mathbb{P}_{\mu}(\tau>m-k)+2\left\Vert \mu_{k}-\mu\right\Vert _{{\rm TV}}.
\]
\end{proof}
\begin{cor}
\label{cor:combined-rate}Assume that for all $n\in\mathbb{N}$, $\left\Vert P^{n}(x,\cdot)-\mu\right\Vert _{{\rm TV}}\leq C_{1}\alpha^{n}$
and $\mathbb{P}_{\mu}(\tau>n)\leq C_{2}\beta^{n}$. Then
\[
\mathbb{P}_{x}(\tau>n)\leq\left(\frac{C_{2}}{\beta}+2C_{1}\right)\left\{ \beta^{\frac{\log(\alpha)}{\log(\beta)+\log(\alpha)}}\right\} ^{n},
\]
where $\gamma=\beta^{\frac{\log(\alpha)}{\log(\beta)+\log(\alpha)}}$
and $\mathbb{P}_{\nu}(\tau>n)$ is as in Lemma~\ref{lem:combine-rates-x-mu}.
\end{cor}

\begin{proof}
Lemma~\ref{lem:combine-rates-x-mu} provides that
\begin{align*}
\mathbb{P}_{x}(\tau>n) & \leq\mathbb{P}_{\mu}(\tau>n-k)+2\left\Vert P^{k}(x,\cdot)-\mu\right\Vert _{{\rm TV}}\\
 & =C_{2}\beta^{n-k}+2C_{1}\alpha^{k},
\end{align*}
and it remains to choose $k$ appropriately. Let
\[
k_{\star}=n\frac{\log(\beta)}{\log(\beta)+\log(\alpha)},
\]
which may not be an integer. If we take $k=\lceil k_{\star}\rceil\geq k_{\star}$,
we have $n-k\geq n-k_{\star}-1$. Hence we have
\begin{align*}
\mathbb{P}_{x}(\tau>n) & \leq C_{2}\beta^{n-k_{\star}-1}+2C_{1}\alpha^{k_{\star}}\\
 & \leq\frac{C_{2}}{\beta}\beta^{n-k_{\star}}+2C_{1}\alpha^{k_{\star}}\\
 & =\left(\frac{C_{2}}{\beta}+2C_{1}\right)\left\{ \beta^{\frac{\log(\alpha)}{\log(\beta)+\log(\alpha)}}\right\} ^{n}.
\end{align*}
\end{proof}
\begin{proof}[Proof of Proposition~\ref{prop:ar1-d-bound}]
We denote $\mu_{n}=P^{n}(x,\cdot)=\text{Normal}(\phi^{n}x,\frac{1-\phi^{2n}}{1-\phi^{2}})$
and we take $\mu$ to be the stationary distribution $\text{Normal}(0,\frac{1}{1-\phi^{2}})$.
Then we can compute
\begin{align*}
\left\Vert \mu_{k}-\mu\right\Vert _{{\rm TV}} & \leq\frac{3}{2}\phi^{2k}+\frac{\phi^{k}|x|}{2}\sqrt{1-\phi^{2}}\\
 & \leq\frac{|x|+3}{2}\phi^{k},
\end{align*}
where we have used \citet[Theorem~1.3]{devroye2018total} in the first
line. From Proposition~\ref{prop:ar1-drift-survival-bound}, we find
that
\[
\mathbb{P}_{x}(\tau>n)\leq V(x)\tilde{\beta}^{n},
\]
where $V(x)=1+(1-\phi^{2})x^{2}$ and since
\[
\mathbb{P}_{\mu}(\tau>n)=\int\mu({\rm d}x)\mathbb{P}_{x}(\tau>n)\leq\mu(V)\tilde{\beta}^{n}=2\tilde{\beta}^{n},
\]
we may deduce that $\mathbb{P}_{\mu}(\tau>n)\leq2\tilde{\beta}^{n}$.
Hence, we obtain by Corollary~\ref{cor:combined-rate},
\[
\mathbb{P}_{x}(\tau>n)\leq\left(\frac{2}{\tilde{\beta}}+|x|+3\right)\left\{ \tilde{\beta}^{\frac{\log(\phi)}{\log(\tilde{\beta})+\log(\phi)}}\right\}^{n},
\]
 and obtain the final bound.
\end{proof}

\section{Hamiltonian Monte Carlo\label{sec:hmc}}

\subsection{Basic version of HMC}

Consider a target distribution on $\mathbb{R}^d$ with density $\pi$ and gradient of the log-density $\nabla \log \pi$, that can be evaluated pointwise.
With Hamiltonian Monte Carlo \citep[e.g.][]{neal2011mcmc}, we augment the state space $\mathbb{R}^d$ with a momentum variable, also taking values in $\mathbb{R}^d$. The momentum is periodically resampled from a $\text{Normal}(0,M)$ distribution, with covariance matrix $M\in\mathbb{R}^{d\times d}$. 
Algorithm \ref{alg:leapfrog} describes one step of the leapfrog integrator, which updates a pair $(x,m)$ of position and momentum variables in a way that mimics certain Hamiltonian dynamics.  

\begin{algorithm}
 \begin{enumerate}
  \item Compute $\tilde{m} = m + \frac{\epsilon}{2}\nabla \log \pi(x)$.
  \item Compute $x' = x + \epsilon M^{-1} \tilde{m}$.
  \item Compute $m' = \tilde{m} + \frac{\epsilon}{2}\nabla \log \pi(x')$.
  \item Return $(x',m')$.
 \end{enumerate}
 \caption{Function $(x,m)\mapsto \text{leapfrog}(x,m)$, with stepsize $\epsilon>0$.\label{alg:leapfrog}}
\end{algorithm}

The transition kernel of HMC is described in Algorithm \ref{alg:hmc1transition}.
From an initial position, it draws an independent momentum from a centered multivariate Normal, 
then performs $L$ leapfrog steps, and finally the terminal position is accepted or not as the new state, with a Metropolis--Rosenbluth--Teller--Hastings acceptance step.

\begin{algorithm}
  \begin{enumerate}
   \item Sample $m_0\sim \text{Normal}(0,M)$, and set $x_0 = x$.
   \item For $l\in \{1,\ldots,L\}$, $(x_l,m_l) = \text{leapfrog}(x_{l-1},m_{l-1})$, as described in Algorithm \ref{alg:leapfrog}.
   \item With probability:
    \[1\wedge \frac{\pi(x_L)\text{Normal}(m_L;0,M)}{\pi(x_0)\text{Normal}(m_0;0,M)},\]
    return $x_L$, otherwise return $x$.
  \end{enumerate}
  \caption{Hamiltonian Monte Carlo with stepsize $\epsilon>0$, mass matrix $M\in\mathbb{R}^{d\times d}$, and number of leapfrog steps $L\in\mathbb{N}$, starting from position $x\in\mathbb{R}^d$.\label{alg:hmc1transition}}
 \end{algorithm}

The tuning parameters of HMC are the \emph{mass matrix} $M$, often chosen as an estimate of the inverse of the covariance matrix of the target distribution, the stepsize $\epsilon$ and the number $L$ of leapfrog steps per transition. We can consider stepsizes that are different for different components of the state, but this is equivalent to a change of mass matrix. Indeed, consider the case where $M$ is diagonal with entries $(m_1,\ldots,m_d)$. As noted in \citet{neal2011mcmc}, the algorithm performs as if  the mass matrix was the identity matrix and the $i$-th stepsize $\epsilon_i$ was equal to $\epsilon/\sqrt{m_i}$. Thus, we can use a scalar stepsize without loss of generality. We can also randomize both the stepsize $\epsilon$ and the number of leapfrog steps $L$, at the start of each transition.

\subsection{Relation with MALA}

If $L=1$, the algorithm simplifies (5.2 in \citet{neal2011mcmc}): from the current position $x_0$, one draws $x'\sim\text{Normal}(x+(\epsilon^2)/2 M^{-1} \nabla \log \pi(x), \epsilon^2 M^{-1})$.
This is the proposal distribution in the Metropolis-adjusted Langevin algorithm (MALA, \citet{rossky1978brownian}).
Naturally the acceptance ratio of HMC with $L=1$ corresponds exactly to the acceptance ratio in MALA.
For MALA we have the following ratio of proposal densities
\begin{equation}\label{eq:mala:proposalratio}
  \frac{\text{Normal}(x; x'+(\epsilon^2)/2 M^{-1} \nabla \log \pi(x'), \epsilon^2 M^{-1})}{\text{Normal}(x'; x+(\epsilon^2)/2 M^{-1} \nabla \log \pi(x), \epsilon^2 M^{-1})}.
\end{equation}
A calculation confirms that this corresponds 
to the ratio $\text{Normal}(m';0,M)/\text{Normal}(m;0,M)$ in HMC with $L=1$.

The MALA acceptance ratio can be computed more efficiently than by separately computing two Normal density evaluations in \eqref{eq:mala:proposalratio},
as described in Proposition 1 of \citet{titsias2024optimal}. Indeed the log-ratio can directly be computed as
\begin{equation}\label{eq:mala:proposalratio:titsias}
  h(x,x') - h(x',x),\; \text{with} \; h(z,v) = \frac{1}{2}\left(z - v - \frac{\epsilon^2}{4} M^{-1} \nabla \log \pi(z) \right)^T\nabla \log \pi(v).
\end{equation}

\subsection{Coupling of HMC}

It is difficult to obtain exact meetings between two chains that evolve according to Hamiltonian Monte Carlo in general. Indeed, if the first chain at state $x$ draws a momentum $v$ and ends up in $x'$, it is not obvious that a second chain at state $\tilde{x}$ could also arrive at state $x'$ at the next step. Section 4 of  \citet{chen2023does} addresses the question of existence of a momentum vector $\tilde{v}$ with which a leapfrog trajectory started at $\tilde{x}$ would end up in $x'$. Even when such a momentum exists, it may not be available numerically. On the other hand, 
a simple coupling that consists of using the same momentum variable to propagate the two chains can result in a contraction \citep{mangoubi2017rapid,bou2020coupling}, i.e. the distance between the two chains goes to zero as the iterations go. Algorithm \ref{alg:hmc:crncoupling} describes the common random numbers coupling of HMC.

\begin{algorithm}
  \begin{enumerate}
   \item Sample $m_0\sim \text{Normal}(0,M)$, and set $x_0 = x,\tilde{x}_0=\tilde{x}$ and $\tilde{m}_0=m_0$.
   \item For $l\in \{1,\ldots,L\}$, $(x_l,m_l) = \text{leapfrog}(x_{l-1},m_{l-1})$, 
   and $(\tilde{x}_l,\tilde{m}_l) = \text{leapfrog}(\tilde{x}_{l-1},\tilde{m}_{l-1})$, as described in Algorithm \ref{alg:leapfrog}.
    \item Draw $U\sim\text{Uniform}(0,1)$, and define
    \[\alpha(x_0,m_0) = 1\wedge \frac{\pi(x_L)\text{Normal}(m_L;0,M)}{\pi(x_0)\text{Normal}(m_0;0,M)}.\]
    \item Compute
    \begin{align*}
      x' &= x_L \mathds{1}(U<\alpha(x_0,m_0)) + x_0 \mathds{1}(U\geq \alpha(x_0,m_0)), \\
      \tilde{x}' &= \tilde{x}_L \mathds{1}(U<\alpha(\tilde{x}_0,\tilde{m}_0)) + \tilde{x}_0 \mathds{1}(U\geq \alpha(\tilde{x}_0,\tilde{m}_0)).
    \end{align*}
    \item Return $(x',\tilde{x}')$.
  \end{enumerate}
  \caption{Coupled Hamiltonian Monte Carlo with common random numbers, started from states $x$ and $\tilde{x}$.\label{alg:hmc:crncoupling}}
 \end{algorithm}

In view of this \citet{heng2019unbiased} propose the following strategy. Denote by $P_{\text{H}}$ the transition of HMC, and by $\bar{P}_{\text{H}}$ the common random number coupling of HMC. When chains are close to one another, a coupling $\bar{P}_{\text{M}}$ of an MRTH kernel $P_{\text{M}}$  with maximally coupled random walk proposals may result in an exact meeting (see Appendix~\ref{appx:couplingrh}). Thus \citet{heng2019unbiased} employ an MCMC algorithm with transition $\omega_{\text{H}} P_{\text{H}} + \omega_{\text{M}} P_{\text{M}}$, with $(\omega_{\text{H}},\omega_{\text{M}})$ probabilities summing to one, and they define its coupling
as $\omega_{\text{H}} \bar{P}_{\text{H}} + \omega_{\text{M}} \bar{P}_{\text{M}}$. To generate these chains, with probability $\omega_{\text{H}}$ one employs the HMC kernel, otherwise one employs the MRTH kernel. The intuition is that the coupled HMC kernel brings the pair of chains closer to one another, and that the coupled MRTH kernel triggers exact meetings when the chains are close enough. The strategy is also used in \citet{xu2021couplings} for multinomial HMC.

In Section \ref{subsec:logitrandom} we employ the following approach, suggested in the supplementary materials of \citet{heng2019unbiased}. We make the number of leapfrog steps $L$ random, and sample it uniformly in $\{1,\ldots,L_{\max}\}$. In the event $\{L=1\}$, HMC reverts to MALA, and we couple it with a reflection-maximal coupling of the proposal distributions, see Algorithms \ref{alg:reflmax} and \ref{alg:mala:reflmaxcoupling}. When $\{L>1\}$ we employ a common random numbers coupling of HMC (Algorithm \ref{alg:hmc:crncoupling}). The tuning parameters are $\text{L}_{\max}$, $\epsilon$, and $M$. 
To summarise, the coupling of HMC employed in our experiments of Section~\ref{subsec:logitrandom} is in Algorithm \ref{alg:proposedhmccoupling}.

\begin{algorithm}
  \begin{enumerate}
    \item Compute proposal means
    \begin{align*}
      \mu &= x + (\epsilon^2)/2 M^{-1} \nabla \log \pi(x),\\
      \tilde{\mu} &= \tilde{x} + (\epsilon^2)/2 M^{-1} \nabla \log \pi(\tilde{x}).
    \end{align*}
    \item Sample proposals $(x^\star,\tilde{x}^\star)$ from a reflection-maximal coupling of $\text{Normal}(\mu,\epsilon^2M^{-1})$ and  $\text{Normal}(\tilde{\mu},\epsilon^2M^{-1})$, see Algorithm \ref{alg:reflmax}.
    \item Draw $U\sim\text{Uniform}(0,1)$, and define
    \[\alpha(x,x^\star) = 1\wedge \frac{\pi(x^\star)}{\pi(x)}\exp(h(x,x^\star) - h(x^\star,x)),\]
    where $h$ is as in \eqref{eq:mala:proposalratio:titsias}.
    \item Compute new states with accept/reject mechanism
    \begin{align*}
      x' &= x^\star \mathds{1}(U<\alpha(x,x^\star)) + x \mathds{1}(U\geq \alpha(x,x^\star)), \\
      \tilde{x}' &= \tilde{x}^\star \mathds{1}(U<\alpha(\tilde{x},\tilde{x}^\star)) + \tilde{x} \mathds{1}(U\geq \alpha(\tilde{x},\tilde{x}^\star)).
    \end{align*}
    \item Return $(x',\tilde{x}')$.
  \end{enumerate}
  \caption{Coupled MALA with maximal coupling of proposals, started from $x$ and $\tilde{x}$.\label{alg:mala:reflmaxcoupling}}
 \end{algorithm}

\begin{algorithm}
  \begin{enumerate}
    \item Draw $L\sim\text{Uniform}(\{1,\ldots,L_{\max}\})$.
    \item If $L=1$, obtain $(x',\tilde{x}')$ from the coupled MALA transition in Algorithm \ref{alg:mala:reflmaxcoupling}.
    \item If $L>1$, obtain $(x',\tilde{x}')$ from the coupled HMC transition in Algorithm \ref{alg:hmc:crncoupling}, with $L$ leapfrog steps.
    \item Return $(x',\tilde{x}')$.
  \end{enumerate}
  \caption{Coupling strategy for HMC, started from $x$ and $\tilde{x}$.\label{alg:proposedhmccoupling}}
 \end{algorithm}

\section{MCMC for Bayesian linear regression with shrinkage prior\label{appx:highdimreg}}

We examine a Bayesian linear regression of
$n=71$ responses on $p=4088$ predictors of the riboflavin data set
\citep{buhlmann2014high}, using a continuous shrinkage prior on the coefficients \citep[e.g.][]{bhadra2019lasso}. The model, MCMC algorithm and its coupling
are taken from \citet{johndrow2020scalable,biswas2021couplingbased}.
Computational difficulties
arise from the use of the shrinkage prior, that induces multimodality and heavy tails in the posterior distribution
of the coefficients, denoted by $\beta$. The target is defined on the space of coefficients, global precision, local precisions,
and the variance of the observation noise. The state space is of dimension
$2p+2=8178$.
In this example, the meeting times have been shown to have Exponential tails
in \citet[][Proposition 6]{biswas2021couplingbased} under the assumption that the global precision $\xi$ has a compact support,
in which case Assumption \ref{assu:tau-moment-kappa}
holds for any $\kappa$. In the experiments, we do not restrict the support of $\xi$, and
we use a half-Cauchy distribution on $\xi^{-1/2}$.

We provide here a short and self-contained description of one particular
version of the Gibbs sampler and its coupling; many more algorithmic
considerations can be found in \citet{biswas2021couplingbased}.

\subsection{Model}

The context is that of linear regression, with $n$ individuals and
$p$ covariates, with $p\gg n$. The generative model is described
below, where $Y$ is the outcome, $X$ the vector of explanatory variables,
$\beta\in\mathbb{R}^{p}$ the regression coefficients, $\sigma^{2}\in\mathbb{R}_{+}$
the observation noise, $\xi$ is called the global precision and $\eta_{j}$
is the local precision associated with $\beta_{j}$ for $j\in\{1,\ldots,p\}$,
\begin{align*}
 & Y\sim\text{Normal}(X\beta,\sigma^{2}I_{n}),\\
 & \sigma^{2}\sim\text{InverseGamma}(a_{0}/2,b_{0}/2),\\
 & \xi^{-1/2}\sim\text{Cauchy}(0,1)^{+},\\
\text{for }j=1,\ldots,p\quad & \beta_{j}\sim\text{Normal}(0,\sigma^{2}/\xi\eta_{j}),\quad\eta_{j}^{-1/2}\sim t(\nu)^{+}.
\end{align*}
The distribution $t(\nu)^{+}$ refers to the Student t-distribution
with $\nu$ degrees of freedom, truncated on $(0,\infty)$, with density
$x\mapsto(1+x^{2}/\nu)^{-(\nu+1)/2}\mathds{1}(x\in(0,\infty))$ up
to a multiplicative constant. The hyper-parameters are set as $a_{0}=1,b_{0}=1,\nu=2$.
In our experiments we initialize Markov chains from the prior distribution.

\subsection{Gibbs sampler}

The main steps of the Gibbs sampler under consideration are as follows.
\begin{itemize}
\item For $j=1,\ldots,p$, sample each $\eta_{j}$ given $\beta,\xi,\sigma^{2}$
using slice sampling.
\item Given $\eta$, sample $\beta,\xi,\sigma^{2}$:
\begin{itemize}
\item $\xi$ given $\eta$ using an MRTH step,
\item $\sigma^{2}$ given $\eta,\xi$ from an Inverse Gamma distribution,
\item $\beta$ given $\eta,\xi,\sigma^{2}$ from a p-dimensional Normal
distribution.
\end{itemize}
\end{itemize}
Overall the computational complexity is of the order of $n^{2}p$
operations per iteration, therefore it can be used with large $p$
and moderate values of $n$. Details on each step can be found below.

\subsubsection{$\eta$-update}

The conditional distribution of $\eta$ given the rest has density
\[
\pi(\eta|\beta,\sigma^{2},\xi)\propto\prod_{j=1}^{p}\frac{e^{-m_{j}\eta_{j}}}{\eta_{j}^{\frac{1-\nu}{2}}(1+\nu\eta_{j})^{\frac{\nu+1}{2}}}\quad\text{ where }m_{j}=\frac{\xi\beta_{j}^{2}}{2\sigma^{2}},
\]
which we can target with the slice sampler described in Algorithm
\ref{alg:slice_sampling_1}, applied independently component-wise.

\begin{algorithm}
\begin{enumerate}
\item Sample $V\sim\Uniform(0,1)$.
\item Sample $U_{j}|\eta_{j}\sim\Uniform(0,(1+\nu\eta_{j})^{-\frac{\nu+1}{2}})$
by setting $U_{j}=V\times(1+\nu\eta_{j})^{-\frac{\nu+1}{2}}$.
\item Sample $\eta_{j}|U_{j}$ from the distribution with unnormalized density
$\eta_{j}\mapsto\eta_{j}^{s-1}e^{-m_{j}\eta_{j}}$ on $(0,T_{j})$,
with $T_{j}=(U_{j}^{-2/(1+\nu)}-1)/\nu$ and $s=(1+\nu)/2$. This
can be done by sampling $U^{*}\sim\Uniform(0,1)$ and setting
\[
\eta_{j}=\frac{1}{m_{j}}\gamma_{s}^{-1}\left(\gamma_{s}(m_{j}T_{j})U^{*}\right),
\]
where $\gamma_{s}(x):=\Gamma(s)^{-1}\int_{0}^{x}t^{s-1}e^{-t}dt\in[0,1]$
is the cdf of the $\mathrm{Gamma}(s,1)$ distribution. 
\end{enumerate}
\caption{Iteration of slice sampling targeting $\eta_{j}\protect\mapsto(\eta_{j}^{\frac{1-\nu}{2}}(1+\nu\eta_{j})^{\frac{\nu+1}{2}})^{-1}e^{-m_{j}\eta_{j}}$
on $(0,\infty)$.}
\label{alg:slice_sampling_1}
\end{algorithm}

\subsubsection{$\xi$-update}

The conditional distribution of $\xi$ given $\eta$ has density
\[
\pi(\xi|\eta\propto L(y|\xi,\eta)\pi_{\xi}(\xi),
\]
where $L(y|\xi,\eta)$ is the marginal likelihood of the observations
given $\xi$ and $\eta$, and $\pi_{\xi}$ is the prior density for
$\xi$. We sample $\xi|\eta$ using a Metropolis--Rosenbluth--Teller--Hastings
scheme. Given the current value of $\xi$, propose $\log(\xi^{*})\sim\text{Normal}(\log(\xi),\sigma_{\text{MRTH}}^{2})$,
where we set $\sigma_{\text{MRTH}}=0.8$. Then calculate the ratio
\[
q=\frac{L(y|\xi^{*},\eta)\pi_{\xi}(\xi^{*})\xi^{*}}{L(y|\xi,\eta)\pi_{\xi}(\xi)\xi},
\]
using
\[
\log(L(y|\xi,\eta))=-\frac{1}{2}\log(|M_{\xi,\eta}|)-\frac{a_{0}+n}{2}\log(b_{0}+y^{T}M_{\xi,\eta}^{-1}y).
\]
where $M_{\xi,\eta}:=I_{n}+\xi^{-1}X\,\text{Diag}(\eta^{-1})\,X^{T}$.
Set $\xi:=\xi^{*}$ with probability $\min(1,q)$, otherwise keep
$\xi$ unchanged.

\subsubsection{$\sigma^{2}$-update}

Using the same notation $M_{\xi,\eta}=I_{n}+\xi^{-1}X\,\text{Diag}(\eta^{-1})\,X^{T}$,
the conditional distribution of $\sigma^{2}$ given $\xi,\eta$ is
Inverse Gamma:
\[
\sigma^{2}|\xi,\eta\sim\text{InverseGamma}\bigg(\frac{a_{0}+n}{2},\frac{y^{T}M_{\xi,\eta}^{-1}y+b_{0}}{2}\bigg).
\]

\subsubsection{$\beta$-update}

With the notation $\Sigma=X^{T}X+\xi\text{Diag}(\eta)$, the distribution
of $\beta$ given the rest is Normal with mean $\Sigma^{-1}X^{T}y$
and covariance matrix $\sigma^{2}\Sigma^{-1}$. We can sample from
such Normals in a cost of order $n^{2}p$ using Algorithm \ref{alg:fast_mvn_bhattacharya_1},
as described in \citet{bhattacharya2016fast}.

\addtocounter{algorithm}{-1}
\begin{algorithm}
\begin{enumerate}
\item Sample $r\sim\text{Normal}(0,I_{p})$, $\delta\sim\text{Normal}(0,I_{n})$
\item Compute $u=\frac{1}{\sqrt{\xi\eta}}r$, $v=Xu+\delta$.
\item Compute $v^{*}=M^{-1}(\frac{y}{\sigma}-v)$ where $M=I_{n}+(\xi)^{-1}X\text{Diag}(\eta^{-1})X^{T}$.
\item Define $U$ as $X^{T}$ with the $j$-th row divided by $\xi\eta_{j}$.
\item Return $\beta=\sigma(u+Uv^{*}).$ \caption{Bhattacharya's algorithm with common random numbers.}
\end{enumerate}
\caption{Sampling from Normal$((X^{T}X+\xi\text{Diag}(\eta))^{-1}X^{T}y,\sigma^{2}(X^{T}X+\xi\text{Diag}(\eta))^{-1})$.}
\label{alg:fast_mvn_bhattacharya_1}
\end{algorithm}

\subsection{Coupled Gibbs sampler}

We consider only one of the variants in \citet{biswas2021couplingbased},
which is not necessarily the most efficient but achieves good performance
in the experiments of Appendix~\ref{subsec:highdim:experiments} and is simpler
than the ``two-scale'' coupling described in \citet{biswas2021couplingbased}.
We describe how to couple each update, with the first chain in state
$\eta,\xi,\sigma^{2},\beta$ and the second in state $\tilde{\eta},\tilde{\xi},\tilde{\sigma}^{2},\tilde{\beta}$.

\subsubsection{$\eta$-update}

We consider two strategies to couple the slice sampling updates of
$\eta_{j}$, as described in Algorithm \ref{alg:slice_sampling_1}.
\begin{enumerate}
\item We can use a common uniform $V$ in the first step of Algorithm \ref{alg:slice_sampling_1},
to define $U_{j}$ for the first chain and $\tilde{U}_{j}$ for the
second. Then we can sample from a maximal coupling of the distributions
of $\eta_{j}|U_{j}$ and $\tilde{\eta}_{j}|\tilde{U}_{j}$, using
Algorithm \ref{alg:maximalcoupling}. This strategy results in a non-zero
probability for the event $\{\eta_{j}=\tilde{\eta}_{j}\}$.
\item We can use a common uniform $V$ in the first step of Algorithm \ref{alg:slice_sampling_1},
and then a common uniform $U^{*}$ in the third step. This is a pure
``common random numbers'' (CRN) strategy.
\end{enumerate}
We adopt a ``switch-to-CRN'' strategy: we scan the components $j\in\{1,\ldots,p\}$,
and sample $\eta_{j},\tilde{\eta}_{j}$ using the maximal coupling
strategy above. If any component fails to meet, we switch to the CRN
strategy for the remaining components.

\subsubsection{$\xi$-update}

To update $\xi,\tilde{\xi}$, we draw the proposals in the MRTH step
using a maximal coupling as in Algorithm \ref{alg:maximalcoupling}.
We then employ a common uniform variable for the two acceptance steps.

\subsubsection{$\sigma^{2}$-update}

To sample $\sigma^{2},\tilde{\sigma}^{2}$, we employ a maximal coupling
of Inverse Gamma distributions implemented using Algorithm \ref{alg:maximalcoupling}.

\subsubsection{$\beta$-update}

We use a CRN strategy, which amounts to using the same draws $r,\delta$
in the first step of Algorithm \ref{alg:fast_mvn_bhattacharya_1}
to sample both $\beta$ and $\tilde{\beta}$.

\subsection{Experiments\label{subsec:highdim:experiments}}

\begin{figure}[t]
  \centering \begin{subfigure}[b]{0.3\columnwidth} \includegraphics[width=1\columnwidth]{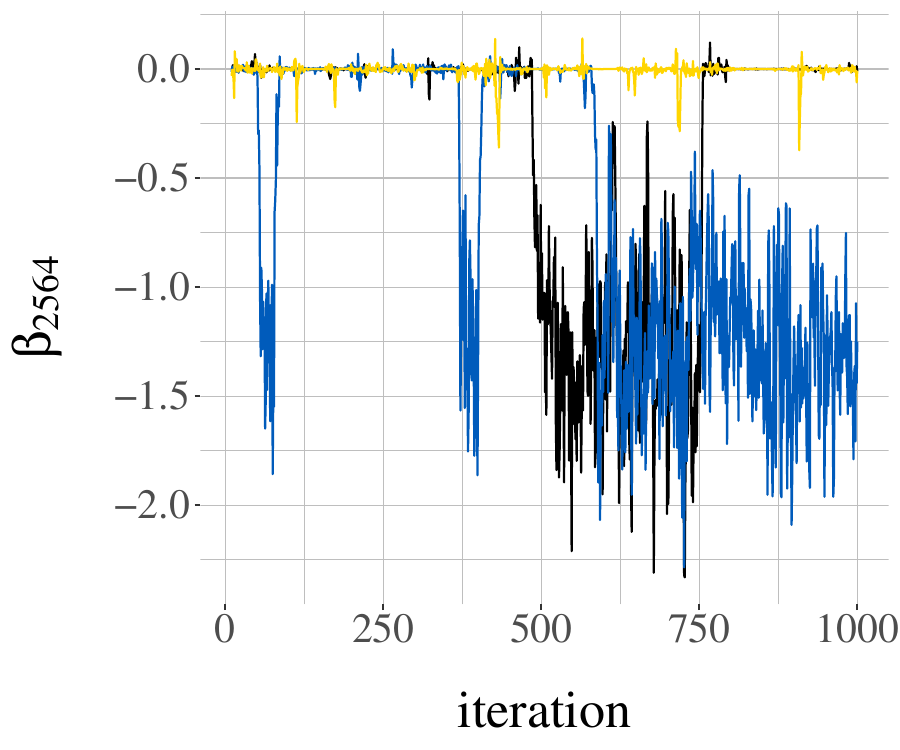}
\caption{}
\label{subfig:highdimreg:trace} \end{subfigure} \hspace*{0.1cm}
\begin{subfigure}[b]{0.3\columnwidth} \includegraphics[width=1\columnwidth]{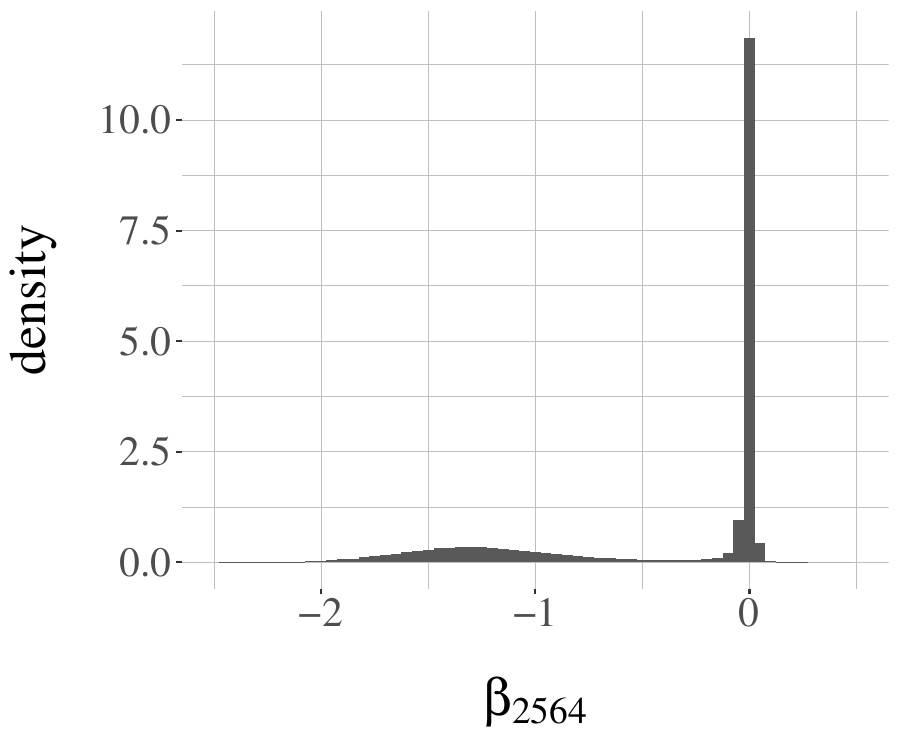}
\caption{}
\label{subfig:highdimreg:hist} \end{subfigure} \hspace*{0.1cm}
\begin{subfigure}[b]{0.3\columnwidth} \includegraphics[width=1\columnwidth]{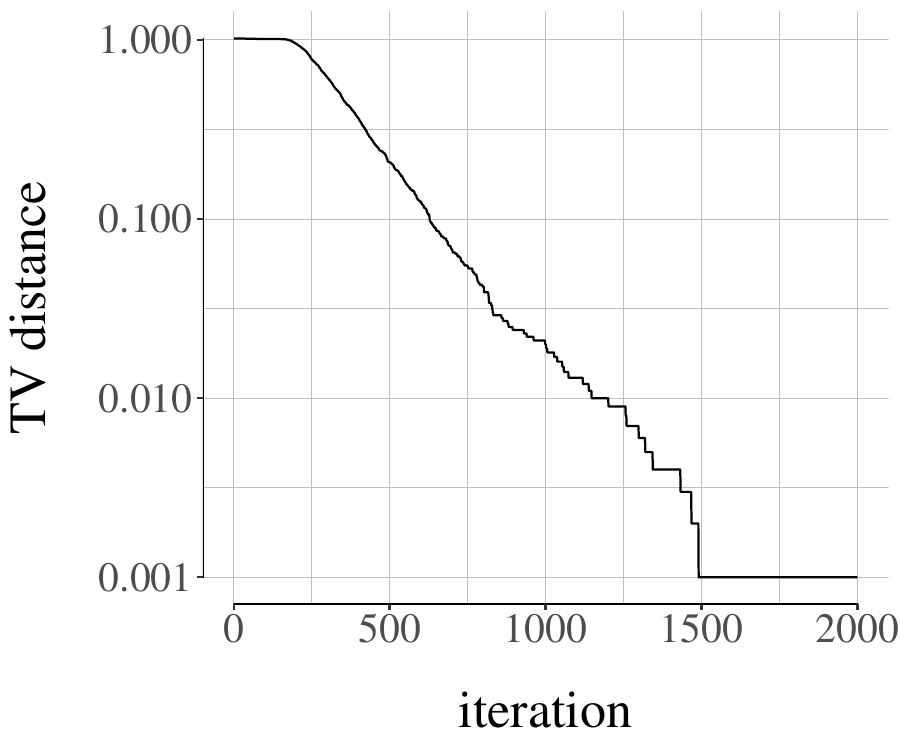}
\caption{}
\label{subfig:highdimreg:tvbounds} \end{subfigure} \caption{Gibbs sampler for linear
regression with shrinkage prior. Left: trace of the component $\beta_{2564}$ of three
independent chains. Middle: histogram
of $\beta_{2564}$, obtained from long MCMC runs. Right: upper
bounds on $|\pi_0 P^t - \pi|_{\text{TV}}$. \label{fig:highdimreg}}
\end{figure}

Based on preliminary runs, we choose the test function $\test:x\mapsto\beta_{2564}$,
which is a coordinate of the regression coefficients with a clearly
bimodal marginal posterior distribution. Figure \ref{subfig:highdimreg:trace}
shows three independent traces of $\beta_{2564}$ over 1000
iterations of the chain. Figure \ref{subfig:highdimreg:hist} presents a histogram
of $\beta_{2564}$, obtained from 10 independent chains run for 50,000
iterations each and discarding 2000 iterations as burn-in. Figure
\ref{subfig:highdimreg:tvbounds} shows upper bounds on $|\pi_{0}P^t-\pi|_{\text{TV}}$
obtained with the method of \citet{biswas2019estimating}, see Appendix~\ref{appx:tvupperbounds}, using a
lag $L=1000$ and $10^{3}$ independent meeting times. From this we
choose $k=L=1000$ and $\ell=5k$.

To define $\fishytest=g_{y}$, we draw $y$ once from the prior, and
keep it fixed. We generate $M=10^{3}$ independent estimates of $v(P,\test)$,
for $R\in\{1,5,10\}$. The results are summarized in Table \ref{tab:highdimreg:unbiased}.
We again observe tangible gains in efficiency when increasing
$R$, with diminishing returns. Overall we obtain relatively precise
information about $v(P,h)$.

\begin{table}[t]
\centering 
\begin{tabular}{r|l|l|l|l|l}
\hline
R & estimate & total cost & fishy cost & variance of estimator & inefficiency\\
\hline
1 & [77 - 98] & [12308 - 12383] & [1522 - 1594] & [2.2e+04 - 3.2e+04] & [2.7e+08 - 4.1e+08]\\
\hline
5 & [78 - 87] & [18469 - 18635] & [7683 - 7836] & [5.4e+03 - 6.9e+03] & [1e+08 - 1.3e+08]\\
\hline
10 & [78 - 85] & [26219 - 26444] & [15428 - 15652] & [2.6e+03 - 3.1e+03] & [6.7e+07 - 8.1e+07]\\
\hline
\end{tabular} \caption{Gibbs sampler for linear
regression with shrinkage prior: proposed estimators of $v(P,\protect\test)$.}
\label{tab:highdimreg:unbiased}
\end{table}

Here, unbiased MCMC estimators
of $\pi(h)$ have an expected cost of $5394$ and a variance of $0.020$, leading to an
inefficiency of $106$, while the
asymptotic variance $v(P,h)$ is estimated at $81$. Thus, unbiased MCMC is about 30\% less efficient than  ergodic average MCMC.
Users can then decide whether increasing the values of $k$, $L$ or $\ell$ to reduce the inefficiency of unbiased MCMC is warranted.

Finally, we compare with batch means and spectral variance estimators, as in the previous section. 
We compute BM and SV from $M=25$ independent runs, each of which involves $4$ chains. We discard the first 1,500 discarded as burn-in, and run the chains for $t=2\cdot10^5$ iterations. Figure \ref{fig:highdimreg:bmsv}
shows the resulting estimates. The BM and SV estimates have a low variability
compared to UPAVE. For example, the estimators ``SV r = 2'' with $ t=2\cdot10^5$ have a cost of $8\times 10^5$ and a variance of $4$. To match that variance, one would need to average about 700 UPAVE runs with $R=10$, which would cost about $2\times 10^7$ units of transitions, that is 23 times more than ``SV r = 2''.
Meanwhile, the bias of batch means and spectral variance is noticeable in Figure \ref{fig:highdimreg:bmsv}.

\begin{figure}[t]
  \centering 
\includegraphics[width=0.8\columnwidth]{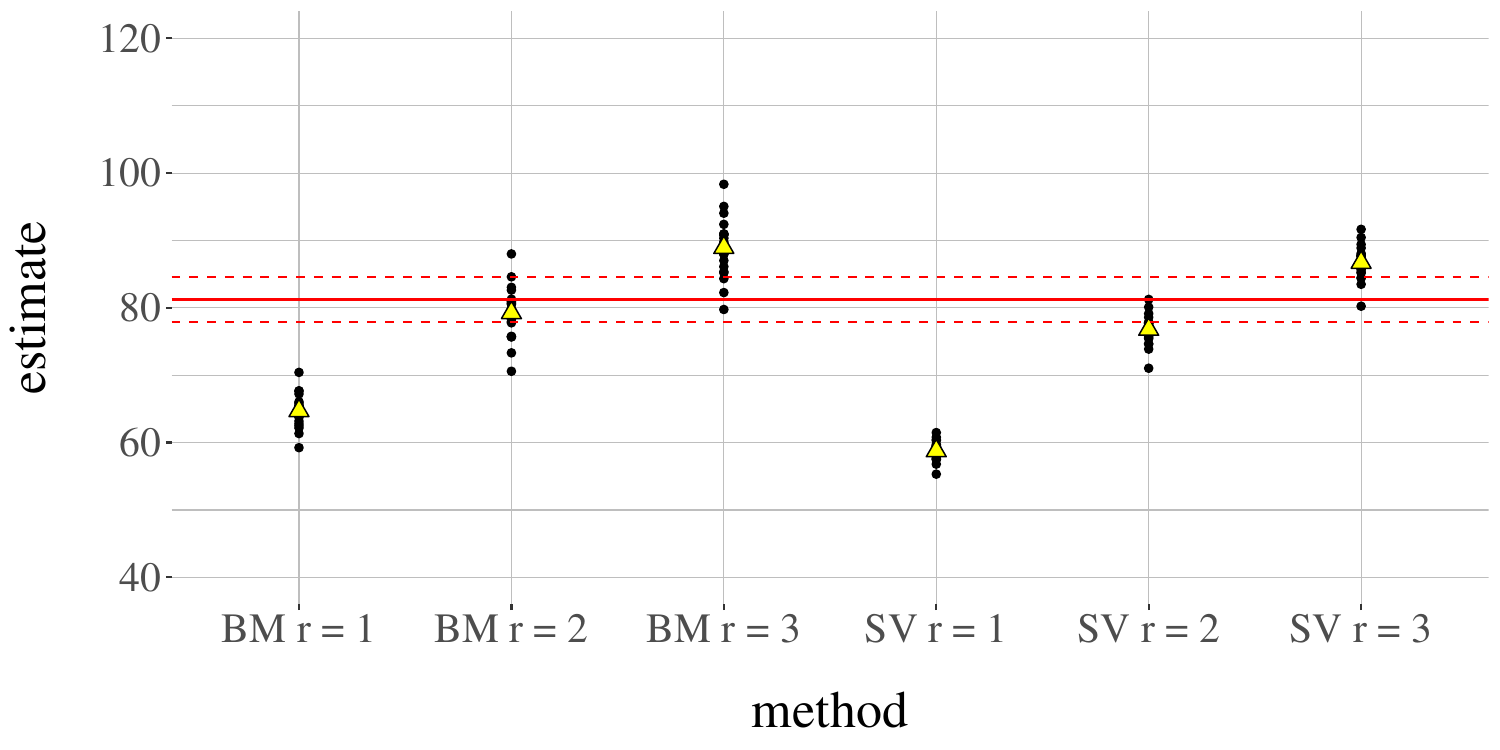}
\caption{Gibbs sampler for linear regression with shrinkage prior: batch mean and spectral variance estimators, from runs of length $2\cdot10^5$. Each dot represents an estimate of $v(P,h)$. The yellow triangles represent the means for each method. The horizontal lines represent the estimate of $v(P,h)$ and a $95\%$ confidence interval from 1000 runs of UPAVE with $R=10$. \label{fig:highdimreg:bmsv}}
\end{figure}

\section{Particle marginal Metropolis--Hastings\label{subsec:pmmh}}

We consider the state space model (SSM) example in \citet[Section 4.2,][]{middleton2020unbiased},
inspired by a model capturing the activation of neuron of rats when responding to a periodic stimulus.
The observations are counts of neuron activations over $50$ experiments. We consider $100$ data points
represented in Figure \ref{subfig:ssm:data}.
They are modelled as
$$ y_t | x_t \sim \text{Binomial}(50, \text{logistic}(x_t)),$$
where  $\text{logistic}:x\mapsto 1/(1+\exp(-x))$ and
$$ x_0 \sim \text{Normal}(0,1),\quad \text{and} \quad \forall t \geq 1\quad x_t | x_{t-1} \sim \text{Normal}(\alpha x_{t-1}, \sigma^2).$$
The prior is $\text{Uniform}(0,1)$ on $\alpha$,
and $\sigma^2$ is fixed to 1.5 here for simplicity.
The likelihood is intractable but can be estimated using a particle filter.
As in \citet{middleton2020unbiased} we use controlled SMC \citep{heng2020controlled},
and we plug the likelihood estimator in the particle marginal Metropolis--Hastings algorithms \citep[PMMH,][]{andrieu:doucet:holenstein:2010}.
We use 3 iterations of controlled SMC at each PMMH iteration. The proposal on $\alpha$ is a Normal
random walk, with a standard deviation drawn from $\text{Uniform}(0.001,0.2)$ at each iteration.
The coupling operates with a reflection-maximal coupling of the proposals,
and independent runs of SMC if the proposals differ. Verification
of Assumption \ref{assu:tau-moment-kappa} for the algorithm employed here
is an open question. Relevant comments can be found in \citet[][Section 2.3]{middleton2020unbiased}.
We initialize the chains from the prior $\text{Uniform}(0,1)$. An approximation of the posterior
distribution is shown in Figure \ref{subfig:ssm:posterior}, obtained from 100 chains of length $20,000$ and a burn-in of $1000$ steps.

\begin{figure}[t]
  \centering \begin{subfigure}[b]{0.45\columnwidth} \includegraphics[width=1\columnwidth]{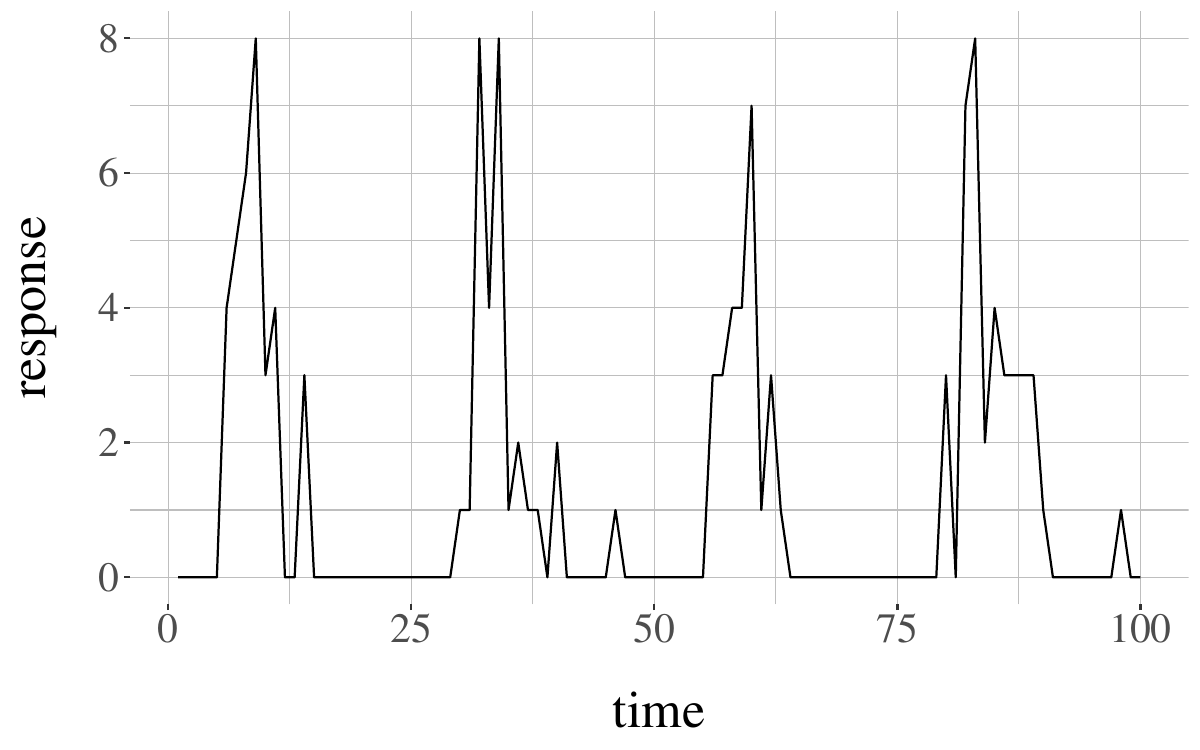}
\caption{}
\label{subfig:ssm:data} \end{subfigure} \hspace*{1cm} \begin{subfigure}[b]{0.45\columnwidth}
\includegraphics[width=1\columnwidth]{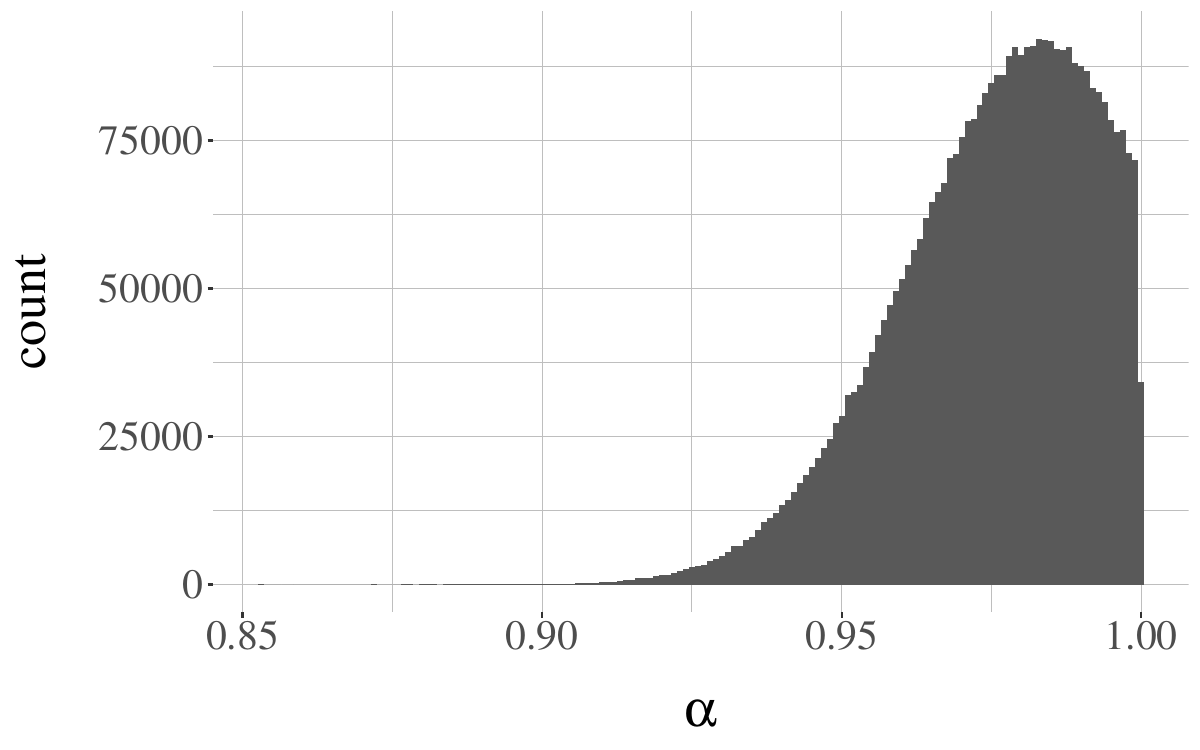}
\caption{}
\label{subfig:ssm:posterior} \end{subfigure}
\caption{Particle marginal Metropolis--Hastings: 100 observations (left) and posterior distribution on $\alpha$ (right) approximated with particle marginal Metropolis--Hastings.\label{fig:ssm:dataposterior}}
\end{figure}

Expecting PMMH to be polynomially ergodic,
we examine the tails of the distribution of the meeting times. We generate $10^5$ meeting times, either using 64 or 256 particles in each run of SMC within PMMH.
The empirical survival functions of the meeting times $\tau$, or more exactly of $\tau - L$ with a lag $L=100$, are shown in Figure \ref{fig:ssm:meetings}.
Since both axes are on logarithmic scale, a straight line indicates a polynomial decay for $\mathbb{P}(\tau > t)$. We indeed observe straight lines
on the parts of figure where $t$ is large enough.
Using linear regression we estimate the polynomial rate to be around 1 when using 64 particles (focusing on $t>200$),
and above 2 when using 256 particles (focusing on $t>100$).

\begin{figure}[t]
  \centering \begin{subfigure}[b]{0.45\columnwidth} \includegraphics[width=1\columnwidth]{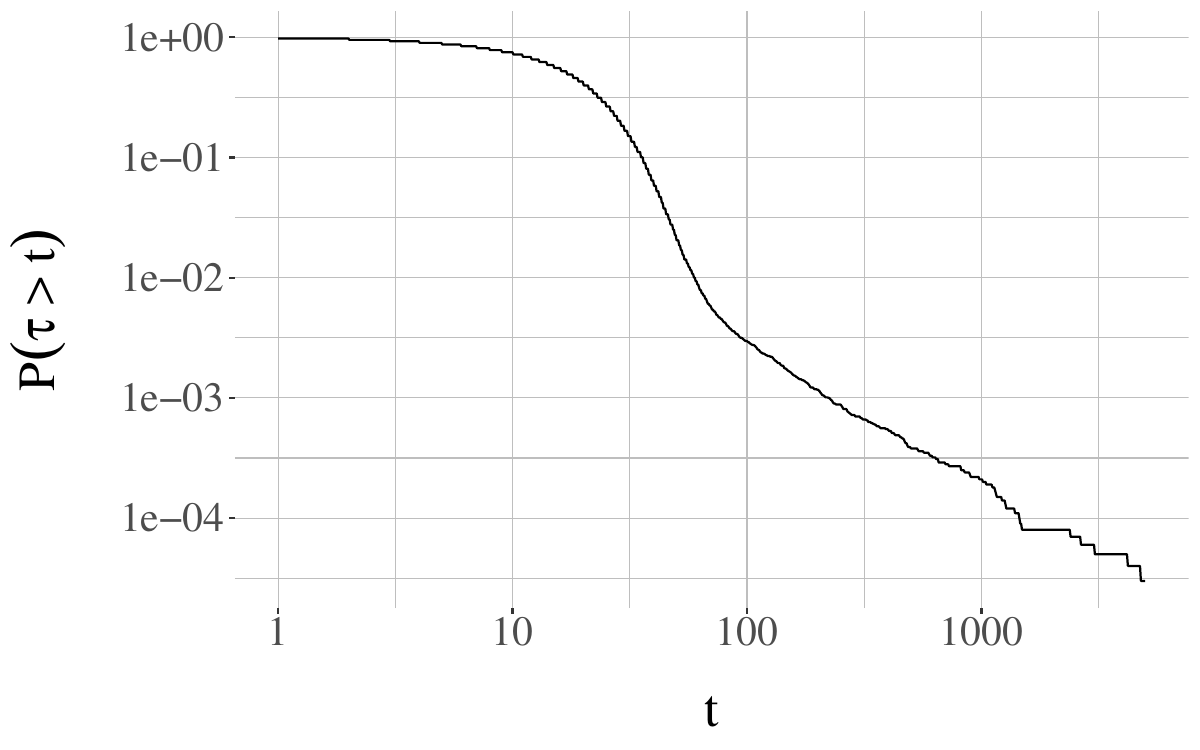}
\caption{}
\label{subfig:ssm:meetings64} \end{subfigure} \hspace*{1cm} \begin{subfigure}[b]{0.45\columnwidth}
\includegraphics[width=1\columnwidth]{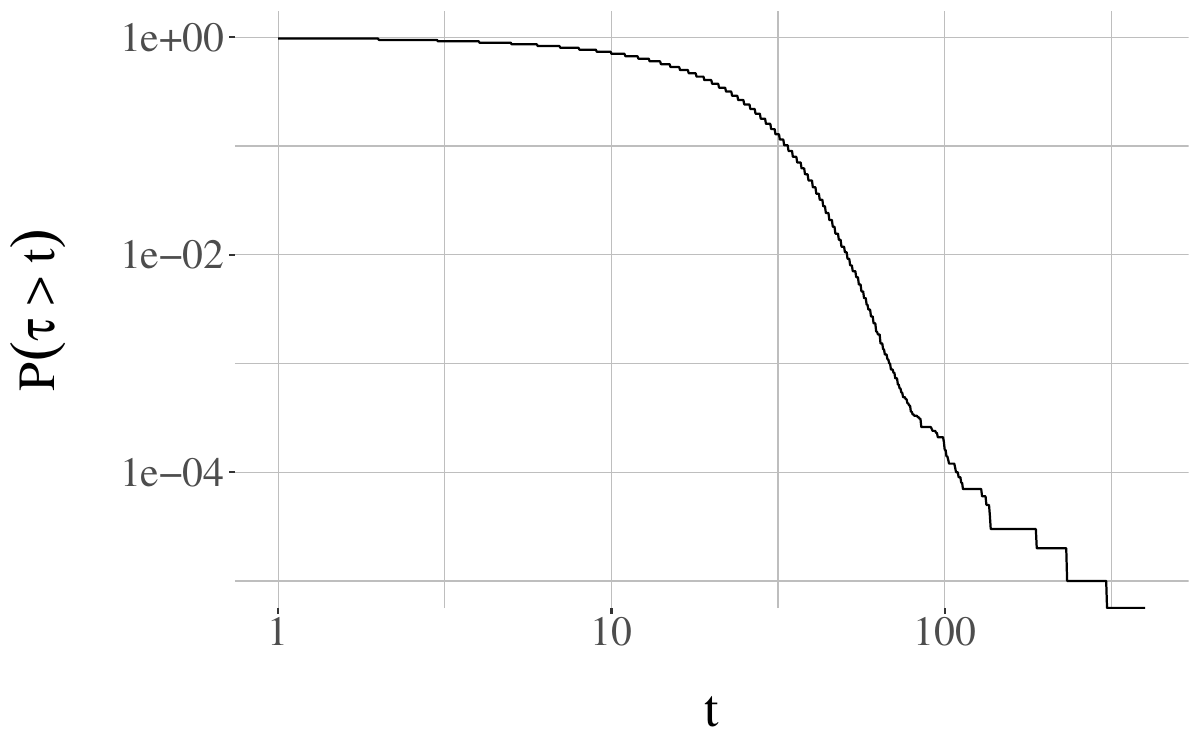}
\caption{}
\label{subfig:ssm:meetings256} \end{subfigure}
\caption{Particle marginal Metropolis--Hastings: survival function $\mathbb{P}(\tau>t)$ when using 64 particles (left) or 256 particles (right) in PMMH.
Both axes use logarithmic scale.\label{fig:ssm:meetings}}
\end{figure}

Given the heavy tails of $\tau$ when using 64 particles, we were not able to reliably estimate
the associated $v(P,\test)$.
We thus focus on the use of 256 particles, and we generate UPAVE to estimate $v(P,h)$,
$M=500$ times independently, for $h:x\mapsto x$.
We set $y=0.5$ in the definition of the fishy function estimator $\fishytestestimator_y$.
We use $k=L=500$ and $\ell=5k$ for unbiased MCMC approximations. We choose $R=50$,
the number of atoms at which $\fishytest_y$ is estimated per signed measure.
From the UPAVE runs, we can extract all the locations at which $\fishytest_y$ is estimated by $\fishytestestimator_y$,
along with the estimates.  We then represent an approximation of $\fishytest_y$ in Figure \ref{subfig:ssm:htilde:badanchor},
and a histogram of the $500$ estimates of $v(P,\test)$  in Figure \ref{subfig:ssm:UPAVE:badanchor}. We see that the relative variance
is fairly large, and notice that many estimates are negative.

\begin{figure}[t]
  \centering \begin{subfigure}[b]{0.45\columnwidth} \includegraphics[width=1\columnwidth]{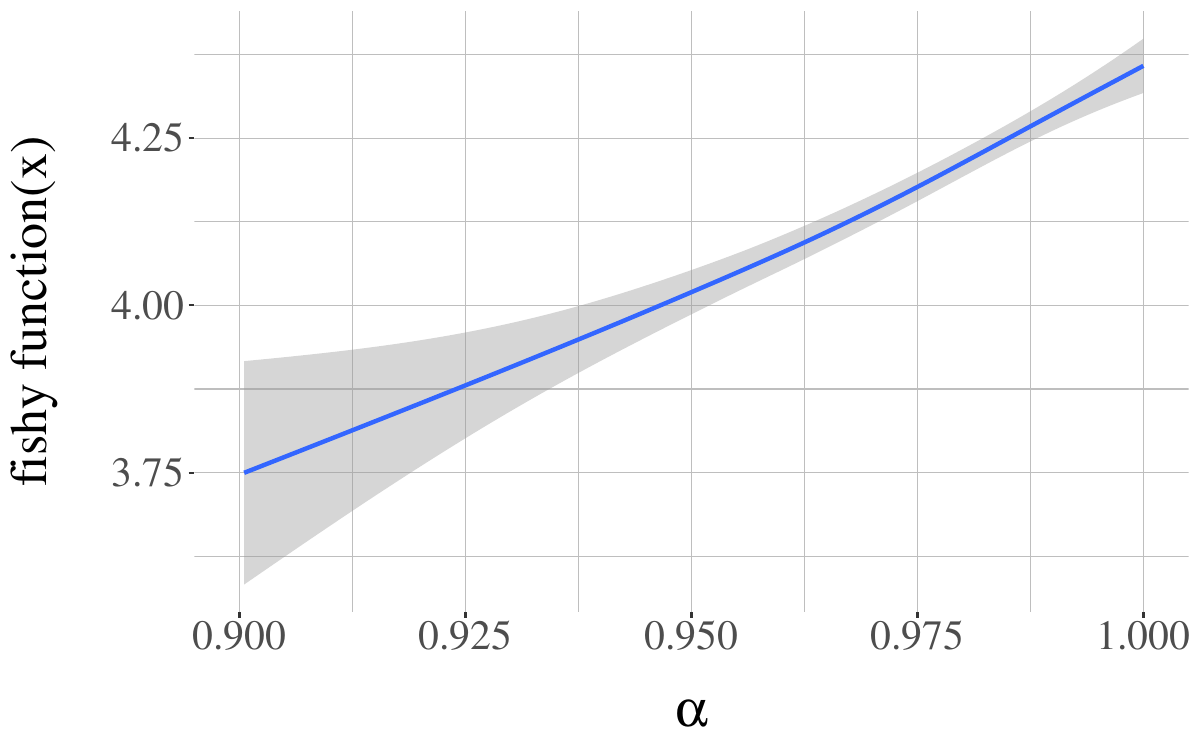}
\caption{}
\label{subfig:ssm:htilde:badanchor} \end{subfigure} \hspace*{1cm} \begin{subfigure}[b]{0.45\columnwidth}
\includegraphics[width=1\columnwidth]{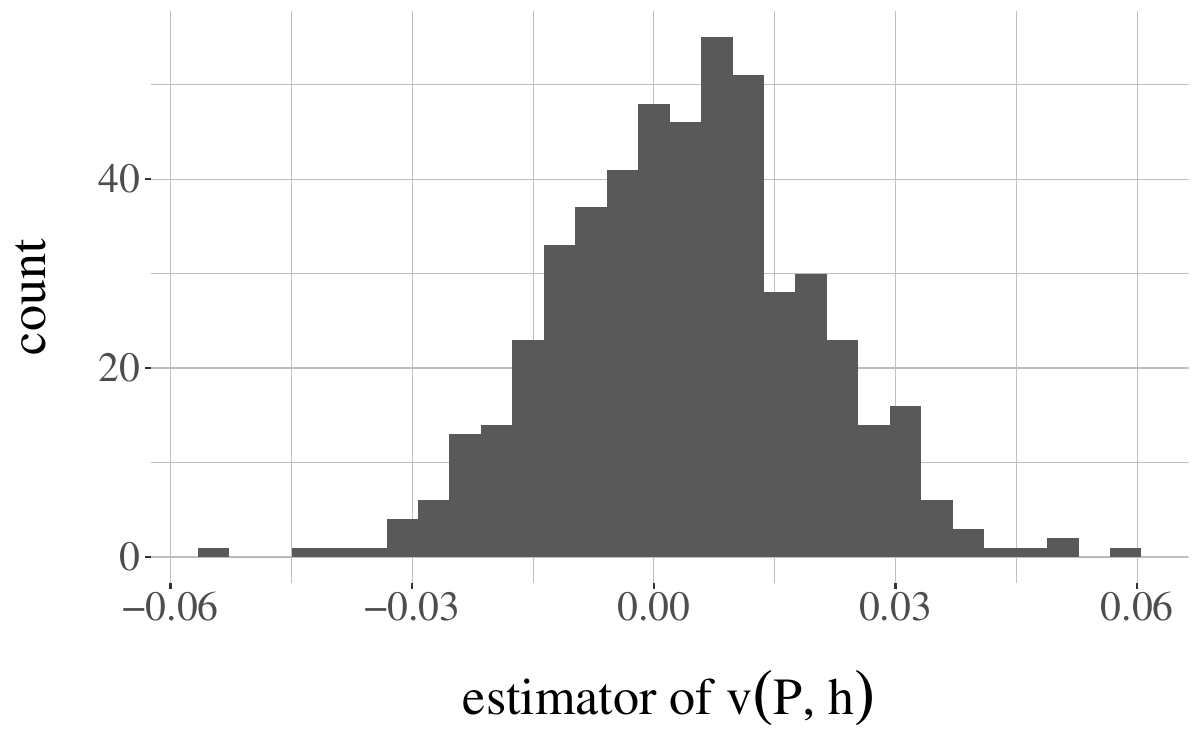}
\caption{}
\label{subfig:ssm:UPAVE:badanchor} \end{subfigure}
\caption{Particle marginal Metropolis--Hastings: using $y=0.5$, estimation of $\fishytest_y(x)$ (left) and histogram of proposed estimators of $v(P,\test)$ (right).\label{fig:ssm:badanchor}}
\end{figure}

We then change $y$ from $0.5$ to $0.975$, i.e. we place $y$ in the middle of the posterior distribution as shown in Figure \ref{subfig:ssm:posterior}, and reproduce the same plots in Figure \ref{fig:ssm:goodanchor}.
We see that the fishy function takes smaller values and its estimation is more precise.
As a result, the distribution
of $\hat{v}(P,\test)$ is considerably more concentrated.
The effect of the choice of $y$ is summarized in Table \ref{tab:ssm:comparey},
where all entries are confidence intervals based on the nonparametric bootstrap.
We observe that the choice of $y$ impacts the cost of fishy function estimation,
as well as its variance and thus the variance of UPAVE. Here this results in orders
of magnitude of difference in efficiency.

\begin{figure}[t]
  \centering \begin{subfigure}[b]{0.45\columnwidth} \includegraphics[width=1\columnwidth]{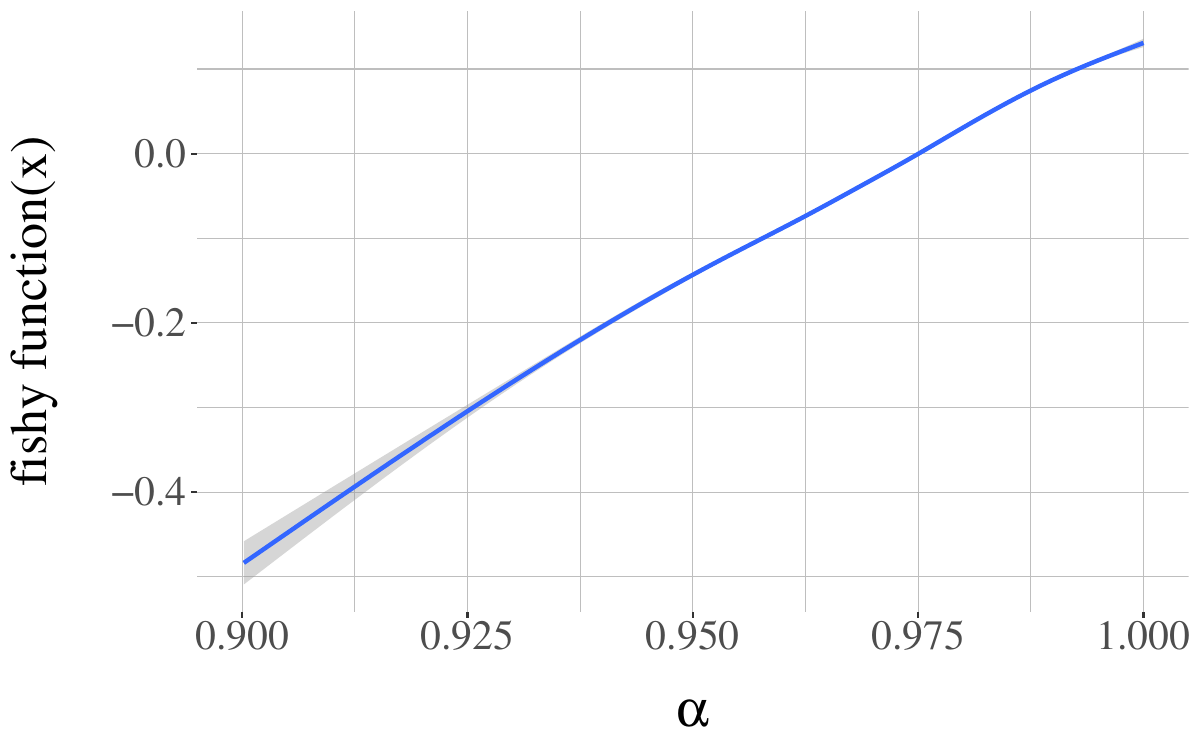}
\caption{}
\label{subfig:ssm:htilde:goodanchor} \end{subfigure} \hspace*{1cm} \begin{subfigure}[b]{0.45\columnwidth}
\includegraphics[width=1\columnwidth]{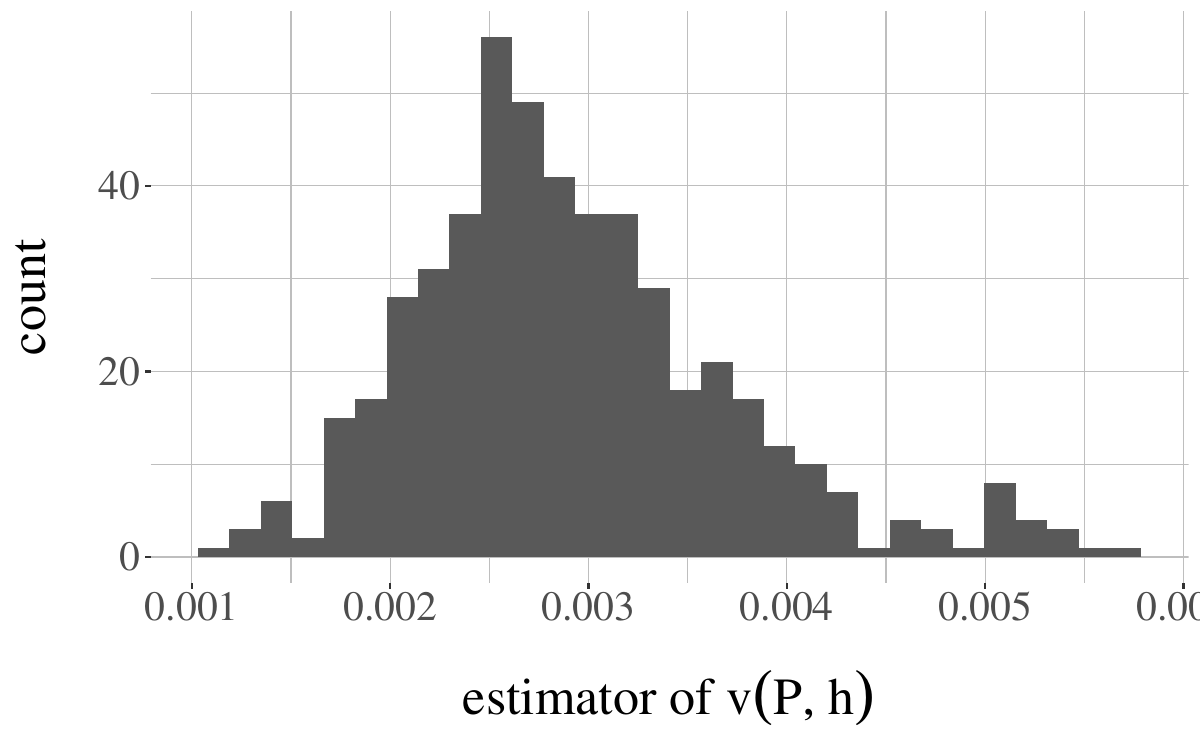}
\caption{}
\label{subfig:ssm:UPAVE:goodanchor} \end{subfigure}
\caption{Particle marginal Metropolis--Hastings: using $y=0.975$, estimation of $\fishytest_y(x)$ (left) and histogram of proposed estimators of $v(P,\test)$ (right).\label{fig:ssm:goodanchor}}
\end{figure}

\begin{table}[t]
	\centering 
\begin{tabular}{l|l|l|l|l}
\hline
y & estimate & total cost & fishy cost & variance of estimator\\
\hline
0.5 & [2.68e-03 - 5.36e-03] & [8.65e+03 - 8.7e+03] & [3.61e+03 - 3.67e+03] & [2.2e-04 - 2.8e-04]\\
\hline
0.975 & [2.85e-03 - 2.99e-03] & [6.04e+03 - 6.08e+03] & [1.01e+03 - 1.05e+03] & [5.4e-07 - 7.3e-07]\\
\hline
\end{tabular} \caption{Particle marginal Metropolis--Hastings: effect of $y$ on the proposed confidence interval for $v(P,\test)$,
	the cost of fishy function estimation, the variance of $\hat{v}(P,\test)$ and its inefficiency, based on $M=500$ independent repeats, and using $R=50$,
$k=L=500$, $\ell=5k$.}
\label{tab:ssm:comparey}
\end{table}

We can compare $v(P,\test)$, which is approximately $2.9\times 10^{-3}$, to the inefficiency associated with unbiased MCMC with $k=L=500$ and $\ell=5k$.
We compute the variance and the expected cost of unbiased MCMC estimators of $\pi(h)$ and find an inefficiency of $3.8\times 10^{-3}$. The loss
of efficiency of unbiased MCMC relative to  ergodic average MCMC is approximately $30\%$. 

Figure \ref{fig:binomialssm:bmsv} presents 25 BM and SV estimates, each based on 4 parallel chains. The chains are of length $40,000$
and the first $1000$ iterations are discarded as burn-in. Here the bias is hard to notice. Meanwhile, the variability of BM and SV is comparable with that of UPAVE: for example the SV estimator with $r = 2$ has a variance of $3.3 \cdot 10^{-8}$, comparable to that of an average of 18 UPAVE with $R=50$, which would amount to a cost of 110,000 transitions, which is less than the cost of producing the SV estimator with 4 chains of length $40,000$.

\begin{figure}[t]
  \centering 
\includegraphics[width=.8\columnwidth]{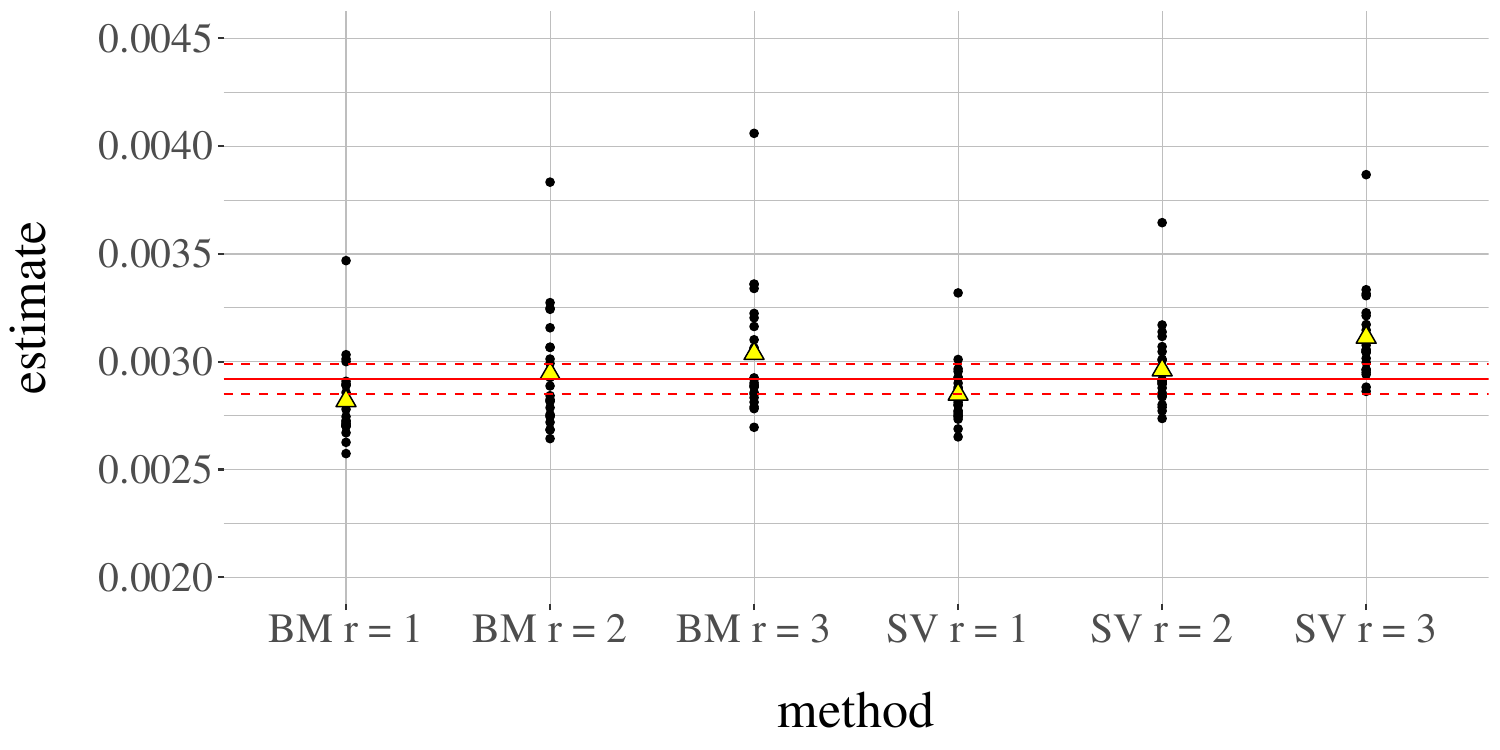}
\caption{Particle marginal Metropolis--Hastings: batch mean and spectral variance estimators, from runs of length $4\cdot10^4$. Each dot represents an estimate of $v(P,h)$. The yellow triangles represent the means for each method. The horizontal lines represent the estimate of $v(P,h)$ and a $95\%$ confidence interval obtained from 500 runs of UPAVE with $R=50$. \label{fig:binomialssm:bmsv}}
\end{figure}

\section{More experiments in the Cauchy location model\label{appx:twocompetingmcmc}}

\begin{figure}[t]
  \centering \begin{subfigure}[b]{0.45\columnwidth} \includegraphics[width=1\columnwidth]{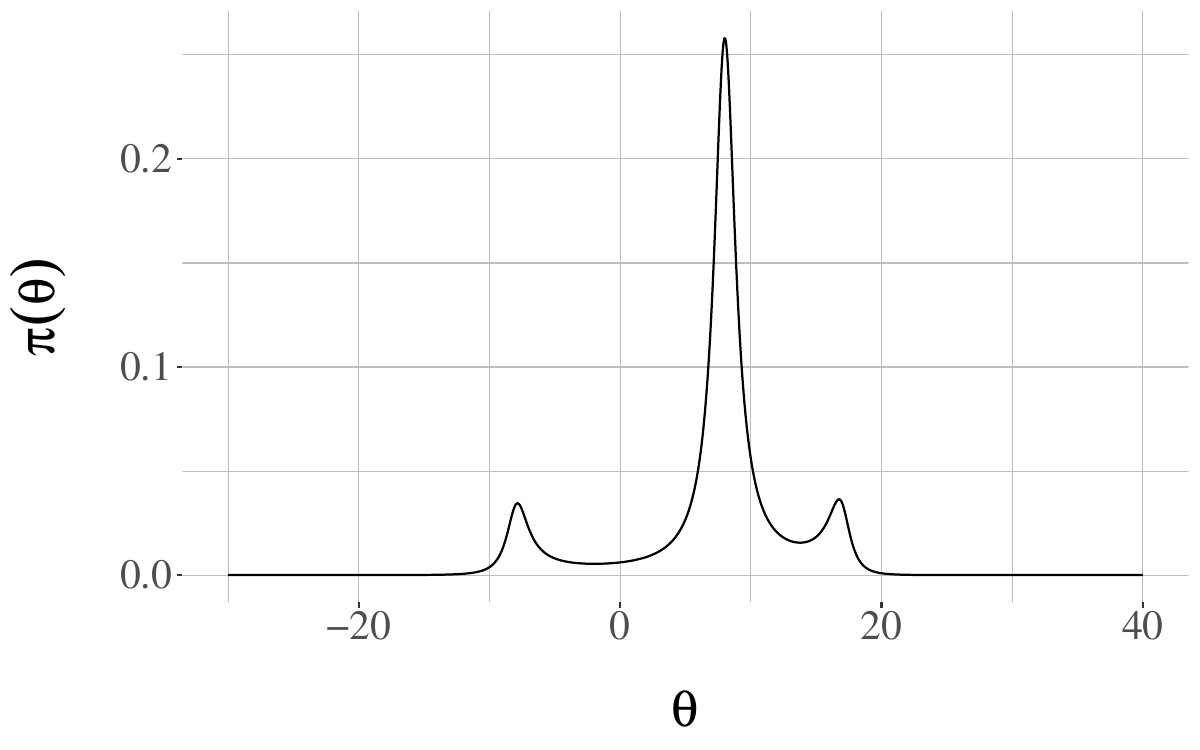}
\caption{}
\label{subfig:cauchynormal:target} \end{subfigure} \hspace*{1cm}
\begin{subfigure}[b]{0.45\columnwidth} \includegraphics[width=1\columnwidth]{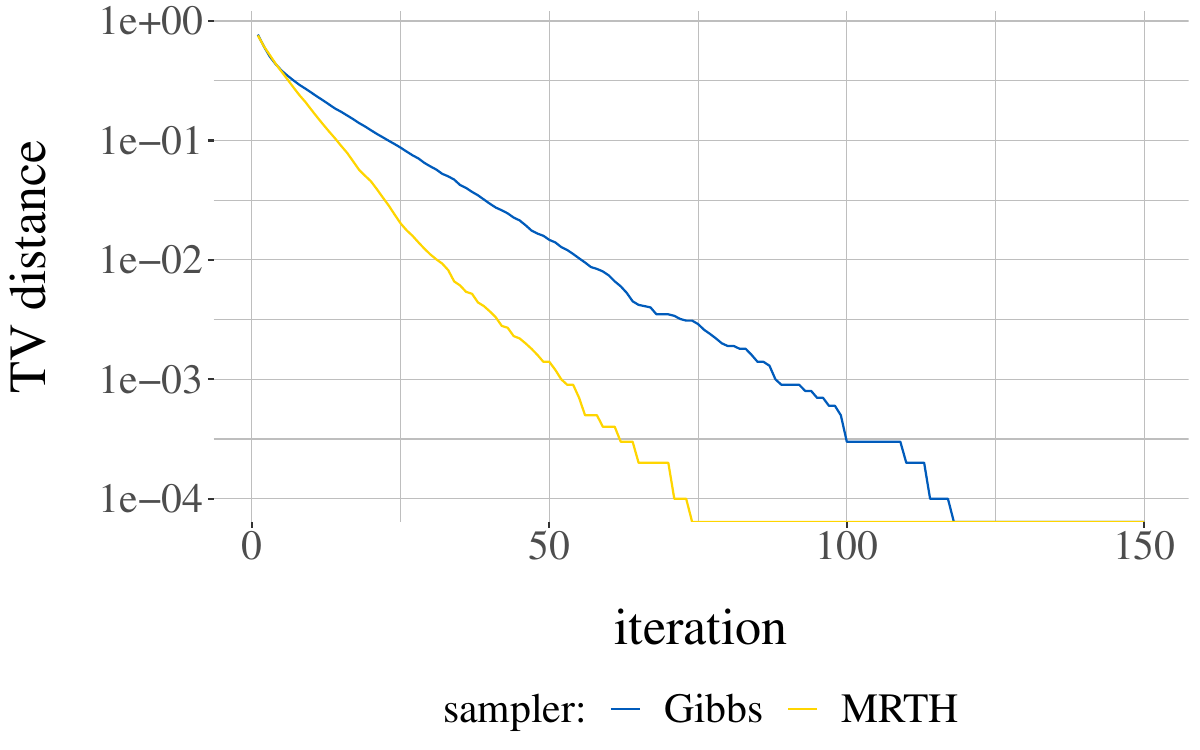}
\caption{}
\label{subfig:cauchynormal:tvupperbounds} \end{subfigure}
\caption{Cauchy-Normal example: target density (left) and upper bounds on $|\pi_0 P^t - \pi|_{\text{TV}}$
for two algorithms (right).
\label{fig:cauchynormal:targetandtvupperbounds}}
\end{figure}

We revisit the Cauchy location model and the two MCMC algorithms defined in Section \ref{subsec:fishyillustration}.
Figure \ref{subfig:cauchynormal:target}
shows the target density.
The initial distribution, for both chains, is set
to $\pi_{0}=\text{Normal}(0,1)$, for which one can verify that ${\rm d}\pi_0 / {\rm d}\pi$ is bounded.
Figure \ref{subfig:cauchynormal:tvupperbounds}
provides upper bounds on $|\pi_0 P^t - \pi|_{\text{TV}}$
for the two algorithms, using the method of \citet{biswas2019estimating}.

From Figure  \ref{subfig:cauchynormal:tvupperbounds} we select $k=100$, $L=100$ for the Gibbs sampler,
and $k=75$, $L=75$ for MRTH. In both cases we use $\ell=5k$. We generate UPAVE estimators
using different values of $R$, for both MCMC algorithms. We first use uniform selection probabilities,
$\xi = 1/N$. The results of $M=10^3$ independent runs are shown in Tables
\ref{tab:cauchynormal:gibbs} and \ref{tab:cauchynormal:mrth}.
Each entry shows a $95\%$ confidence interval obtained with the nonparametric bootstrap
from the $M$ independent replications.
The columns correspond to:
1) $R$: the number of atoms in each signed measure $\hat{\pi}$ at which fishy function estimators are generated,
2) estimate: overall estimate of $v(P,h)$, obtained by averaging $M=10^3$ independent estimates,
3) total cost of each proposed estimate, in units of ``Markov transitions''
4) fishy cost: subcost associated with the fishy function estimates (increases with $R$), 5) empirical variance of the proposed estimators (decreases with $R$), and 6) inefficiency: product of variance and total cost (smaller is better).
We observe that it is worth increasing $R$ up to the point where the fishy cost accounts for a significant portion of the total cost.
From these tables we can confidently conclude that MRTH leads to a smaller asymptotic variance than the Gibbs sampler.
Combined with an implementation-specific measure of the wall-clock time per iteration this can lead to a practical ranking of these two algorithms.

\begin{table}[h]
  \centering 
\begin{tabular}{r|l|l|l|l|l}
\hline
R & estimate & total cost & fishy cost & variance of estimator & inefficiency\\
\hline
1 & [736 - 992] & [1049 - 1054] & [32 - 36] & [3e+06 - 6.4e+06] & [3.1e+09 - 6.7e+09]\\
\hline
10 & [835 - 923] & [1349 - 1363] & [332 - 345] & [4.7e+05 - 5.9e+05] & [6.4e+08 - 8e+08]\\
\hline
50 & [849 - 903] & [2686 - 2713] & [1667 - 1696] & [1.7e+05 - 2.1e+05] & [4.7e+08 - 5.6e+08]\\
\hline
100 & [856 - 903] & [4379 - 4423] & [3361 - 3406] & [1.4e+05 - 1.7e+05] & [6.3e+08 - 7.4e+08]\\
\hline
\end{tabular} \caption{Cauchy-Normal example: estimators of $v(P,\protect\test)$ for the Gibbs sampler.}
\label{tab:cauchynormal:gibbs}
\end{table}

\begin{table}
  \centering 
\begin{tabular}{r|l|l|l|l|l}
\hline
R & estimate & total cost & fishy cost & variance of estimator & inefficiency\\
\hline
1 & [299 - 388] & [786 - 788] & [23 - 25] & [4e+05 - 7.3e+05] & [3.2e+08 - 5.8e+08]\\
\hline
10 & [331 - 364] & [996 - 1003] & [233 - 240] & [6.2e+04 - 7.9e+04] & [6.3e+07 - 7.8e+07]\\
\hline
50 & [333 - 351] & [1947 - 1966] & [1185 - 1203] & [1.9e+04 - 2.3e+04] & [3.8e+07 - 4.6e+07]\\
\hline
100 & [335 - 349] & [3139 - 3168] & [2376 - 2405] & [1.3e+04 - 1.6e+04] & [4.2e+07 - 5e+07]\\
\hline
\end{tabular} \caption{Cauchy-Normal example: estimators of $v(P,\protect\test)$ for the MRTH sampler.}
\label{tab:cauchynormal:mrth}
\end{table}

Using the fishy function estimates shown in Figure
\ref{fig:cauchynormal:htilde}, we fit generalized additive models
\citep{woodgam2017} with a cubic spline basis for the function
$x\mapsto\mathbb{E}[{\fishytestestimator}(x)^{2}]$ in order to approximate the
optimal selection probabilities $\xi$ in \eqref{eq:optimalselection}.
We then
run the proposed estimators of $v(P,\test)$, for both algorithms, with $R=10$
and $M=10^{3}$ independent replicates, using the approximated optimal $\xi$.  The results are shown in Table
\ref{tab:cauchynormal:comparison}.
We report the fishy cost, and we note that it is impacted
by the optimal tuning of selection probabilities: for Gibbs it increases,
while for MRTH it decreases. The variance of the estimator decreases, as expected.
Overall the inefficiency decreases by a factor of $2$ or $3$ in this example.

\begin{table}
\centering 
\begin{tabular}{l|l|l|l|l}
\hline
algorithm & selection $\xi$ & fishy cost & variance of estimator & inefficiency\\
\hline
Gibbs & uniform & [332 - 345] & [4.7e+05 - 5.9e+05] & [6.4e+08 - 8e+08]\\
\hline
Gibbs & optimal & [408 - 422] & [2.2e+05 - 2.8e+05] & [3.1e+08 - 4e+08]\\
\hline
MRTH & uniform & [233 - 240] & [6.2e+04 - 7.8e+04] & [6.2e+07 - 7.8e+07]\\
\hline
MRTH & optimal & [190 - 196] & [2.2e+04 - 2.7e+04] & [2.1e+07 - 2.6e+07]\\
\hline
\end{tabular} \caption{Cauchy-Normal example: estimators of $v(P,\protect\test)$
  for Gibbs and MRTH, using either optimal or uniform selection probabilities $\xi$.}
\label{tab:cauchynormal:comparison}
\end{table}

\end{document}